\documentclass[12pt,a4paper]{article}
\usepackage[a4paper,margin=1in]{geometry}
\usepackage{amsmath,amssymb,amsthm,mathtools,bm}
\usepackage{graphicx}
\usepackage[numbers,sort&compress]{natbib}
\usepackage{hyperref}
\usepackage{xcolor}
\newcommand{\dd}{\,d}
\newcommand{\ii}{\mathrm{i}}
\newcommand{\e}{\mathrm{e}}
\newcommand{\calM}{\mathcal{M}}

\newcommand{\Ker}{\operatorname{Ker}}

\newtheorem{theorem}{Theorem}[section]
\newtheorem{proposition}[theorem]{Proposition}
\newtheorem{remark}[theorem]{Remark}

\numberwithin{equation}{section}

\usepackage{authblk}
\usepackage{tikz}
\usetikzlibrary{calc,arrows.meta}

\begin{document}

\title{Perturbation Theory for Time-Dependent Point Interactions with Discrete Spectrum}

\author[1]{Kaya G\"{u}ven Akba\c{s}}
\author[2]{Fatih Erman}
\author[3]{O. Teoman Turgut}
\affil[1, 3]{Department of Physics, Bo\u{g}azi\c{c}i University, Bebek, 34342, \.{I}stanbul, Turkey}
\affil[2]{Department of Mathematics, \.{I}zmir Institute of Technology, Urla, 35430, \.{I}zmir, Turkey}
\affil[1]{haci.akbas1989@gmail.com}
\affil[2]{fatih.erman@gmail.com}
\affil[3]{turgutte@boun.edu.tr}

\date{}
\maketitle

\begin{abstract}
We develop a perturbative framework for time-dependent point interactions in quantum systems with purely discrete unperturbed spectrum.  In two and three dimensions the point interaction is described through the renormalized resolvent obtained from heat-kernel regularization, while in one dimension the diagonal Green function is finite and no renormalization is required.  Time dependence is introduced either through the interaction parameter or through the motion of the support point.  In both cases we expand the non-autonomous Hamiltonian in the spectral basis of the corresponding static point-interaction problem and derive the first- and second-order pole shifts, projection corrections, and transition amplitudes.  We also describe how finite degeneracies are treated by separating the coupled point-interaction direction from the unaffected orthogonal subspace in each degenerate eigenspace.  The method is illustrated by explicit examples: point interactions with time-dependent coupling, harmonic oscillators perturbed by moving point interactions, and a particle on a sphere with a moving interaction center.  We also discuss the one-dimensional harmonic oscillator with a moving delta potential, showing that in the non-renormalized case the present formulation reduces to the standard time-dependent perturbation theory. Finally, an exactly solvable delta center moving with constant velocity on a circle is used as a discrete-spectrum example for the local moving-center expansion and its velocity-dependent dynamical corrections.
\end{abstract}

\section{Introduction}

Point interactions form a well-established class of exactly solvable singular perturbations of Schr\"odinger operators \cite{AlbeverioGesztesyHoeghKrohnHolden2005}.  In one dimension the formal expression
\begin{equation}
    H=H_0-\alpha\,\delta(\cdot,a) \label{Hamiltonian}
\end{equation}
can be interpreted directly, either through the usual matching conditions at the support point or through a finite rank-one resolvent formula.  In two and three dimensions, however, the same expression requires renormalization: the diagonal Green function diverges and the coupling constant must be replaced by a scale-dependent parameter.  This difference is important for time-dependent perturbation theory.  In the renormalized case the natural variables are the principal function and the poles of the resolvent, while in one dimension the physical coupling strength itself may be varied without introducing a subtraction scale.

Time-dependent point interactions have been studied in several settings, including non-autonomous three-dimensional point interactions \cite{DellAntonioFigariTeta1996}, moving interaction centers \cite{DellAntonioFigariTeta2000,Posilicano2007}, ionization under time-dependent driving \cite{CorreggiDellAntonioFigariMantile2005}, one-dimensional time-dependent delta interactions \cite{HmidiMantileNier2010}, and time-dependent two-dimensional renormalized point interactions \cite{CarloneCorreggiFigari2017,BorrelliCarloneTentarelli2022}.  Explicit propagators for time-dependent delta-type interactions have also been obtained in Refs.~\cite{DodonovMankoNikonov1992,ErmanGadellaUncu2020,MudraChakraborty2021}.  The purpose of the present work is complementary.  We assume that the unperturbed Hamiltonian has purely discrete spectrum and formulate a perturbation theory directly in terms of the renormalized spectral data of the static point-interaction Hamiltonian.

To keep the notation transparent, the perturbative formulas in Sections~\ref{sec:strength} and \ref{sec:moving} are first written for simple shifted poles. However, the unperturbed Hamiltonian may have degenerate eigenvalues. For this reason, Section~\ref{sec:strength} begins with a separate discussion of finite degeneracies adapted to the time-dependent problem. There we explain how a single point interaction selects one coupled direction in each degenerate eigenspace, while the orthogonal subspace remains unaffected in the fixed-center strength problem and may rotate in the moving-center problem. The sphere and isotropic-oscillator examples in Section~\ref{sec:examples} then illustrate this adapted-sector construction explicitly.

The paper is organized as follows. Section~\ref{sec:static} recalls the static
renormalized point-interaction Hamiltonian and introduces the spectral notation
used throughout the paper. Section~\ref{sec:strength} develops the
perturbative expansion when the renormalized interaction parameter is allowed
to vary in time and also contains the adapted-sector discussion needed when
finite degeneracies are present. Section~\ref{sec:moving} treats the case in
which the strength is kept fixed but the support point of the interaction
moves.
Section~\ref{sec:examples} presents explicit applications of the general
formalism, including point interactions with time-dependent coupling
constants, the two-dimensional isotropic harmonic oscillator with a moving
renormalized delta center, a particle on a sphere with a moving delta center,
and the one-dimensional harmonic oscillator with a moving delta potential,
where no renormalization is required. Section~\ref{sec:uniform-moving-delta-circle-example}
then compares the local moving-center expansion with an exactly solvable
point interaction moving uniformly on a circle.

\section{Static Renormalized Point Interaction}
\label{sec:static}

If $H_0$ represents the free Hamiltonian of a particle moving intrinsically on a compact two or three dimensional manifold $\mathcal{M}$, $H_0=-\frac{\hbar^2}{2m} \nabla^{2}_g$, where $\nabla^{2}_g$ is the Laplace-Beltrami operator with the metric structure $g$, given by
\begin{eqnarray}
    (\nabla_{g}^2\psi)(x)= \frac{1}{\sqrt{\det g}} \sum_{i,j=1}^{d} \frac{\partial}{\partial x^i} \left( \sqrt{\det g} g^{ij} \frac{\partial \psi(x)}{\partial x^j}\right)   \;,
\end{eqnarray}
given in some local coordinate system $x= (x^1, \ldots, x^d)$, and $g^{ij}$ are the components of the inverse of the metric $g$, and $d$ is the dimension of the manifold. It is well known 
\cite{Rosenberg, chavel2} that there exists a complete orthonormal system of $C^\infty$
eigenfunctions $\{\phi_k\}_{k=0}^{\infty}$ in $L^2(\mathcal{M})$ and
the spectrum consists of only its eigenvalues 
\begin{eqnarray}
H_0 \phi_k =E_k \phi_k \;,    
\end{eqnarray}
with $E_k$ tending
to infinity as $k \rightarrow \infty$ and each eigenvalue has
finite multiplicity. Some eigenvalues are repeated according to their multiplicity. The
multiplicity of the first eigenvalue $E_0=0$ is one and the
corresponding eigenfunction is constant. If $H_0= - \frac{\hbar^2}{2m} \nabla^2 +V$ with some potential $V$, we assume that there also exists a complete orthonormal 
system of eigenfunctions with the spectrum consisting solely of its eigenvalues. For both situations, we use the same notation $E_k$ and $\phi_k(x)$ for the eigenvalues and eigenfunctions. When degeneracies are present, we include additional degeneracy labels in the eigenfunctions, as in the spherical case where the states are labelled by $(l,m)$ while the eigenvalue depends only on $l$. Throughout the paper, $d\mu(x)=d\mu_g(x)=\sqrt{\det g(x)}\,dx^1\cdots dx^d$ denotes the Riemannian volume measure on $\mathcal M$; in the Euclidean examples it denotes the usual Lebesgue measure.

As is well known, a proper description of the spectrum of the perturbed Hamiltonian \(H\) requires a renormalization procedure, since divergences arise in finding the spectrum of the problem or in the construction of the resolvent, or equivalently the Green function, associated with the formal Hamiltonian above. We note, however, that alternative singularity-free descriptions of point interactions have recently been proposed in the scattering setting, notably the dynamical formulation of stationary scattering developed by Loran and Mostafazadeh, which resolves the coincidence-limit problem for multi-center delta potentials without the usual coupling-constant renormalization \cite{LoranMostafazadeh2022}.

Here we will follow the standard renormalization description of the bound state problem for the above point delta interactions \eqref{Hamiltonian}. We begin by regularizing the singular interaction \(\delta(x,a)\) with the heat kernel \(K_{\epsilon}(x,a)\), where \(\epsilon>0\) is a cutoff parameter, and then construct the corresponding Green's function \cite{AlbeverioKurasov2000},
\begin{eqnarray}
    G^{\epsilon}(x,y|E)= G_0(x,y|E) + \frac{G_0(x,a|E) G_0(a,y|E)}{\frac{1}{\alpha}-G_{0}^{\epsilon}(a,a|E)} \;,
\end{eqnarray}
Here \(G_{0}^{\epsilon}(a,a|E)\) denotes the regularized diagonal Green's function of \(H_0\), defined by \cite{ErmanTurgut2010}
\begin{eqnarray}
G_{0}^{\epsilon}(a,a|E) =    \int_{\epsilon}^{\infty}K_s(a,a) \, e^{s E} \, ds \;,
\end{eqnarray} 
for \(\mathrm{Re}(E)<0\), where \(K_s(a,a)\) is the diagonal heat kernel associated with \(H_0\). The heat kernel satisfies
\[
H_0 K_s(x,y)=\frac{\partial K_s(x,y)}{\partial s}\;,
\]
for each fixed \(y\), with \(H_0\) acting on the \(x\)-variable. The divergence of \(G(a,a|E)\) originates from the short-time asymptotics of the heat kernel \cite{DaviesHeatKernels},
\begin{eqnarray}
    K_s(a,a) \sim \frac{1}{(4 \pi s)^{D/2}} \;,
\end{eqnarray}
in dimensions \(D=2\) and \(D=3\). To eliminate this singularity, one introduces a cutoff-dependent coupling constant \(\alpha(\epsilon)\) chosen so that its \(\epsilon\)-dependence cancels 
the divergence in the diagonal Green's function. Specifically, one defines \cite{ErmanTurgut2010}
\begin{eqnarray}
    \frac{1}{\alpha(\epsilon)}:= \frac{1}{\alpha_R(M)} + \int_{\epsilon}^{\infty} K_{s}(a,a)e^{s M} ds  \;,
\end{eqnarray}
where \(M\) is an arbitrary renormalization scale. With this choice, the limit \(\epsilon\to 0^+\) exists and yields the finite renormalized Green's function
\begin{eqnarray}
    G(x,y|E)= G_0(x,y|E) + \frac{G_0(x,a|E) G_0(a,y|E)}{\frac{1}{\alpha_R(M)}+ \int_{0}^{\infty}K_s(a,a)(e^{s M}-e^{s E}) ds} \;.
\end{eqnarray}
Let us denote the expression in the denominator by 
\begin{eqnarray}
\Phi(E,M):= \frac{1}{\alpha_R(M)}+ \int_{0}^{\infty}K_s(a,a)\bigl(e^{sM}-e^{sE}\bigr)\, ds \;, 
\end{eqnarray}
which we shall refer to as the \emph{principal function}. The leading small-$s$ divergence cancels because $e^{sM}-e^{sE}=s(M-E)+O(s^2)$. The remaining arbitrariness in the formulation can be eliminated by fixing the renormalized coupling through a physical input (renormalization condition). The bound-state energies are determined by the poles of the full Green function on the real \(E\)-axis, or if there is no pole cancellation with the Green's function $G_0$, the bound state energies are equivalently determined by the real zeros of \(\Phi(E,M)\). Then, one obtains a relation among the renormalization scale \(M\), the renormalized coupling \(\alpha_R(M)\), and the bound-state energies. The dependence of \(\alpha_R\) on \(M\) is such that physical quantities, in particular the bound-state spectrum, remain independent of the arbitrary scale \(M\). We choose the renormalization scale as \(M=-\mu_0\), corresponding to a bound state lying below the ground-state energy \(E_0\) of \(H_0\). A particularly convenient choice for bound-state calculations is to impose
\begin{equation}
\frac{1}{\alpha_R(M)}=0 \qquad \text{at} \qquad M=-\mu_0 \;.    
\end{equation}
With this prescription, no additional unphysical parameter remains, and all other energy levels are fixed by the assigned binding energy. Although this choice is especially useful for determining bound states, it is not always optimal in other settings, for instance in perturbative expansions in \(\alpha_R\).
With this choice, we obtain
\begin{eqnarray}
    G(x,y|E)= G_0(x,y|E) + \frac{G_0(x,a|E) G_0(a,y|E)}{\Phi(E,\mu_0)} \;, \label{eq:renormalized-resolvent}
\end{eqnarray}
where
\begin{eqnarray}
\Phi(E,\mu_0)=  \int_{0}^{\infty}K_s(a,a)(e^{-s \mu_0}-e^{s E}) ds \;. \label{principalfunction}
\end{eqnarray}
Using an analytic continuation argument and the eigenfunction expansion of heat kernel of $H_0$ \cite{Rosenberg},
\begin{equation}
    K_t(x,y)=\sum_{n=0}^{\infty}\phi_n(x)\overline{\phi_n(y)}\,\e^{-tE_n} \;,
\end{equation}
we obtain the expansion of the Green function $G_0$, the kernel of $(H_0-E)^{-1}$ for $E\notin\sigma(H_0)$, as
\begin{equation}
    G_0(x,y|E)=\int_0^\infty K_t(x,y)\e^{tE}\dd t = \sum_{n=0}^{\infty}
    \frac{\phi_n(x)\overline{\phi_n(y)}}{E_n-E} \;.
\end{equation}
Occasionally, we use the heat-kernel representation of the Green function for clarity of exposition, although the representation may require analytic continuation outside its initial domain of convergence. 
Then, we have
\begin{eqnarray}
    \Phi(E,\mu_0)
    & = & \sum_{n=0}^{\infty}|\phi_n(a)|^2
    \left(\frac{1}{E_n+\mu_0}-\frac{1}{E_n-E}\right) \nonumber \\ 
    &= & -(E+\mu_0)\sum_{n=0}^{\infty}
    \frac{|\phi_n(a)|^2}{(E_n+\mu_0)(E_n-E)}     \label{eq:Phi-series} \;.
\end{eqnarray}
The convergence of these series follows from heat-kernel estimates; see Refs.~\cite{annals23,Frontiers25} for a detailed proof.

The poles of equation \eqref{eq:renormalized-resolvent} are the zeros of $\Phi(E,\mu_0)$, except when a pole of $G_0$ is cancelled by the numerator.  We denote the shifted eigenvalues by $E_k^*$.  If $\phi_k(a)\ne0$, then $E_k^*$ is the unique zero of $\Phi(\cdot,\mu_0)$ in the corresponding spectral gap.  If $\phi_k(a)=0$, then the eigenvalue $E_k$ is unchanged and the associated eigenfunction is the original $\phi_k$.

For a shifted simple pole $E_k^*$ we introduce the following notation for simplicity:
\begin{eqnarray}
    G_0^k(x)& := & G_0(x,a|E_k^*),\nonumber \\ 
    G_{0,E}^k(x)& := & \left.\partial_EG_0(x,a|E)\right|_{E=E_k^*}, \nonumber 
\\     \Phi_E^k & := & \left.\partial_E\Phi(E,\mu_0)\right|_{E=E_k^*},
\end{eqnarray}
and similarly for higher derivatives. Here $\partial_E$ denotes differentiation with respect to the spectral parameter. Notice that
\begin{equation}
    \Phi_E^k
    =-\sum_{n=0}^{\infty}\frac{|\phi_n(a)|^2}{(E_n-E_k^*)^2}
    =-\left.\partial_EG_0(a,a|E)\right|_{E=E_k^*}<0 \;.
\end{equation}
One can also find the bound state wave function $\psi_k(x)$ associated with each simple pole $E_{k}^*$ of the renormalized Green's function from the Riesz projection onto eigenspace associated with this pole: $\psi_k(x) \overline{\psi_k(y)}= - \frac{1}{2\pi i} \oint_{\Gamma_k} G(x,y|E) dE$, where $\Gamma_k$ is the contour around $E=E_{k}^{*}$ with counterclockwise orientation.  It follows from residue theorem that 
\begin{equation}
    \psi_k(x)=\frac{G_0^k(x)}{\sqrt{-\Phi_E^k}} \;.
    \label{eq:psi-k}
\end{equation}
The associated rank-one projection has kernel
\begin{equation}
    P_k(x,y)=\psi_k(x)\overline{\psi_k(y)}
    =\frac{G_0^k(x)\overline{G_0^k(y)}}{\left(-\Phi_E^k\right)}
    =\frac{G_0(x,a|E_k^*)G_0(a,y|E_k^*)}{\left(-\Phi_E^k\right)} \;.
    \label{eq:Pk-complex-kernel}
\end{equation}
Although the formal expression \(H_0-\alpha\delta(\cdot,a)\) does not by
itself define a Hamiltonian in two and three dimensions, the renormalized
resolvent formula gives the candidate resolvent of the corresponding
singular perturbation. The usual approach is to show that this resolvent is
the resolvent of a self-adjoint extension of the symmetric operator obtained
by removing the support point \(a\) \cite{AlbeverioKurasov2000,ErmanTurgut2010}.
In the present setting, we also use the completeness result proved in
Ref.~\cite{Frontiers25}. Namely, for the class of renormalized point
interactions considered here, the eigenfunctions \(\{\psi_k\}\) associated
with the poles of the renormalized Green function form a complete
orthonormal system. Hence one may define the operator spectrally by
\begin{equation}
    D(H)=
    \left\{
    f\in L^2(\mathcal M):
    \sum_k |E_k^*|^2 |\langle\psi_k,f\rangle|^2<\infty
    \right\},
\end{equation}
and
\begin{equation}
    Hf=
    \sum_k E_k^*
    \langle\psi_k,f\rangle\psi_k,
    \qquad f\in D(H).
\end{equation}
Since the eigenvalues \(E_k^*\) are real and \(\{\psi_k\}\) is complete, this
spectral operator is self-adjoint by the spectral theorem. Equivalently,
\begin{equation}
    H=\sum_k E_k^*P_k,
    \qquad
    P_k=|\psi_k\rangle\langle\psi_k|. \label{eq:H-static}
\end{equation}
This is the sense in which the completeness result of Ref.~\cite{Frontiers25}
allows the renormalized Hamiltonian to be represented as a self-adjoint
operator with the above spectral decomposition.
Here, in the standard Dirac notation, \(|u\rangle\langle v|\) denotes the
rank-one operator on \(L^2(\mathcal M)\) defined by
\begin{equation}
    (|u\rangle\langle v|)f
    =
    u\,\langle v,f\rangle ,
    \qquad
    \langle v,f\rangle
    =
    \int_{\mathcal M}\overline{v(x)}f(x)\,d\mu_g(x) \;.
\end{equation}

In the present paper, we use the spectral representation \eqref{eq:H-static} as the starting point for time-dependent perturbation theory.

\section{Perturbation Theory for Time-dependent Renormalized Strength}
\label{sec:strength}

We next allow the physical parameter fixed by the renormalization condition
to depend on time. In the static problem the subtraction point was chosen as
$M=-\mu_0,\qquad \mu_0>0$, together with $\frac{1}{\alpha_R(-\mu_0)}=0$.
With this choice the principal function $\Phi$ is given by \eqref{principalfunction},
and the lowest pole is fixed at \(E=-\mu_0\). Thus \(\mu_0\) is not used
below as an arbitrary renormalization scale, but as the physical
bound-state parameter specifying the renormalized point interaction.

We introduce a time-dependent renormalized strength by replacing
\begin{equation}
    \mu_0 \longrightarrow \mu(t)=\mu_0+\eta(t),
    \qquad |\eta(t)|\ll1,
    \qquad \eta(0)=0 \;.
    \label{eq:mu-time}
\end{equation}
The parameter $\mu(t)$ has a direct physical interpretation: in the present renormalization convention it fixes the ground state  of the point interaction, which is always below the ground state of the original Hamiltonian. A prescribed variation of $\mu(t)$ may therefore model an external control of the underlying short-range impurity, for example through a tunable scattering parameter or through a controlled change of the microscopic potential whose zero-range limit is being used, the coupling constant strength is replaced by the ground state energy.

We assume in what follows that $\eta$ is sufficiently smooth on the time interval under consideration, for instance $\eta\in C^2$, and that its size is small enough that the relevant pole remains simple and separated from neighboring spectral points. Under this assumption the family of frozen Hamiltonians is a smooth family in the norm-resolvent sense near the chosen static operator.
The corresponding instantaneous subtraction point is \(M(t)=-\mu(t)\).
For each fixed value of \(t\), the frozen Hamiltonian is described by the
instantaneous $\Phi$ function
\begin{equation}
    \Phi(E,\mu(t))
    =
    \int_0^\infty K_s(a,a)
    \left(e^{-s\mu(t)}-e^{s E}\right)\,ds \;.
    \label{eq:Phi-instantaneous}
\end{equation}
The zeros of this function give the instantaneous poles of the resolvent.
We denote the pole connected to the static pole \(E_k^*\) by
\begin{equation}
    E_k(t)
    =
    E_k^*+\Delta_k(t),
    \qquad
    \Delta_k(t)
    =
    \Delta_k^{(1)}(t)+\Delta_k^{(2)}(t)+\cdots.
    \label{eq:instantaneous-pole}
\end{equation}
They are determined by the zeros of the time-dependent principal function \eqref{eq:Phi-instantaneous}
\begin{equation}
    \Phi(E_k(t),\mu(t))=0 \;.
    \label{eq:pole-equation-time}
\end{equation}
Equation \eqref{eq:pole-equation-time} should not be confused with the
solution of the time-dependent Schrödinger equation. It only determines the
instantaneous spectrum of the Hamiltonian obtained by freezing the parameter
\(\mu(t)\) at a fixed time. This is analogous to the ordinary
time-dependent Hamiltonian \(H(t)=H_0+\lambda(t)V\), for which one may first
solve the instantaneous eigenvalue problem
\begin{equation}
    H(t)\psi_n(t)= E_n(t)\psi_n(t),
\end{equation}
but the actual state is still governed by
\begin{equation}
    i\hbar\frac{\partial}{\partial t}\Psi(t)=H(t)\Psi(t) \;.
\end{equation}
In the present renormalized setting the equation
$\Phi(z,\mu(t))=0$ replaces the usual characteristic equation for the
instantaneous eigenvalues. Once the instantaneous poles and projections are
found, the dynamics is obtained from the time-dependent Schr\"odinger equation
with
\begin{equation}
    H(t)=\sum_k E_k(t)P_k(t) \;.
\end{equation}
We now specify the assumptions needed to pass from the frozen spectral problem
to a genuine non-autonomous Schr\"odinger evolution.  In the present section
the support point \(a\) is fixed and only the physical renormalization
parameter is varied,
$\mu(t)=\mu_0+\eta(t)$.
Thus the singularity remains at the same point; the time dependence enters
only through the self-adjoint extension parameter.  We assume that
\(\eta\in C^2([0,T])\) and that \(\eta\) is sufficiently small on the time
interval considered.  The smallness is understood relative to the local
spectral gaps.  More precisely, the instantaneous pole branch \(E_k(t)\)
followed in the perturbative expansion is assumed to remain simple and
uniformly separated from the rest of the instantaneous spectrum: there exists
a number \(g_k>0\) such that
$\operatorname{dist}\!\left(
    E_k(t),
    \sigma(H_{\mu(t),a})\setminus\{E_k(t)\}
    \right)
    \geq g_k,
    \qquad 0\leq t\leq T$.
Let us explain why this assumption is useful.  The instantaneous pole is
determined by the scalar equation
$\Phi(E,\mu(t))=0$. At \(t=0\), this equation has the solution \(E=E_k^*\).  Since the pole is
simple, we have
$\partial_E\Phi(E_k^*,\mu_0)
    =
    \Phi_E^k
    =
    -\sum_n
    \frac{|\phi_n(a)|^2}{(E_n-E_k^*)^2}
    \neq 0$. Therefore, by the implicit-function theorem, for \(\mu\) sufficiently close to
\(\mu_0\) there exists a unique smooth function \(E_k(\mu)\) satisfying
$\Phi(E_k(\mu),\mu)=0,
    \qquad
    E_k(\mu_0)=E_k^*$. Since \(\mu(t)=\mu_0+\eta(t)\) is \(C^2\), the instantaneous pole
$E_k(t):=E_k(\mu(t))$ is also \(C^2\) in time.  The uniform gap condition then ensures that this
pole remains isolated throughout \([0,T]\).  Hence the corresponding Riesz
projection \cite{ReedSimonI} is well defined by $P_k(t)
    =
    -\frac{1}{2\pi i}
    \oint_{\Gamma_k}
    (H_{\mu(t),a}-z)^{-1}\,dz$, where the contour \(\Gamma_k\) encloses only \(E_k(t)\) and no other spectral
point.  Since the renormalized resolvent depends smoothly on \(t\), this
contour formula shows that \(P_k(t)\) also depends smoothly on \(t\).

At the operator level, we also assume that the family of renormalized
resolvents is sufficiently smooth in \(t\), for instance in the
norm-resolvent sense.  This assumption is needed because self-adjointness of
each frozen Hamiltonian \(H_{\mu(t),a}\) separately does not by itself imply
that the time-dependent Schr\"odinger equation is well posed. 
Under this assumption, the standard theory of non-autonomous self-adjoint 
Hamiltonians yields a unique strongly continuous unitary propagator \(U(t,s)\) on finite time intervals.  
More explicitly, this propagator satisfies
$U(t,t)=I$, with   $U(t,r)U(r,s)=U(t,s)$, and $U(t,s)^*=U(s,t)$,
which solves
$i\hbar\frac{\partial}{\partial t}U(t,s)\psi
    =
    H_{\mu(t),a}U(t,s)\psi$, where $U(s,s)\psi=\psi$, for vectors in the appropriate operator, weak, or 
    quadratic-form sense, as provided by the abstract theory in Ref. \cite{Kato1980}.  Once this propagator exists, 
    the isolated instantaneous spectral data can be used perturbatively.  
    The gap condition allows one to follow a single pole branch and its Riesz projection smoothly; 
    if an eigenvalue crossing occurs, the label \(k\) no longer identifies a single smooth 
    branch and one must instead use a degenerate time-dependent perturbation analysis in the crossing subspace \cite{Kato1980}.

This fixed-center strength variation should be distinguished from the
moving-center problem of Section~\ref{sec:moving}.  Here the singular support
does not move; hence the additional geometric domain difficulty caused by a
moving singular boundary condition is absent.  In contrast, when the support
point moves, the domains of the self-adjoint extensions may themselves vary
with time, and the required assumptions must be compared with the moving
point-interaction construction of Posilicano \cite{Posilicano2007}.  Thus,
in the present section, the motion of the poles describes the
time-dependent instantaneous spectrum, while transition amplitudes between
the static eigenstates are computed by time-dependent perturbation theory,
as developed below.

\subsection{Degenerate eigenspaces and adapted sectors in the time-dependent problem}
\label{subsec:degenerate-adapted}

Before deriving the simple-pole formulas, we briefly recall how the construction
is adapted when the discrete spectrum of \(H_0\) has finite degeneracies.  This
point is relevant for the isotropic harmonic oscillator and for the Laplacian
on the sphere.  The detailed analysis of the degenerate case, including the
rank-one matrix of point values and the corresponding pole-cancellation
mechanism in the resolvent, was given in Appendix~D of Ref.~\cite{annals23}.
Here we only summarize the consequence needed for the present time-dependent
perturbation theory.

Let \(\{\varepsilon_j\}\) denote the distinct spectral values of \(H_0\),
counted without multiplicity.  The corresponding eigenspace is
$\mathcal H_j=\Ker(H_0-\varepsilon_j)$, with $d_j=\dim\mathcal H_j<\infty$. Choosing an orthonormal basis
\(\{\phi_{j\sigma}\}_{\sigma=1}^{d_j}\) of \(\mathcal H_j\), the spectral
projection onto this eigenspace has the following integral kernel
$\Pi_j(x,y)
    =
    \sum_{\sigma=1}^{d_j}
    \phi_{j\sigma}(x)\overline{\phi_{j\sigma}(y)}$. Accordingly,
\begin{equation}
    G_0(x,y|E)
    =
    \sum_j
    \frac{\Pi_j(x,y)}{\varepsilon_j-E},
    \qquad E\notin\sigma(H_0).
\end{equation}

In the renormalized formulation, the support point \(a\) enters through the
Green-function factors \(G_0(x,a|E)\) and \(G_0(a,y|E)\) in the resolvent
formula \eqref{eq:renormalized-resolvent}.  After decomposing \(G_0\) into
the eigenspaces of \(H_0\), the contribution of the degenerate eigenspace
\(\mathcal H_j\) to \(G_0(x,a|E)\) is
$\frac{\Pi_j(x,a)}{\varepsilon_j-E}$.
Thus the rank-one correction in the renormalized resolvent involves, from
\(\mathcal H_j\), only the single vector
$\Pi_j(\cdot,a)
    =
    \sum_{\sigma=1}^{d_j}
    \phi_{j\sigma}(\cdot)\overline{\phi_{j\sigma}(a)}$.
It does not couple independently to all \(d_j\) basis vectors in
\(\mathcal H_j\).  The basis-independent squared norm of this vector is
\begin{equation}
    \rho_j(a)
    :=
    \|\Pi_j(\cdot,a)\|^2
    =
    \Pi_j(a,a)
    =
    \sum_{\sigma=1}^{d_j}|\phi_{j\sigma}(a)|^2 .
\end{equation}
If \(\rho_j(a)>0\), the normalized direction selected by the resolvent is
\begin{equation}
    u_{j,a}(x)
    :=
    \frac{\Pi_j(x,a)}{\sqrt{\rho_j(a)}} .
\end{equation}

This is the only direction in \(\mathcal H_j\) that enters the singular part
of the rank-one resolvent correction.  Indeed, since
$\Pi_j(x,a)=\sqrt{\rho_j(a)}\,u_{j,a}(x)$, the contribution of the \(j\)-th degenerate eigenspace to the numerator of
the rank-one correction is proportional to
$u_{j,a}(x)\overline{u_{j,a}(y)}$. Thus, inside \(\mathcal H_j\), the correction acts only through the rank-one
operator $|u_{j,a}\rangle\langle u_{j,a}|$. Consequently, if
\begin{equation}
    f\in\mathcal H_j^\perp(a)
    :=
    \{f\in\mathcal H_j:\langle u_{j,a},f\rangle=0\},
\end{equation}
then $|u_{j,a}\rangle\langle u_{j,a}|f=0$. Therefore this orthogonal subspace does not enter the rank-one resolvent
correction and remains unaffected by the point-interaction resolvent.
Equivalently, since
\begin{equation}
    \langle u_{j,a},f\rangle
    =
    \frac{f(a)}{\sqrt{\rho_j(a)}},
    \qquad f\in\mathcal H_j,
\end{equation}
one may also write $\mathcal H_j^\perp(a)=\{f\in\mathcal H_j:f(a)=0\}$. Therefore the adapted decomposition is
\begin{equation}
    \mathcal H_j
    =
    \operatorname{span}\{u_{j,a}\}
    \oplus
    \mathcal H_j^\perp(a),
    \qquad
    \Pi_j^\perp(a)
    =
    \Pi_j-|u_{j,a}\rangle\langle u_{j,a}|.
\end{equation}
If \(\rho_j(a)=0\), then \(\Pi_j(\cdot,a)=0\), the eigenspace
\(\mathcal H_j\) does not enter the rank-one resolvent correction, and the
whole eigenspace is unaffected; in this case we set $\Pi_j^\perp(a)=\Pi_j$.

The shifted eigenvalues in the coupled sector are obtained from the zeros of
the corresponding principal function, which in the degenerate notation is
\begin{equation}
    \Phi(E,\mu_0)
    =
    \sum_j \rho_j(a)
    \left(
        \frac{1}{\varepsilon_j+\mu_0}
        -
        \frac{1}{\varepsilon_j-E}
    \right),
\end{equation}
with the same heat-kernel interpretation in the renormalized two- and
three-dimensional cases.  If \(E_\alpha^*\) is a simple zero, the corresponding
normalized eigenfunction is
\begin{equation}
    \psi_\alpha(x)
    =
    \frac{G_0(x,a|E_\alpha^*)}
    {\sqrt{-\Phi_E(E_\alpha^*,\mu_0)}}
    =
    \frac{1}{\sqrt{-\Phi_E(E_\alpha^*,\mu_0)}}
    \sum_j
    \frac{\Pi_j(x,a)}{\varepsilon_j-E_\alpha^*}.
\end{equation}
Thus the static point-interaction Hamiltonian has the adapted spectral
decomposition
\begin{equation}
    H
    =
    \sum_\alpha E_\alpha^*P_\alpha
    +
    \sum_j \varepsilon_j\Pi_j^\perp(a),
    \qquad
    P_\alpha=|\psi_\alpha\rangle\langle\psi_\alpha|.
\end{equation}
The first sum is the coupled point-interaction sector, while the second sum
contains the unaffected parts of the degenerate eigenspaces.

\begin{remark}[Time-dependent strength with fixed support]
For the strength variation studied in the present section, the support point
is fixed at \(a\) and only the renormalized parameter \(\mu(t)\) varies.
Therefore the adapted decomposition
$\mathcal H_j
    =
    \operatorname{span}\{u_{j,a}\}
    \oplus
    \mathcal H_j^\perp(a)$ is time independent.  The projections \(\Pi_j^\perp(a)\) remain fixed and are
not coupled by changing \(\mu(t)\), because all vectors in
\(\mathcal H_j^\perp(a)\) vanish at the support point.  Hence the simple-pole
formulas derived below apply directly to the shifted branches \(E_\alpha(t)\)
in the coupled sector, while the unaffected blocks
\(\varepsilon_j\Pi_j^\perp(a)\) remain unchanged.
\end{remark}

\begin{remark}[Moving center]
For the moving-center problem of Section~\ref{sec:moving}, the strength is
kept fixed but the support point changes, \(q(t)=\gamma(s(t))\).  The
evaluation functional inside each degenerate eigenspace then changes with
\(q(t)\), and one has the moving adapted decomposition
$\mathcal H_j
    =
    \operatorname{span}\{u_{j,q(t)}\}
    \oplus
    \mathcal H_j^\perp(q(t))$.
Thus both the shifted projections \(P_\alpha(q(t))\) and the unaffected
projections \(\Pi_j^\perp(q(t))\) may depend on time.  In applying the local
expansion in \(s(t)\), we assume that the rank of the evaluation functional
is constant along the trajectory in each relevant degenerate block.  That is,
for each block considered, either \(\rho_j(q(t))>0\) throughout the interval,
so that the unaffected subspace has dimension \(d_j-1\), or
\(\rho_j(q(t))=0\) throughout the interval, so that the whole block is
unaffected.  Points where \(\rho_j(q(t))\) crosses zero would change the
adapted decomposition and require a separate analysis.  Under this
constant-rank assumption, one computes the matrix elements of
\(\delta H^{(1)}\) and \(\delta H^{(2)}\) in the relevant finite-dimensional
degenerate blocks and diagonalizes the resulting perturbation matrices when
several states have the same unperturbed energy.
\end{remark}

\begin{proposition}[Pole shifts for a time-dependent scale] \label{Proppoleshifts}
For a simple shifted eigenvalue $E_k^*$, the first- and second-order pole shifts are given by
\begin{align}
    \Delta_k^{(1)}(t)
    &=-\frac{\Phi_\mu^k}{\Phi_E^k}\,\eta(t),
    \label{eq:Delta1-mu}\\
    \Delta_k^{(2)}(t)
    &=-\frac{1}{2\Phi_E^k}
    \left[
        \Phi_{EE}^k\left(\Delta_k^{(1)}(t)\right)^2
        +\Phi_{\mu\mu}^k\,\eta(t)^2
    \right].
    \label{eq:Delta2-mu}
\end{align}
Here all derivatives are evaluated at $(E_k^*,\mu_0;a)$ and the second order pole shift can be explicitly found by substituting the first order pole shift \eqref{eq:Delta1-mu} into \eqref{eq:Delta2-mu}. Higher order pole shifts can be similarly found iteratively. 
\end{proposition}

\begin{proof}
Taylor expanding $\Phi(E_k^*+\Delta_k,\mu_0+\eta)$ and using $\Phi(E_k^*,\mu_0)=0$ gives
\begin{equation}
0=\Phi_E^k\Delta_k+\Phi_\mu^k\eta
+\frac12\Phi_{EE}^k\Delta_k^2
+\Phi_{E\mu}^k\Delta_k\eta
+\frac12\Phi_{\mu\mu}^k\eta^2+\cdots.
\end{equation}
For the function \eqref{eq:Phi-series}, the $E$- and $\mu$-dependences are separated; hence $\Phi_{E\mu}^k=0$.  Equating equal orders in $\eta$ gives equations \eqref{eq:Delta1-mu} and \eqref{eq:Delta2-mu}.
\end{proof}

\begin{remark}
Although both \(\Delta_k(t)\) and \(\eta(t)\) depend on time, the above
Taylor expansion is performed pointwise in \(t\) as an expansion of the
two-variable function \(\Phi(E,\mu)\) around \((E_k^*,\mu_0)\). Thus
\(\Delta_k(t)\) and \(\eta(t)\) are treated as small increments in the
variables \(E\) and \(\mu\), respectively. Their possible dependence through
the pole equation is imposed only after the expansion. 
\end{remark}

The relevant derivatives are finite and have the explicit forms
\begin{align}
    \Phi_\mu^k
    &=-\sum_{n=0}^{\infty}\frac{|\phi_n(a)|^2}{(E_n+\mu_0)^2}
      =-\langle a|(H_0+\mu_0)^{-2}|a\rangle \;,
    \label{eq:Phi-mu}\\
    \Phi_{\mu\mu}^k
    &=2\sum_{n=0}^{\infty}\frac{|\phi_n(a)|^2}{(E_n+\mu_0)^3} \;, \\
    \Phi_{EE}^k
    &=-2\sum_{n=0}^{\infty}\frac{|\phi_n(a)|^2}{(E_n-E_k^*)^3} \;.
\end{align}
Although $\Phi_\mu^k$ is independent of $E_k^*$ for the present renormalization convention, we keep the superscript $k$ to emphasize the evaluation point.

The instantaneous projections are given by
\begin{equation}
    P_k(t)=|\Psi_k(t)\rangle\langle\Psi_k(t)| \;,
\end{equation}
where 
\begin{eqnarray}
    \Psi_k(x,t)
    =
    \frac{G_0(x,a|E_k(t))}
    {\left(-\Phi_E(E_k(t),\mu(t))\right)^{1/2}} \;.
\end{eqnarray}
The time dependence of the projection \(P_k(t)\) comes from the
time-dependent parameters entering the instantaneous eigenfunction, namely
the pole position \(E_k(t)\) and the physical parameter
\(\mu(t)\). More explicitly, one may regard
\begin{equation}
    P_k(t)=P_k(E_k(t),\mu(t)) \;.
\end{equation}
However, in the present renormalization scheme the normalization factor
\(-\Phi_E(E_k(t),\mu(t))\) is independent of \(\mu(t)\), since
\begin{equation}
    \Phi_E(E,\mu)
    =
    -\int_0^\infty sK_s(a,a)e^{sE}\,ds .
\end{equation}
Therefore \(P_k(t)\) depends on \(\mu(t)\) only indirectly through the
motion of the pole \(E_k(t)\). This allows us to write, for
notational simplicity, $P_k(t)=P_k(E_k(t))$ and to expand the projection only in powers of
$\Delta_k(t):=E_k(t)-E_k^*$. Thus
\begin{equation}
    P_k(t)
    =
    P_k
    +
    P_{k,E}\Delta_k^{(1)}(t)
    +
    P_{k,E}\Delta_k^{(2)}(t)
    +
    \frac12P_{k,EE}\bigl(\Delta_k^{(1)}(t)\bigr)^2
    +\cdots.
\end{equation}
Here \(P_{k,E}\) and \(P_{k,EE}\) denote derivatives with respect to the pole
position, evaluated at \(E_k^*\). This notation is used only to
simplify the formulas; in a more general parametrization, where the
normalization depends explicitly on the coupling parameter, additional terms
involving derivatives with respect to \(\mu\) would have to be included.

For completeness, we briefly derive the projection derivatives used below. Let us write, at a variable spectral parameter $E$ near $E_k^*$,
\begin{equation}
    P_k(E;x,y)=
    \frac{G_0(x,a|E)\overline{G_0(y,a|E)}}{-\Phi_E(E,\mu_0)} \;.
\end{equation}
Equivalently, if $N_k(E):=-\Phi_E(E,\mu_0)$ and $G_E(x):=G_0(x,a|E)$, then 

\begin{equation}
    P_k(E)=N_k(E)^{-1}|G_E\rangle\langle G_E| \;.
\end{equation}

Differentiating this rank-one expression and then setting $E=E_k^*$ gives the first derivative below; differentiating once more gives the second derivative. The integral kernel of the first derivative of $P_k$ with respect to $E$ at $E_{k}^*$ is
\begin{align}
    P_{k,E}(x,y)
    &=-\frac{G_{0,E}^k(x)\overline{G_0^k(y)}+G_0^k(x)\overline{G_{0,E}^k(y)}}{\Phi_E^k}
    +\frac{\Phi_{EE}^k}{(\Phi_E^k)^2}G_0^k(x)\overline{G_0^k(y)} \;.
    \label{eq:P-E}
\end{align}
Similarly, the second derivative of $P_k$ at $E_{k}^*$ has the following kernel
\begin{align}
    P_{k,EE}(x,y)
    &=-\frac{G_{0,EE}^k(x)\overline{G_0^k(y)}+2G_{0,E}^k(x)\overline{G_{0,E}^k(y)}
    +G_0^k(x)\overline{G_{0,EE}^k(y)}}{\Phi_E^k}
    \nonumber\\
    &\quad
    +2\frac{\Phi_{EE}^k}{(\Phi_E^k)^2}
    \bigl[G_{0,E}^k(x)\overline{G_0^k(y)}+G_0^k(x)\overline{G_{0,E}^k(y)}\bigr]
    \nonumber\\
    &\quad
    +\left(
        \frac{\Phi_{EEE}^k}{(\Phi_E^k)^2}
        -2\frac{(\Phi_{EE}^k)^2}{(\Phi_E^k)^3}
    \right)G_0^k(x)\overline{G_0^k(y)} \;.
    \label{eq:P-EE}
\end{align}
Consequently, the first- and second-order changes of the Hamiltonian are given by 
\begin{align}
    \delta H^{(1)}(t)
    &=\sum_k\Delta_k^{(1)}(t)\bigl(P_k+E_k^*P_{k,E}\bigr) \;,
    \label{eq:H1} \\
    \delta H^{(2)}(t)
    &=\sum_k\Bigg[
        \Delta_k^{(2)}(t)\bigl(P_k+E_k^*P_{k,E}\bigr)
        \nonumber \\ & +\left(\Delta_k^{(1)}(t)\right)^2P_{k,E}
        +\frac12E_k^*\left(\Delta_k^{(1)}(t)\right)^2P_{k,EE}
    \Bigg] \;,
    \label{eq:H2}
\end{align}
where $\Delta_{k}^{(1)}$, $\Delta_{k}^{(2)}$, $P_{k,E}$, and $P_{k,EE}$  are given by (\ref{eq:Delta1-mu}), (\ref{eq:Delta2-mu}), (\ref{eq:P-E}), and (\ref{eq:P-EE}).

To find the transition amplitudes, we need to find the following matrix elements
\begin{eqnarray}
    R_{mn}(t) &:= & \langle \psi_m | H(t)|\psi_n \rangle \;, \\ d_n(t) & := & \langle \psi_n | H(t)|\psi_n \rangle  \;,
\end{eqnarray}
up to first and second order.

The normalization identity $\int |\psi_k(x)|^2 d\mu(x)=1$ is time-independent, or equivalently $E$-independent so that
\begin{eqnarray}
    \frac{\partial}{\partial E} \int_{\mathcal{M}} |G_0(x,a|E)|^2 \dd\mu(x)\Big|_{E=E_{k}^{*}} = -\frac{\partial^2 \Phi}{\partial E^2}\Big|_{E=E_{k}^{*}} \;.
\end{eqnarray}
This gives
\begin{eqnarray}
    2\,\operatorname{Re}\int_{\mathcal{M}} \overline{G_0^k(x)}\,G_{0,E}^k(x)\,\dd\mu(x)=-\Phi_{EE}^{k} \;.
\end{eqnarray}
Then, using the orthonormality of the eigenfunctions $\psi_n$ we obtain
\begin{eqnarray}
\langle \psi_n |P_{k,E}|\psi_n\rangle & = & \iint_{\mathcal{M} \times \mathcal{M}} \overline{\psi_n(x)}\,P_{k,E}(x,y)\,\psi_n(y)\,\dd\mu(x)\dd\mu(y) \nonumber \\ 
& = & 0 \;.    \label{eq:PE-diag-zero}
\end{eqnarray}
Thus the first-order diagonal correction is simply
\begin{equation}
    d_n^{(1)}(t):=\langle\psi_n|\delta H^{(1)}(t)|\psi_n\rangle
    =\Delta_n^{(1)}(t) \;.
    \label{eq:diag-first}
\end{equation}
For $m\ne n$, the first-order off-diagonal matrix element is
\begin{align}
    R_{mn}^{(1)}(t)
    &:=\langle\psi_m|\delta H^{(1)}(t)|\psi_n\rangle
    \nonumber\\
    &=E_n^*\Delta_n^{(1)}(t)
    \frac{\langle\psi_m|G_{0,E}^n\rangle}{\sqrt{-\Phi_E^n}}
    +E_m^*\Delta_m^{(1)}(t)
    \frac{\langle G_{0,E}^m|\psi_n\rangle}{\sqrt{-\Phi_E^m}} \;.
    \label{eq:Rmn-first}
\end{align}
Equivalently, after substituting \eqref{eq:Delta1-mu},
\begin{align}
    R_{mn}^{(1)}(t)
    =\eta(t)\left[
    -E_n^*\frac{\Phi_\mu^n}{\Phi_E^n}
    \frac{\langle\psi_m|G_{0,E}^n\rangle}{\sqrt{-\Phi_E^n}}
    -E_m^*\frac{\Phi_\mu^m}{\Phi_E^m}
    \frac{\langle G_{0,E}^m|\psi_n\rangle}{\sqrt{-\Phi_E^m}}
    \right].
    \label{eq:Rmn-first-eta}
\end{align}

We now insert these Hamiltonian matrix elements into time-dependent perturbation theory.  A state is expanded as
\begin{equation}
    \Psi(x,t)=\sum_m b_m(t)\psi_m(x)\exp\left(-\frac{i}{\hbar}E_m^*t\right).
\end{equation}
The diagonal part is removed by
\begin{equation}
    b_m(t)=\exp\left[-\frac{i}{\hbar}\int_0^t d_m(t')\,dt'\right]c_m(t) \;,
\end{equation}
where $d_m(t)=d_m^{(1)}(t)+d_m^{(2)}(t)+\cdots$.  The coefficients $c_m(t)$ satisfy
\begin{equation}
\frac{dc_m}{dt}
=-\frac{i}{\hbar}\sum_n R_{mn}(t)
\exp\left[\frac{i}{\hbar}(E_m^*-E_n^*)t+\frac{i}{\hbar}\int_0^t(d_m(t')-d_n(t'))\,dt'\right]c_n(t) \;.
\end{equation}
Expanding $c_m=c_m^{(0)}+c_m^{(1)}+c_m^{(2)}+\cdots$ gives
\begin{eqnarray}
    \dot c_m^{(0)} & = & 0 \;, \\ 
 \dot c_m^{(1)} &
    = & -\frac{i}{\hbar}\sum_n R_{mn}^{(1)}(t)
    \exp\left[\frac{i}{\hbar}(E_m^*-E_n^*)t\right]c_n^{(0)} \;, \\ 
    \dot c_m^{(2)} & = &
-\frac{i}{\hbar}\sum_n R_{mn}^{(1)}(t)
\exp\left[\frac{i}{\hbar}(E_m^*-E_n^*)t\right]c_n^{(1)}(t) \nonumber 
\\
& & -\frac{i}{\hbar}\sum_n R_{mn}^{(1)}(t)
\exp\left[\frac{i}{\hbar}(E_m^*-E_n^*)t\right]
\left[\frac{i}{\hbar}\int_0^t(d_m^{(1)}(t')-d_n^{(1)}(t'))\,dt'\right]c_n^{(0)} \nonumber 
\\
& &-\frac{i}{\hbar}\sum_n R_{mn}^{(2)}(t)
\exp\left[\frac{i}{\hbar}(E_m^*-E_n^*)t\right]c_n^{(0)}.
\label{eq:second-order-transition-equation}
\end{eqnarray}
The term $d_m^{(2)}$ does not enter equation \eqref{eq:second-order-transition-equation}; it would multiply $R^{(1)}$ in the phase expansion and hence contributes only at third order.
The second-order dynamics therefore requires exactly the three terms displayed above. 

If the system is initially in the state $\psi_l$, then $c_{n}^{(0)}=\delta_{nl}$. Assuming $c_{m}^{(1)}(0)=0$, the first-order transition amplitude to $\psi_m$ is
\begin{equation}
    c_m^{(1)}(t)=
    -\frac{\ii}{\hbar}\int_0^t
    R_{ml}^{(1)}(t')\,
    \e^{\ii\omega_{ml}t'}\dd t',
    \qquad
    \omega_{ml}=\frac{E_m^*-E_l^*}{\hbar} \;.
    \label{eq:first-amplitude}
\end{equation}

We summarize the main results of this section in the following two
propositions.

\begin{proposition}[Hamiltonian expansion for a time-dependent renormalized strength] \label{prop:strength-H-expansion}
Let \(E_k^*\) be simple shifted poles of the static renormalized
Hamiltonian, and let the renormalized strength be varied according to
\[
    \mu(t)=\mu_0+\eta(t),
    \qquad |\eta(t)|\ll 1 .
\]
Assume that the instantaneous projections depend on \(\mu(t)\) only through
the motion of the poles \(E_k(t)=E_k^*+\Delta_k(t)\). Then the
non-autonomous Hamiltonian admits the perturbative expansion
\[
    H(t)=H+\delta H^{(1)}(t)+\delta H^{(2)}(t)+O(\eta^3),
\]
where
\[
    H=\sum_k E_k^*P_k,
\]
and
\begin{eqnarray}
    \begin{aligned}
    \delta H^{(1)}(t)
    &=
    \sum_k
    \Delta_k^{(1)}(t)
    \bigl(P_k+E_k^*P_{k,E}\bigr),
    \label{prop:eq:H1-summary} \\
    \delta H^{(2)}(t)
    &=
    \sum_k
    \Bigg[
        \Delta_k^{(2)}(t)
        \bigl(P_k+E_k^*P_{k,E}\bigr)
        +
        \bigl(\Delta_k^{(1)}(t)\bigr)^2P_{k,E}
        +
        \frac12E_k^*
        \bigl(\Delta_k^{(1)}(t)\bigr)^2P_{k,EE}
    \Bigg].
    \label{prop:eq:H2-summary}
\end{aligned} \nonumber
\end{eqnarray}
Here \(P_{k,E}\) and \(P_{k,EE}\) are the first and second derivatives of
the rank-one projection with respect to the pole position, evaluated at
\(E=E_k^*\), while \(\Delta_k^{(1)}\) and \(\Delta_k^{(2)}\) are the first-
and second-order pole shifts given in Proposition \ref{Proppoleshifts}.
\end{proposition}

\begin{proposition}[Perturbative equations for transition amplitudes] \label{prop:strength-transition-equations}
Let the state be expanded in the static eigenbasis as
\[
    \Psi(x,t)
    =
    \sum_m b_m(t)\psi_m(x)
    \exp\left(-\frac{i}{\hbar}E_m^*t\right),
\]
and remove the diagonal part by writing
\[
    b_m(t)
    =
    \exp\left[
    -\frac{i}{\hbar}\int_0^t d_m(t')\,dt'
    \right]c_m(t).
\]
If
\[
    c_m(t)=c_m^{(0)}(t)+c_m^{(1)}(t)+c_m^{(2)}(t)+\cdots,
\]
then, up to second order, the coefficients satisfy
\begin{eqnarray}
\begin{aligned}
    \dot c_m^{(0)}
    &=
    0,
    \label{prop:eq:c0-summary}\\
    \dot c_m^{(1)}
    &=
    -\frac{i}{\hbar}
    \sum_n
    R_{mn}^{(1)}(t)
    \exp\left[
    \frac{i}{\hbar}(E_m^*-E_n^*)t
    \right]
    c_n^{(0)},
    \label{prop:eq:c1-summary}\\
    \dot c_m^{(2)}
    &=
    -\frac{i}{\hbar}
    \sum_n
    R_{mn}^{(1)}(t)
    \exp\left[
    \frac{i}{\hbar}(E_m^*-E_n^*)t
    \right]
    c_n^{(1)}(t)\\
    &\quad
    -\frac{i}{\hbar}
    \sum_n
    R_{mn}^{(1)}(t)
    \exp\left[
    \frac{i}{\hbar}(E_m^*-E_n^*)t
    \right]
    \left[
    \frac{i}{\hbar}
    \int_0^t
    \bigl(d_m^{(1)}(t')-d_n^{(1)}(t')\bigr)\,dt'
    \right]
    c_n^{(0)}
     \\
    &\quad
    -\frac{i}{\hbar}
    \sum_n
    R_{mn}^{(2)}(t)
    \exp\left[
    \frac{i}{\hbar}(E_m^*-E_n^*)t
    \right]
    c_n^{(0)} \;, 
    \label{prop:eq:c2-summary}
\end{aligned} \nonumber
\end{eqnarray}
where $R_{mn}^{(r)}=\langle \psi_m|\delta H^{(r)}(t)|\psi_n\rangle$ and $\delta H^{(r)}$ are given by Proposition~\ref{prop:strength-H-expansion}. 
Thus the second-order transition amplitude has three distinct contributions:
the time-ordered iteration of the first-order off-diagonal perturbation, the
correction coming from the first-order diagonal phase, and the direct
second-order off-diagonal perturbation.
\end{proposition}

\subsection{Delta-kick perturbation}
\label{subsec:delta-kick}

As a simple impulsive example, take
\begin{equation}
    \eta(t)=\epsilon\delta(t-t_0) \;,
\end{equation}
for $t_0>0$. Let 
$\mathcal V_{mn}$ be the coefficient of $\eta(t)$ in equation \eqref{eq:Rmn-first-eta}.  If the system is initially in the state $\psi_l$,  for \(t\neq t_0\)
\begin{equation}
    c_m^{(1)}(t)
    =-\frac{i\epsilon}{\hbar}\Theta(t-t_0)
    \exp\left[\frac{i}{\hbar}(E_m^*-E_l^*)t_0\right] \mathcal V_{ml} \;,
    \label{eq:delta-kick-first-order}
\end{equation}
where we have used $\int_{0}^{t} \delta(t'-t_0)f(t') dt' = \Theta(t-t_0) f(t_0)$ and $\Theta$ is the Heaviside step function. The value at \(t=t_0\) depends on the convention chosen for the Heaviside function, and is immaterial for the transition probabilities considered here. Thus there is no transition before the kick, and after the kick the transition amplitude is frozen at the value produced at $t=t_0$.  The corresponding first-order transition probability is
\begin{equation}
    P_{l \to m}^{(1)}(t)=\frac{\epsilon^2}{\hbar^2}|\mathcal V_{ml}|^2,
    \qquad t>t_0 \;.
\end{equation}

At second order one must distinguish the time-ordered iteration of two first-order insertions from the direct second-order Hamiltonian term.  The latter contains $(\eta(t))^2$ and is therefore not defined for an ideal Dirac delta unless the pulse is regularized.

\section{Perturbation Theory for a Moving Delta Potential}
\label{sec:moving}

Let \(\gamma:(-\varepsilon,\varepsilon)\to\mathcal M\) be a smooth curve
with \(\gamma(0)=a\) and with local coordinates $\gamma(s)=(\gamma^1(s), \cdots,\gamma^d(s))$. We next keep the renormalized strength fixed and let the support point move on a prescribed curve, parametrized by its arclength $s$.  Let
\begin{equation}
    q(t)=\gamma(s(t)),
    \qquad
    \gamma(0)=a,
    \qquad
    s(0)=0 \;,
\end{equation}
where $s(t)$ is small. Therefore the Green factors in the instantaneous Hamiltonian change in two independent ways:
\begin{equation}
G_0(x,a|E_k^*)\quad\longrightarrow\quad G_0(x,\gamma(s(t))|E_k^*(t)) \;.    
\end{equation}
There is an energy dependence through \(E_k^*(t)\), and there is a spatial dependence through the
moving point \(\gamma(s(t))\).  The calculation below keeps both effects.

Let us also specify the assumptions under which the moving-center construction
is used.  Posilicano's analysis of a moving point interaction in
\(\mathbb R^3\) shows that the main difficulty is not merely the time
dependence of the Hamiltonian, but the motion of the singularity itself:
the singular boundary condition moves with the support point, and therefore
the operator domain, and in general also the quadratic-form domain, may vary
with time \cite{Posilicano2007}.  In Posilicano's setting, a \(C^2\)
trajectory is sufficient for the construction of a strongly continuous
unitary propagator, while \(C^3\) regularity gives stronger control at the
operator-domain level.

Our setting is different, since the unperturbed Hamiltonian has purely
discrete spectrum and the perturbative expansion is written in terms of the
instantaneous spectral data of the corresponding renormalized
point-interaction Hamiltonian.  We therefore do not claim that Posilicano's
theorem applies verbatim on the manifolds or confining systems considered
here.  Rather, we assume that the changing family of self-adjoint Hamiltonians admit a strongly continous unitary propagator $U(t,s)$ with sufficient regularity  and we impose the analogous local assumptions needed for the
present perturbative construction.  Namely, the strength parameter is kept
fixed, the support point is written as $q(t)=\gamma(s(t))$ with $\gamma(0)=a$ and $s(0)=0$,
where \(\gamma\) is \(C^3\) in a neighbourhood of \(s=0\), \(s\in C^2([0,T])\),
and \(q(t)\) remains in a single coordinate neighbourhood of \(a\) for
\(0\leq t\leq T\).  The expansion is local in the arclength displacement
\(s(t)\), not only in the velocity of the center.  Thus we assume
$\sup_{0\leq t\leq T}|s(t)|\ll \ell_{\mathrm{loc}}$,
where \(\ell_{\mathrm{loc}}\) denotes the smallest length scale on which the
Green function, the principal function, and the relevant Riesz projections
change appreciably along the curve.  We also assume that the pole branch
under consideration remains simple and uniformly separated from the rest of
the instantaneous spectrum: $\operatorname{dist}\!\left(
    E_k(q(t)),
    \sigma(H_{\mu_0,q(t)})\setminus\{E_k(q(t))\}
    \right)
    \geq g_k>0$,
with  $0\leq t\leq T$.
Under these hypotheses the renormalized resolvent, the pole \(E_k(q(t))\),
and the corresponding Riesz projection depend smoothly on the moving support
point, and the expansion below is understood as a local asymptotic expansion
of the non-autonomous evolution generated by the moving self-adjoint point
interaction.

\begin{remark}
In the moving-center case we do not prove, from the stated local assumptions,
a general existence theorem for the non-autonomous propagator.  Rather, we
assume that the moving self-adjoint point interaction admits a strongly
continuous unitary propagator \(U(t,s)\) with sufficient regularity for the
instantaneous resolvent and Riesz projection expansion used above. However, we believe that there could be an
analogue of Posilicano's construction for compact manifolds: one may first
freeze the moving singularity by a time-dependent unitary transformation and
then apply the Kisyński--Kato theorem \cite{Kato1970, Kisynski1964} to the resulting fixed-domain
quadratic forms. This requires a separate careful analysis that we refrain from attempting in this work. 
\end{remark}

Throughout this section all derivatives of \(G_0\) and \(\Phi\) are evaluated at
\begin{equation}
E=E_k^*,\qquad q=a=\gamma(0) \;,    
\end{equation}
unless stated otherwise.  Note that {\it spatial derivatives act on the moving support coordinate} \(q\).  We write
\begin{equation}
v^i:=\dot\gamma^i(0),
\qquad
b^i:=\frac{D\dot\gamma^i}{Ds}(0) \;,    
\end{equation}
where \(D/Ds\) is the covariant derivative along the curve. In local
coordinates \cite{Rosenberg},
\begin{equation}
    \frac{D\dot\gamma^i}{Ds}
    =
    \frac{d^2\gamma^i}{ds^2}
    +
    \Gamma^i_{jk}(\gamma(s))
    \frac{d\gamma^j}{ds}
    \frac{d\gamma^k}{ds} \;,
\end{equation}
where $\Gamma^{i}_{jk}$ is the Christoffel symbol. Thus $b^i=\ddot\gamma^i(0)+\Gamma^i_{jk}(a)v^jv^k$. Then, the expansion of $\gamma$ near \(s=0\) is given by
$\gamma^i(s)=a^i+s v^i+\frac{s^2}{2}b^i+\cdots$. If \(\gamma\) is a geodesic with affine parameter, then \(b^i=0\). For a scalar function \(F\), the first derivative along the curve at $s=0$ is
\begin{equation}
    \left.\frac{d}{ds}F(\gamma(s))\right|_{s=0}
    =
    v^i\partial_iF(a) \;.
\end{equation}
The second derivative is
\begin{equation}
    \frac{d^2}{ds^2}F(\gamma(s))
    =
    \ddot\gamma^i(s)\partial_iF(\gamma(s))
    +
    \dot\gamma^i(s)\dot\gamma^j(s)
    \partial_j\partial_iF(\gamma(s)) \;.
\end{equation}
Using
\begin{equation}
    b^i
    =
    \ddot\gamma^i(0)
    +
    \Gamma^i_{jk}(a)v^jv^k
\end{equation}
and
\begin{equation}
    \nabla_j\partial_iF
    =
    \partial_j\partial_iF
    -
    \Gamma^m_{ji}\partial_mF \;,
\end{equation}
the Christoffel-symbol terms cancel. Therefore
\begin{equation}
    \left.\frac{d^2}{ds^2}F(\gamma(s))\right|_{s=0}
    =
    b^i\partial_iF(a)
    +
    v^iv^j\nabla_j\partial_iF(a) \;.
\end{equation}
Substituting this into the ordinary one-dimensional Taylor expansion of
\(F(\gamma(s))\) gives
\begin{equation}
    F(\gamma(s))
    =F(a)+sv^i\partial_iF(a)
    +\frac{s^2}{2}\left[b^i\partial_iF(a)+v^iv^j\nabla_i\partial_jF(a)\right]+O(s^3).
    \label{eq:covariant-taylor-moving}
\end{equation}
Spatial derivatives act on the moving support coordinate.  Applying this to the Green function and combining it with the energy expansion gives, up to second order,
\begin{equation}
\begin{aligned}
G_0(x,\gamma(s(t))|E_k^*(t))
={}&G_0(x,a|E_k^*)
+\Delta _k^{(1)}(t)\partial_EG_0(x,a|E_k^*)
+s(t)v^i\partial_iG_0(x,a|E_k^*)
\\
&+\Delta_k^{(2)}(t)\partial_EG_0(x,a|E_k^*)
+\frac12\bigl(\Delta_k^{(1)}(t)\bigr)^2\partial_E^2G_0(x,a|E_k^*)
\\
&+s(t)\Delta_k^{(1)}(t)v^i\partial_i\partial_EG_0(x,a|E_k^*)
\\
&+\frac{s(t)^2}{2}\left[b^i\partial_iG_0(x,a|E_k^*)
+v^iv^j\nabla_i\partial_jG_0(x,a|E_k^*)\right] \;.
\end{aligned}
\label{eq:moving-green-expansion}
\end{equation}
The same formula applies to $G_0(\gamma(s(t)),y|E_k^*(t))$, with spatial derivatives acting on the first argument.

The instantaneous poles are now defined by
\begin{equation}
    \Phi(E_k^*(t),q(t))=0 \;.
\end{equation}
Let $\Phi^{k}_E= \partial_E \Phi(E_{k}^{*},a)$ and $\Phi^{k}_{EE}= \partial_{E}^{2} \Phi(E_{k}^{*},a)$. Expanding this equation by using equation \eqref{eq:covariant-taylor-moving} gives
\begin{eqnarray}
    & \Phi_E^k\Delta_k^{(1)}(t)+s(t)v^i\partial_i\Phi^k + \Phi_E^k \Delta_k^{(2)}(t)
+\frac12\Phi_{EE}^k\bigl(\Delta_k^{(1)}(t)\bigr)^2
+s(t)\Delta_k^{(1)}(t)v^i\partial_i\Phi_E^k \nonumber 
\\
&+\frac{s(t)^2}{2}
\left[b^i\partial_i\Phi^k+v^iv^j\nabla_i\partial_j\Phi^k\right] =0 \;.
\end{eqnarray}
One can solve this equation iteratively order by order. The first-order shift yields
\begin{equation}
    \Phi_E^k\Delta_k^{(1)}(t)+s(t)v^i\partial_i\Phi^k=0,
    \qquad
    \Delta_k^{(1)}(t)=-s(t)\frac{v^i\partial_i\Phi^k}{\Phi_E^k} \;,
    \label{eq:moving-first-order-shift}
\end{equation}
where all derivatives are evaluated at $(E,q)=(E_k^*,a)$.  The second-order shift satisfies
\begin{equation}
\begin{aligned}
\Phi_E^k \Delta_k^{(2)}(t)
={}&-\frac12\Phi_{EE}^k\bigl(\Delta_k^{(1)}(t)\bigr)^2
-s(t)\Delta_k^{(1)}(t)v^i\partial_i\Phi_E^k
\\
&-\frac{s(t)^2}{2}
\left[b^i\partial_i\Phi^k+v^iv^j\nabla_i\partial_j\Phi^k\right].
\end{aligned}
\label{eq:moving-second-order-shift}
\end{equation}
The mixed term in equation \eqref{eq:moving-second-order-shift} carries the factor $s(t)v^i$; this is necessary for it to have the correct perturbative order.
Hence, we prove the following
\begin{proposition}[Pole shifts for a moving support point] \label{movingpoleshifts}
For a simple shifted eigenvalue, the first- and second-order shifts are given by
\begin{align}
    \Delta_k^{(1)}(t)
    &=-s(t)\frac{v^i\partial_i\Phi^k}{\Phi_E^k} \;,
    \label{eq:Delta1-moving}\\
    \Delta_k^{(2)}(t)
    &=-\frac{s(t)^2}{\Phi_E^k}
    \Bigg[
        \frac12\Phi_{EE}^k\biggl(\frac{v^i\partial_i\Phi^k}{\Phi_E^k}\biggr)^2
        - \biggl(\frac{v^j\partial_j \Phi^k}{\Phi_E^k}\biggr) v^i\partial_i\Phi_E^k \nonumber \\ & 
\hspace{5cm} +\frac12\left[b^i\partial_i\Phi^k+v^iv^j\nabla_i\partial_j\Phi^k\right]
    \Bigg] \;.
    \label{eq:Delta2-moving}
\end{align}
\end{proposition}

The small parameter in this expansion is the displacement of the support point measured in the local geometry of the problem. More precisely, the second-order approximation is reliable when $|s(t)|$ is small compared with the length scale on which the Green function and the principal function vary appreciably, and when the induced pole displacement is small compared with the distance to neighboring spectral points. We therefore exclude level crossings and points where the relevant derivatives of $\Phi$ or the normalization factor become singular. The exactly solvable circle example in Section~\ref{sec:uniform-moving-delta-circle-example} makes this statement quantitative in a discrete-spectrum setting: the local expansion is controlled by the displacement measured in the natural variation length of the static eigenfunctions, while the full dynamical comparison also involves the dimensionless rotation, or flux, parameter generated by the velocity.

The instantaneous moving-center Hamiltonian is
\begin{equation}
    H(t) = \sum_k h_k(t)  
    =\sum_k E_k^*(t)
    \frac{G_0(\cdot,q(t)|E_k^*(t))\,\overline{G_0(\cdot,q(t)|E_k^*(t))}}
    {-\Phi_E(E_k^*(t),q(t))} \;.
\end{equation}
The kernel of the instantaneous spectral term is
\begin{equation}
h_{k}(t; x,y) =E_k^*(t)
\frac{G_0(x,q(t)|E_k^*(t))\overline{G_0(y,q(t)|E_k^*(t))}}
{-\partial_E\Phi(E_k^*(t),q(t))} \;.    
\end{equation}
This has the same architecture as before:
\begin{equation}
h_{k}(t; x,y)=E_k^*(t)\,G_k(t;x)\,\overline{G_k(t;y)}\,D_k(t) \;,    
\end{equation}
where $G_k(t;x)=G_0(x,q(t)|E_k^*(t))$, and
\begin{equation}
D_k(t)=\frac{1}{-\partial_E\Phi(E_k^*(t),q(t))} \;.    
\end{equation}
The expansion of the Green function is given by equation \eqref{eq:moving-green-expansion} with the following form
\begin{equation}
G_k(t;x)=G_k^{(0)}(x)+G_k^{(1)}(t;x)+G_k^{(2)}(t;x)+\cdots,    
\end{equation}
and similarly for \(D_k\).  The zeroth-order terms are
\begin{eqnarray}
G_k^{(0)}(x)=G_0(x,a|E_k^*),\qquad
D_k^{(0)}=-\frac1{\Phi_{E}^{k}} \;.    
\end{eqnarray}
The first-order Green factors are
\begin{equation}
G_k^{(1)}(t;x)=
\Delta_k^{(1)}(t) \partial_EG_0(x,a|E_k^*)
+s(t) v^i\partial_iG_0(x,a|E_k^*) \;,    
\end{equation}
where the relation between $s$ and $\Delta_{k}^{(1)}$ is given by equation \eqref{eq:moving-first-order-shift}. The second-order Green factors are
\begin{eqnarray}
\begin{aligned}
G_k^{(2)}(t;x)=&\,
\Delta_k^{(2)}(t) \partial_E G_0(x,a|E_k^*)
+\frac12\left(\Delta_k^{(1)}(t)\right)^2\partial_E^2G_0(x,a|E_k^*)
\cr
&+s (t)\,\Delta_k^{(1)}(t) v^i\partial_i\partial_EG_0(x,a|E_k^*)
\cr
&+\frac{s^2(t)}{2}\left[
 b^i\partial_iG_0(x,a|E_k^*)
+v^iv^j\nabla_i\partial_jG_0(x,a|E_k^*)
\right] \;,
\end{aligned}
\end{eqnarray}
For the denominator, define explicitly
\begin{equation}
Q_k^{(1)}:=\Phi_{EE}^k \Delta_k^{(1)}+s v^i \partial_i \Phi_{E}^{k} \;,    
\end{equation}
where we have suppressed the independent variables in the longer products to improve readability. Similarly,
\begin{eqnarray}
Q_k^{(2)}:=
\Phi_{EE}^{k} \Delta_k^{(2)}
+\frac12\Phi_{EEE}^{k}\left(\Delta_k^{(1)}\right)^2
+s v^i \partial_i \Phi_{EE}^{k} \Delta_k^{(1)}
+\frac{s^2}{2}\left[b^i \partial_i \Phi_{E}^{k} +v^iv^j \nabla_i \partial_j\Phi_{E}^{k} \right].    
\end{eqnarray}
Then
\begin{equation}
\partial_E\Phi(E_k^*(t),q(t))
=\Phi_{E}^{k} +Q_k^{(1)}+Q_k^{(2)}+\text{third and higher order terms},    
\end{equation}
and therefore
\begin{equation}
D_k^{(1)}=\frac{Q_k^{(1)}}{(\Phi_{E}^{k})^2} \;,
\qquad
D_k^{(2)}=\frac{Q_k^{(2)}}{(\Phi_{E}^{k})^2}-\frac{(Q_k^{(1)})^2}{(\Phi_{E}^{k})^3} \;.    
\end{equation}
Then, we have 
\begin{equation}
\frac1{-(\Phi_{E}^k+Q_k^{(1)}+Q_k^{(2)})}
=-\frac1{\Phi_{E}^k}+\frac{Q_k^{(1)}}{(\Phi_{E}^{k})^2}
+\frac{Q_k^{(2)}}{(\Phi_{E}^k)^2}
-\frac{(Q_k^{(1)})^2}{(\Phi_{E}^k)^3} \;.    
\end{equation}
The first-order kernel is then
\begin{eqnarray}
\begin{aligned}
h_{k}^{(1)}(t; x,y)=&\,
\Delta_k^{(1)} (t) G_k^{(0)}(x)\overline{G_k^{(0)}(y)}D_k^{(0)}
\cr
&+E_k^*\left[G_k^{(1)}(t; x)\overline{G_k^{(0)}(y)}+G_k^{(0)}(x)\overline{G_k^{(1)}(t;y)}\right]D_k^{(0)}
\cr
&+E_k^*G_k^{(0)}(x)\overline{G_k^{(0)}(y)}D_k^{(1)}(t).
\end{aligned}
\end{eqnarray}
The second-order kernel is
\begin{eqnarray}
\begin{aligned}
& \, h_{k}^{(2)}(t;x,y)=
\Delta_k^{(2)}(t) G_k^{(0)}(x)\overline{G_k^{(0)}(y)}D_k^{(0)}
\cr
&+\Delta_k^{(1)}(t)\left(G_k^{(1)}(t;x)\overline{G_k^{(0)}(y)}+G_k^{(0)}(x)\overline{G_k^{(1)}(t;y)}\right)D_k^{(0)}
+\Delta_k^{(1)}(t) G_k^{(0)}(x)\overline{G_k^{(0)}(y)}D_k^{(1)}(t)
\cr
&+E_k^*\left(G_k^{(2)}(t;x)\overline{G_k^{(0)}(y)}+G_k^{(0)}(x)\overline{G_k^{(2)}(t;y)}+G_k^{(1)}(x)\overline{G_k^{(1)}(t;y)}\right)D_k^{(0)}
\cr
&+E_k^*\left(G_k^{(1)}(t;x)\overline{G_k^{(0)}(y)}+G_k^{(0)}(x)\overline{G_k^{(1)}(t;y)}\right)D_k^{(1)}(t)
+E_k^*G_k^{(0)}(x)\overline{G_k^{(0)}(y)}D_k^{(2)}(t) \;.
\end{aligned}
\end{eqnarray}
Finally,
\begin{eqnarray}
\delta H^{(1)}(t;x,y)=\sum_k h_{k}^{(1)}(t;x,y),
\qquad
\delta H^{(2)}(t; x,y)=\sum_k h_{k}^{(2)}(t;x,y).    \label{eq:deltaHmoving}
\end{eqnarray}
The diagonal and off-diagonal perturbations are obtained exactly as in the fixed-center calculation:
\begin{equation}
d_m^{(r)}=\langle\psi_m|\delta H^{(r)}|\psi_m\rangle,
\qquad
R_{mn}^{(r)}=\langle\psi_m|\delta H^{(r)}|\psi_n\rangle,
\quad m\neq n \;,    
\end{equation}
where $r=1,2$. At first order, the normalized-projector argument again gives
\begin{equation}
d_m^{(1)}=\Delta_m^{(1)}.    
\end{equation}
The off-diagonal first-order moving-center term is generated by the first derivative of the projector.
Writing it explicitly gives, for \(m\neq n\),
\begin{eqnarray}
\begin{aligned}
&\, R_{mn}^{(1)}=
E_n^*\frac{1}{\sqrt{-\Phi_E(E_n^*,a)}}
\int_{\mathcal{M}} \overline{\psi_m(x)}
\left[
\Delta_n^{(1)}\partial_EG_0(x,a|E_n^*)
+s v^i\partial_iG_0(x,a|E_n^*)
\right]\dd\mu(x)
\cr
&+E_m^*\frac{1}{\sqrt{-\Phi_E(E_m^*,a)}}
\int_{\mathcal{M}}
\overline{\left[
\Delta_m^{(1)}\partial_EG_0(y,a|E_m^*)
+s v^i\partial_iG_0(y,a|E_m^*)
\right]}
\psi_n(y)\dd\mu(y).
\end{aligned}
\end{eqnarray}
For the second-order calculations, it is useful to define
\begin{equation}
N_k:=-\Phi_E(E_k^*,a)>0,    
\end{equation}
and, for \(r=0,1,2\),
\begin{equation}
I_{mk}^{(r)}:=\int_{\mathcal{M}} \overline{\psi_m(x)}\,G_k^{(r)}(x)\dd\mu(x) \;.  
\end{equation}
This is not a new object; it only abbreviates the Green-function integrals already written above.  At
zeroth order,
\begin{equation}
I_{mk}^{(0)}=\sqrt{N_k}\,\delta_{mk} \;,    
\end{equation}
The first-order integral is explicitly
\begin{equation}
I_{mk}^{(1)}=\int_{\mathcal{M}} \overline{\psi_m(x)}
\left[
\Delta_k^{(1)}\partial_E G_0(x,a|E_k^*)
+s v^i\partial_iG_0(x,a|E_k^*)
\right]\dd\mu(x) \;,    \label{eq:I1}
\end{equation}
The second-order integral is
\begin{eqnarray}
\begin{aligned}
I_{mk}^{(2)}=\int_{\mathcal{M}} \overline{\psi_m(x)}\Bigg[&
\Delta_k^{(2)}\partial_E G_0(x,a|E_k^*)
+\frac12\left(\Delta_k^{(1)}\right)^2\partial_E^2G_0(x,a|E_k^*)
\cr
&+s\,\Delta_k^{(1)}v^i\partial_i\partial_EG_0(x,a|E_k^*)
\cr
&+\frac{s^2}{2}\left(
 b^i\partial_iG_0(x,a|E_k^*)
+v^iv^j\nabla_i\partial_jG_0(x,a|E_k^*)
\right)
\Bigg]\dd\mu(x) \;.
\end{aligned}    \label{eq:I2}
\end{eqnarray}
Thus all the quantities below are still written only in terms of the Green functions,
their energy derivatives, their spatial derivatives, and the derivatives of \(\Phi\).

The diagonal second-order contribution is obtained by setting \(m=n\) in the matrix element of
\(\delta H_{\rm mov}^{(2)}\).  From equation \eqref{eq:deltaHmoving}, we find
\begin{eqnarray}
\begin{aligned}
d_m^{(2)}=\sum_k\Bigg\{
\Delta_k^{(2)} & I_{mk}^{(0)}\overline{I_{mk}^{(0)}}D_k^{(0)} 
\cr
&+\Delta_k^{(1)}\left(I_{mk}^{(1)}\overline{I_{mk}^{(0)}}+I_{mk}^{(0)}\overline{I_{mk}^{(1)}}\right)D_k^{(0)}
+\Delta_k^{(1)}I_{mk}^{(0)}\overline{I_{mk}^{(0)}}D_k^{(1)}
\cr
&+E_k^*\left(I_{mk}^{(2)}\overline{I_{mk}^{(0)}}+I_{mk}^{(0)}\overline{I_{mk}^{(2)}}+I_{mk}^{(1)}\overline{I_{mk}^{(1)}}\right)D_k^{(0)}
\cr
&+E_k^*\left(I_{mk}^{(1)}\overline{I_{mk}^{(0)}}+I_{mk}^{(0)}\overline{I_{mk}^{(1)}}\right)D_k^{(1)}
+E_k^*I_{mk}^{(0)}\overline{I_{mk}^{(0)}}D_k^{(2)}
\Bigg\} \;.
\end{aligned} \label{eq:d2moving}
\end{eqnarray}
The first line immediately reduces to \(\Delta_m^{(2)}\), because
\[
I_{mk}^{(0)}\overline{I_{mk}^{(0)}}D_k^{(0)}
=N_k\delta_{mk}\frac{1}{N_k}=\delta_{mk} \;.
\]
The remaining terms are the genuine second-order change of the diagonal matrix element caused by the
motion of the instantaneous eigenprojector.  In the time-dependent perturbation equations of the paper,
only \(d_m^{(1)}\) is needed to form \(\dot c_m^{(2)}\), but the full \(d_m^{(2)}\) is the expression above.

Now take \(m\neq n\).  The zeroth-projector terms vanish off diagonal, but the first- and second-order
projector variations do not.  Starting again from (M4), and using
\(I_{mk}^{(0)}=\sqrt{N_k}\delta_{mk}\), one obtains
\begin{eqnarray}
\begin{aligned}
R_{mn}^{(2)}=
\frac{\Delta_n^{(1)}}{\sqrt{N_n}} I_{mn}^{(1)} &
+\frac{\Delta_m^{(1)}}{\sqrt{N_m}} \overline{I_{nm}^{(1)}}
\cr
&+\frac{E_n^*}{\sqrt{N_n}} I_{mn}^{(2)}
+\frac{E_m^*}{\sqrt{N_m}} \overline{I_{nm}^{(2)}}
\cr
&+E_n^*D_n^{(1)}\sqrt{N_n}\,I_{mn}^{(1)}
+E_m^*D_m^{(1)}\sqrt{N_m}\,\overline{I_{nm}^{(1)}}
\cr
&+\sum_k E_k^*D_k^{(0)} I_{mk}^{(1)}\overline{I_{nk}^{(1)}} \;,
\end{aligned} \label{eq:R2moving}
\end{eqnarray}
for $m\neq n$. If we substitute the explicit first- and second-order energy shifts \eqref{eq:Delta1-moving} and \eqref{eq:Delta2-moving} into
equations \eqref{eq:I1} and \eqref{eq:I2}, then \eqref{eq:d2moving} and \eqref{eq:R2moving} become completely expressed in terms of the prescribed motion
\(s(t)\), the tangent \(v^i\), the covariant acceleration \(b^i\), the Green function \(G_0\), and the
spatial/energy derivatives of \(\Phi\).  This is the full second-order moving-delta analogue of the
fixed-center expansion.

We summarize the off-diagonal transition matrix elements obtained above in
the following proposition.

\begin{proposition}[Off-diagonal matrix elements for a moving delta center]
Let \(m\neq n\), and let
\[
    R_{mn}^{(r)}(t)
    =
    \langle \psi_m|\delta H^{(r)}(t)|\psi_n\rangle,
    \qquad r=1,2,
\]
be the off-diagonal matrix elements of the first- and second-order moving
center perturbations. Define
\[
    N_k:=-\Phi_E(E_k^*,a)>0,
\]
and, for \(r=0,1,2\),
\[
    I_{mk}^{(r)}(t)
    :=
    \int_{\mathcal M}
    \overline{\psi_m(x)}\,G_k^{(r)}(t;x)\,d\mu(x),
\]
with
\[
    I_{mk}^{(0)}=\sqrt{N_k}\,\delta_{mk}.
\]
Then the first-order off-diagonal matrix element is
\[
    R_{mn}^{(1)}(t)
    =
    \frac{E_n^*}{\sqrt{N_n}}\,I_{mn}^{(1)}(t)
    +
    \frac{E_m^*}{\sqrt{N_m}}\,
    \overline{I_{nm}^{(1)}(t)}.
\]
Equivalently, in terms of the Green function,
\[
\begin{aligned}
R_{mn}^{(1)}(t)
={}&
E_n^*
\frac{1}{\sqrt{-\Phi_E(E_n^*,a)}}
\int_{\mathcal M}\overline{\psi_m(x)}
\left[
    \Delta_n^{(1)}(t)\partial_EG_0(x,a|E_n^*)
    +s(t)v^i\partial_iG_0(x,a|E_n^*)
\right]d\mu(x)
\\
&+
E_m^*
\frac{1}{\sqrt{-\Phi_E(E_m^*,a)}}
\int_{\mathcal M}
\overline{
\left[
    \Delta_m^{(1)}(t)\partial_EG_0(y,a|E_m^*)
    +s(t)v^i\partial_iG_0(y,a|E_m^*)
\right]}
\psi_n(y)d\mu(y).
\end{aligned}
\]
The second-order off-diagonal matrix element is
\[
\begin{aligned}
R_{mn}^{(2)}(t)
={}&
\frac{\Delta_n^{(1)}(t)}{\sqrt{N_n}}\,I_{mn}^{(1)}(t)
+
\frac{\Delta_m^{(1)}(t)}{\sqrt{N_m}}\,
\overline{I_{nm}^{(1)}(t)}
\\
&+
\frac{E_n^*}{\sqrt{N_n}}\,I_{mn}^{(2)}(t)
+
\frac{E_m^*}{\sqrt{N_m}}\,
\overline{I_{nm}^{(2)}(t)}
\\
&+
E_n^*D_n^{(1)}(t)\sqrt{N_n}\,I_{mn}^{(1)}(t)
+
E_m^*D_m^{(1)}(t)\sqrt{N_m}\,
\overline{I_{nm}^{(1)}(t)}
\\
&+
\sum_k E_k^*D_k^{(0)}
I_{mk}^{(1)}(t)\overline{I_{nk}^{(1)}(t)} .
\end{aligned}
\]
Here
\[
    D_k^{(0)}=\frac1{N_k},
    \qquad
    D_k^{(1)}(t)=\frac{Q_k^{(1)}(t)}{(\Phi_E^k)^2},
\]
and the quantities \(\Delta_k^{(1)}(t)\), and \(\Delta_k^{(2)}(t)\) are given by Proposition \ref{movingpoleshifts}.
Thus the first-order off-diagonal transition is generated by the first
variation of the moving eigenprojector, while the second-order term receives
contributions from the second variation of the Green function, the first
variation of the normalization factor, and the quadratic product of two
first-order Green-function variations.
\end{proposition}

\section{Representative examples}
\label{sec:examples}

The examples in this section serve two purposes. First, they illustrate how the abstract formulas of Sections~\ref{sec:strength} and \ref{sec:moving} are used in explicit models. Second, they show how degeneracies are handled in practice. In the sphere and isotropic-oscillator examples, the free spectrum is degenerate; the point interaction singles out a coupled sector, while the orthogonal degenerate subspaces remain unaffected or enter only through selection rules. Thus these examples are not merely direct substitutions into the simple-eigenvalue formulas, but adapted applications of the same projection-based method to degenerate systems.

\subsection{Smooth sinusoidal modulation of the renormalized strength}

The delta-kick calculation above is singular in time.  A smoother and physically more standard test of the same formalism is obtained by taking the renormalized parameter to oscillate sinusoidally:
\begin{equation}
\eta(t)=\epsilon\cos(\Omega t) \;.
\label{SSM1}    
\end{equation}
From equation \eqref{eq:Delta1-mu} given in Proposition \ref{Proppoleshifts}, we obtain
\begin{eqnarray}
\Delta_k^{(1)}(t)=\mathcal D_k\epsilon\cos(\Omega t),
\qquad
\mathcal D_k:=-\frac{\Phi_\mu^k}{\Phi_E^k} \;.
\label{SSM2}    
\end{eqnarray}
The diagonal first-order coefficient is the same quantity,
\begin{equation}
d_k^{(1)}(t)=\Delta_k^{(1)}(t)=\mathcal D_k\epsilon\cos(\Omega t) \;.
\label{SSM3}    
\end{equation}
This is the smooth analogue of the coefficient used in the delta-kick calculation.  

Similarly, from equation \eqref{eq:Delta2-mu}, we obtain
\begin{equation}
\Delta_k^{(2)}(t)=\mathcal C_k\epsilon^2\cos^2(\Omega t) \;,
\label{SSM4}    
\end{equation}
with
\begin{equation}
\mathcal C_k:=-\frac{1}{2\Phi_E^k}
\left[
\Phi_{EE}^k\mathcal D_k^2
+\Phi_{\mu\mu}^k
\right] \;.
\label{SSM5}    
\end{equation}
For the sinusoidal driving, equation \eqref{eq:Rmn-first-eta} becomes
\begin{equation}
R_{mn}^{(1)}(t)=\epsilon\mathcal V_{mn}\cos(\Omega t) \;,
\label{SSM6}    
\end{equation}
where
\begin{eqnarray}
\begin{aligned}
\mathcal V_{mn}=&\Bigg[
E_n^*
\frac{\int_\calM \partial_EG_0(x,a|E_n^*)\psi_m(x)\dd\mu(x)}
{\left(\partial_EG_0(a,a|E_n^*)\right)^{3/2}}
\cr
&\qquad +E_m^*
\frac{\int_\calM \partial_EG_0(a,x|E_m^*)\psi_n(x)\dd\mu(x)}
{\left(\partial_EG_0(a,a|E_m^*)\right)^{3/2}}
\Bigg]
\partial_\mu G_0(a,a|\mu_0) \;.
\end{aligned}
\label{SSM7}    
\end{eqnarray}
Assume the system is initially in the state \(\psi_l\), so \(c_n^{(0)}=\delta_{nl}\). 
Equation \eqref{eq:first-amplitude} gives, for \(m \neq l\),
\begin{eqnarray}
\begin{aligned}
c_m^{(1)}(t)
&=-\frac{\ii}{\hbar}\int_0^t
R_{ml}^{(1)}(t')\e^{\ii\omega_{ml}t'}\dd t'
\cr
&=-\frac{\ii\epsilon}{\hbar}\mathcal V_{ml}
\int_0^t\e^{\ii\omega_{ml}t'}\cos(\Omega t')\dd t' \;.
\end{aligned}
\label{SSM8}    
\end{eqnarray}
Expanding the cosine $\cos(\Omega t')=\frac12\left(\e^{\ii\Omega t'}+\e^{-\ii\Omega t'}\right)$, we obtain    
\begin{eqnarray}
\begin{aligned}
\int_0^t\e^{\ii\omega_{ml}t'}\cos(\Omega t')\dd t'
=&\frac12\int_0^t\e^{\ii(\omega_{ml}+\Omega)t'}\dd t'
+\frac12\int_0^t\e^{\ii(\omega_{ml}-\Omega)t'}\dd t'
\cr
=&\frac12\left[
\frac{\e^{\ii(\omega_{ml}+\Omega)t}-1}{\ii(\omega_{ml}+\Omega)}
+
\frac{\e^{\ii(\omega_{ml}-\Omega)t}-1}{\ii(\omega_{ml}-\Omega)}
\right].
\end{aligned}    
\end{eqnarray}
Therefore
\begin{eqnarray}
\begin{aligned}
c_m^{(1)}(t)=&-\frac{\ii\epsilon}{2\hbar}\mathcal V_{ml}
\left[
\frac{\e^{\ii(\omega_{ml}+\Omega)t}-1}{\ii(\omega_{ml}+\Omega)}
+
\frac{\e^{\ii(\omega_{ml}-\Omega)t}-1}{\ii(\omega_{ml}-\Omega)}
\right].
\end{aligned}
\label{SSM9}    
\end{eqnarray}
Equivalently, the amplitude contains the two resonance denominators \(\omega_{ml}+\Omega\) and \(\omega_{ml}-\Omega\).  If one of them vanishes, the corresponding fraction must be read in the limiting sense:
\begin{equation}
\lim_{\nu\to0}\frac{\e^{\ii\nu t}-1}{\ii\nu}=t.
\label{SSM10}    
\end{equation}
This is the smooth-driving analogue of the delta-kick result.  The delta-kick samples the phase at one instant; the sinusoidal drive accumulates phase and produces resonance denominators.

At second order we have
\begin{equation}
R_{ml}^{(2)}(t)=\mathcal W_{ml}\bigl(\eta(t)\bigr)^2
=\epsilon^2\cos^2(\Omega t) \mathcal W_{ml} \;,
\label{SSM11}    
\end{equation}
where the coefficient $\mathcal{W}_{ml}$ is obtained from the same second-order kernel already derived above and it is exactly the off-diagonal matrix element of the second-order Hamiltonian expansion \eqref{eq:H2}
\begin{eqnarray}
\begin{aligned}
     \mathcal{W}_{ml}= \sum_k \Bigg[
     \mathcal{C}_k \langle \psi_m|\bigl(P_k+ & E_k^*P_{k,E}\bigr) |\psi_l \rangle  \\ &  +\mathcal{D}_{k}^{2} \; \langle \psi_m|\left(P_{k,E}
        +\frac12E_k^* P_{k,EE} \right)|\psi_l\rangle 
    \Bigg] \;.
    \end{aligned}
\end{eqnarray}
The full kernels of \(P_{k,E}\) and \(P_{k,EE}\) have already been written explicitly in terms of Green functions and the matrix $\Phi$.

The transition amplitudes $c_m$ can be found by integrating equation \eqref{eq:second-order-transition-equation}, and the second-order contribution has the following form 
\begin{eqnarray}
c_m^{(2)}(t)=
c_{m,\mathrm{iter}}^{(2)}(t)
+c_{m,\mathrm{phase}}^{(2)}(t)
+c_{m,\mathrm{direct}}^{(2)}(t) \;.
\label{SSM14}    
\end{eqnarray}
The three labels indicate the origin of each second-order contribution.
The term \(c_{m,\mathrm{iter}}^{(2)}\) is generated by the time-ordered
iteration of two first-order off-diagonal matrix elements and therefore
contains the intermediate first-order amplitudes. The term
\(c_{m,\mathrm{phase}}^{(2)}\) comes from the first-order correction to the
diagonal dynamical phase, which multiplies the first-order off-diagonal
transition matrix element. Finally, \(c_{m,\mathrm{direct}}^{(2)}\) is due
to the genuinely second-order off-diagonal perturbation \(R^{(2)}\).
Thus the notation separates the three distinct mechanisms that contribute
to the same second-order transition amplitude. We now compute all three explicitly for \(\eta(t)=\epsilon\cos\Omega t\).

It is useful to introduce
\begin{equation}
\mathcal I(\nu;t):=\int_0^t\e^{\ii\nu s}\dd s
=\frac{\e^{\ii\nu t}-1}{\ii\nu},
\qquad
\mathcal I(0;t)=t \;. 
\label{SSM15}    
\end{equation}
Then
\begin{equation}
\mathcal F_{ml}(t):=\int_0^t\e^{\ii\omega_{ml}s}\cos(\Omega s)\dd s
=\frac12\left[\mathcal I(\omega_{ml}+\Omega;t)+\mathcal I(\omega_{ml}-\Omega;t)\right] \;.
\label{SSM16}    
\end{equation}
The first-order coefficient is simply
\begin{equation}
c_m^{(1)}(t)=-\frac{\ii\epsilon}{\hbar}\mathcal V_{ml}\mathcal F_{ml}(t) \;.
\label{SSM17}    
\end{equation}
The first term in Equation \eqref{SSM14} is
\begin{eqnarray}
\begin{aligned}
c_{m,\mathrm{iter}}^{(2)}(t)
=&-\frac{\ii}{\hbar}\sum_n \int_0^t
R_{mn}^{(1)}(s)\e^{\ii\omega_{mn}s}c_n^{(1)}(s)\dd s
\cr
=&-\frac{\ii}{\hbar}\sum_n \int_0^t
\epsilon\mathcal V_{mn}\cos(\Omega s)\e^{\ii\omega_{mn}s}
\left[-\frac{\ii\epsilon}{\hbar}\mathcal V_{nl}\mathcal F_{nl}(s)\right]\dd s \;.
\end{aligned}    
\end{eqnarray}
Then, 
\begin{equation}
 c_{m,\mathrm{iter}}^{(2)}(t)
=-\frac{\epsilon^2}{\hbar^2}
\sum_n \mathcal V_{nl}\mathcal V_{mn}\,
\mathcal J_{mnl}(t) \;,
\label{SSM18}    
\end{equation}
where
\begin{equation}
\mathcal J_{mnl}(t):=\int_0^t
\e^{\ii\omega_{mn}s}\cos(\Omega s)\mathcal F_{nl}(s)\dd s \;.
\label{SSM19}    
\end{equation}
The explicit form of $\mathcal{F}_{nl}$ can be found
\begin{equation}
\mathcal F_{nl}(s)=\frac12\sum_{\sigma=\pm1}
\frac{\e^{\ii(\omega_{nl}+\sigma\Omega)s}-1}
{\ii(\omega_{nl}+\sigma\Omega)} \;.    
\end{equation}
Using $\cos(\Omega s)=\frac12\sum_{\rho=\pm1}\e^{\ii\rho\Omega s}$, we get
\begin{eqnarray}
\begin{aligned}
\mathcal J_{mnl}(t)
=&\frac14\sum_{\rho=\pm1}\sum_{\sigma=\pm1}
\frac{1}{\ii(\omega_{nl}+\sigma\Omega)}
\cr
&\times\int_0^t\left[
\e^{\ii[\omega_{nl}+\rho\Omega+\omega_{mn}+\sigma\Omega]s}
-\e^{\ii(\omega_{mn}+\rho\Omega)s}
\right]\dd s \;.
\end{aligned}    
\end{eqnarray}
Since \(\omega_{mn}+\omega_{nl}=\omega_{ml}\), this becomes
\begin{eqnarray}
\mathcal J_{mnl}(t)=
\frac14\sum_{\rho=\pm1}\sum_{\sigma=\pm1}
\frac{
\mathcal I(\omega_{ml}+(\rho+\sigma)\Omega;t)
-\mathcal I(\omega_{mn}+\rho\Omega;t)
}{\ii(\omega_{nl}+\sigma\Omega)} \;.
\label{SSM20}    
\end{eqnarray}
This is the fully expanded smooth-driving version of the two-step transition
\(l \to n \to m\).

The second term in Equation \eqref{SSM14} is
\begin{eqnarray}
\begin{aligned}
c_{m,\mathrm{phase}}^{(2)}(t)=&
-\frac{\ii}{\hbar}\int_0^t
R_{ml}^{(1)}(s)\e^{\ii\omega_{ml}s}
\left[\frac{\ii}{\hbar}\int_0^s(d_m^{(1)}(r)-d_l^{(1)}(r))\dd r\right]\dd s \;.
\end{aligned}    
\end{eqnarray}
Using \eqref{SSM3}, we have $d_m^{(1)}(r)-d_l^{(1)}(r)
=\epsilon(\mathcal D_m-\mathcal D_l)\cos(\Omega r)$. Therefore
\begin{equation}
\int_0^s(d_m^{(1)}(r)-d_l^{(1)}(r))\dd r
=\epsilon(\mathcal D_m-\mathcal D_l)\frac{\sin(\Omega s)}{\Omega} \;.
\label{SSM21}    
\end{equation}
Substituting into the phase term gives
\begin{eqnarray}
\begin{aligned}
c_{m,\mathrm{phase}}^{(2)}(t)
=&\frac{\epsilon^2}{\hbar^2}
\mathcal V_{ml}(\mathcal D_m-\mathcal D_l)
\int_0^t\e^{\ii\omega_{ml}s}\cos(\Omega s)
\frac{\sin(\Omega s)}{\Omega}\dd s \;.
\end{aligned}    
\end{eqnarray}
Since $\cos(\Omega s)\sin(\Omega s)=\frac12\sin(2\Omega s)$, we get
\begin{eqnarray}
\begin{aligned}
\int_0^t\e^{\ii\omega s}\cos(\Omega s)\frac{\sin(\Omega s)}{\Omega}\dd s
&=\frac{1}{2\Omega}\int_0^t\e^{\ii\omega s}\sin(2\Omega s)\dd s
\cr
&=\frac{1}{4\ii\Omega}
\left[\mathcal I(\omega+2\Omega;t)-\mathcal I(\omega-2\Omega;t)\right].
\end{aligned}    
\end{eqnarray}
Thus
\begin{equation}
 c_{m,\mathrm{phase}}^{(2)}(t)
=\frac{\epsilon^2}{\hbar^2}
\mathcal V_{ml}(\mathcal D_m-\mathcal D_l)
\frac{
\mathcal I(\omega_{ml}+2\Omega;t)-\mathcal I(\omega_{ml}-2\Omega;t)
}{4\ii\Omega} \;.
\label{SSM22}    
\end{equation}
This term vanishes if the two channels have the same first-order diagonal shift, \(\mathcal D_m=\mathcal D_l\).

The third term in Equation \eqref{SSM14} is
\begin{equation}
 c_{m,\mathrm{direct}}^{(2)}(t)
=-\frac{\ii}{\hbar}\int_0^t
R_{ml}^{(2)}(s)\e^{\ii\omega_{ml}s}\dd s \;.    
\end{equation}
Using \eqref{SSM11}, we find
\begin{eqnarray}
\begin{aligned}
 c_{m,\mathrm{direct}}^{(2)}(t)
 = & -\frac{\ii\epsilon^2}{\hbar}\mathcal W_{ml}
\int_0^t\e^{\ii\omega_{ml}s}\cos^2(\Omega s)\dd s \cr = &   -\frac{\ii\epsilon^2}{\hbar}\mathcal W_{ml}
\left[
\frac12\mathcal I(\omega_{ml};t)
+\frac14\mathcal I(\omega_{ml}+2\Omega;t)
+\frac14\mathcal I(\omega_{ml}-2\Omega;t)
\right] \;.   
\end{aligned}
\label{SSM23}
\end{eqnarray}
Combining the three pieces gives, for \(m\neq l\),
\begin{eqnarray}
\begin{aligned}
c_m^{(2)}(t)=&
-\frac{\epsilon^2}{\hbar^2}
\sum_l\mathcal V_{mn}\mathcal V_{nl}\mathcal J_{mnl}(t)
\cr
&+\frac{\epsilon^2}{\hbar^2}
\mathcal V_{ml}(\mathcal D_m-\mathcal D_l)
\frac{
\mathcal I(\omega_{ml}+2\Omega;t)-\mathcal I(\omega_{ml}-2\Omega;t)
}{4\ii\Omega}
\cr
&-\frac{\ii\epsilon^2}{\hbar}\mathcal W_{ml}
\left[
\frac12\mathcal I(\omega_{ml};t)
+\frac14\mathcal I(\omega_{ml}+2\Omega;t)
+\frac14\mathcal I(\omega_{ml}-2\Omega;t)
\right].
\end{aligned}
\label{SSM24}    
\end{eqnarray}
Together with \eqref{SSM9}, this is the explicit first- and second-order response to a smooth sinusoidal modulation of the renormalized parameter.

\subsection{A renormalized point interaction moving along a great circle on $S_R^2$}

This example applies the time-dependent perturbation theory developed above to a point interaction moving along a geodesic on the two-dimensional sphere $\mathcal{M}=S_R^2$ of radius $R$. 
\begin{equation}
H_0=-\frac{\hbar^2}{2M}\Delta_{S_R^2} \;,    
\end{equation}
where $M$ denotes the particle mass.  The support of the renormalized delta interaction moves along a great circle and we choose the initial point to be the north pole $N$ (see Fig. \ref{fig:moving-point-sphere-minimal}) and write
\begin{equation}
q(t)=\gamma(s(t)),
\qquad
\gamma(0)=N,
\qquad
|\dot\gamma(0)|=1,
\qquad
s(t)=\varepsilon\cos(\Omega t) \;.
\label{SPH2}    
\end{equation}

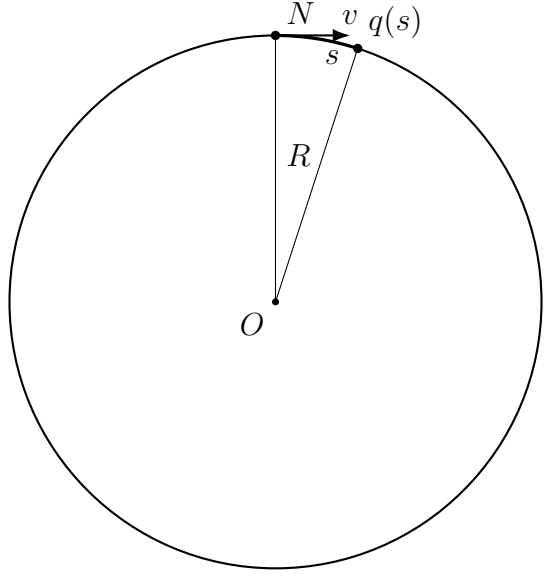
\begin{figure}[t]
\centering
\begin{tikzpicture}[scale=2.2, >=Latex]

    \def\R{1.6}
    \def\ang{72}

    \draw[thick] (0,0) circle (\R);

    \fill (0,0) circle (0.6pt);
    \node[below left] at (0,0) {$O$};

    \coordinate (N) at (0,\R);
    \fill (N) circle (0.8pt);
    \node[above right] at (N) {$N$};

    \coordinate (Q) at ({\R*cos(\ang)},{\R*sin(\ang)});
    \fill (Q) circle (0.8pt);
    \node[above right] at (Q) {$q(s)$};

    \draw[thin] (0,0) -- (N);
    \draw[thin] (0,0) -- (Q);
    \node[right] at ($(0,0)!0.55!(N)$) {$R$};

    \draw[very thick] (90:\R) arc[start angle=90,end angle=\ang,radius=\R];
    \draw[->, thick] (N) -- ++(0.45,0) node[above] {$v$};

    \node at (0.34,1.47) {$s$};

\end{tikzpicture}
\caption{Cross-sectional illustration of the support point \(q(s)\) moving
from the north pole \(N\) along a great circle on \(S_R^2\).}
\label{fig:moving-point-sphere-minimal}
\end{figure}

Because $\gamma$ is a great circle parametrized by arclength,
\begin{equation}
\frac{D\dot\gamma}{Ds}=0 \;.
\label{SPH3}    
\end{equation}
This removes the acceleration term in the covariant Taylor expansion. The eigenvalues of $H_0$ are
\begin{equation}
E_l=\frac{\hbar^2}{2MR^2}l(l+1) \;,
\qquad l=0,1,2,\ldots,
\label{SPH4}    
\end{equation}
with degeneracy $2l+1$.  If $Y_{lm}$ are the usual spherical harmonics normalized on the unit sphere, then the normalized eigenfunctions on $S_R^2$ are
\begin{eqnarray}
\phi_{lm}(\theta,\varphi)=\frac{1}{R}Y_{lm}(\theta,\varphi),
\qquad
\int_{S_R^2}\overline{\phi_{lm}}\phi_{l'm'}\,d\Omega_R=\delta_{ll'}\delta_{mm'} \;.
\label{SPH5}    
\end{eqnarray}
Here $d\Omega_R=R^2\sin\theta\,d\theta d\varphi$. At the north pole $N$ one has
\begin{eqnarray}
\phi_{lm}(N)=0\quad(m\neq0),
\qquad
\phi_{l0}(N)=\frac{1}{R}\sqrt{\frac{2l+1}{4\pi}} \;.
\label{SPH7}    
\end{eqnarray}
This is the first place where the degeneracy is resolved by the point interaction: at the north pole the point interaction sees only the zonal component $m=0$ of each free eigenspace. The eigenfunctions of $H_0$ that vanish at the position of the delta center are unaffected by the point interaction \cite{annals23}.
The eigenfunction expansion of the free Green function is given by
\begin{eqnarray}
G_0(x,q|E)=\sum_{l=0}^{\infty}\sum_{m=-l}^{l}
\frac{\phi_{lm}(x)\overline{\phi_{lm}(q)}}{E_l-E} \;.
\label{SPH8}    
\end{eqnarray}
The renormalized function $\Phi$ may be written spectrally as
\begin{eqnarray}
\Phi(E,q)=\sum_{l=0}^{\infty}\sum_{m=-l}^{l}
|\phi_{lm}(q)|^2
\left(\frac{1}{E_l+\mu}-\frac{1}{E_l-E}\right) \;.
\label{SPH9}    
\end{eqnarray}
By the addition theorem for spherical harmonics \cite{HassaniMathPhys}, we obtain
\begin{eqnarray}
\sum_{m=-l}^{l}|\phi_{lm}(q)|^2 = \frac{1}{R^2} \sum_{m=-l}^{l}|Y_{lm}(q)|^2=\frac{2l+1}{4\pi R^2} \;,
\label{SPH10}    
\end{eqnarray}
independently of point $q$ on the sphere.  Substituting this into \eqref{SPH9} gives
\begin{equation}
\Phi(E,q)=\Phi(E)=
\sum_{l=0}^{\infty}\frac{2l+1}{4\pi R^2}
\left(\frac{1}{E_l+\mu}-\frac{1}{E_l-E}\right).
\label{SPH12}    
\end{equation}
Therefore the derivatives of $\Phi$ with respect to the support point vanish:
\begin{eqnarray}
\nabla_i\Phi(E,q)=0,
\qquad
\nabla_i\nabla_j\Phi(E,q)=0 \;.
\label{SPH13}    
\end{eqnarray}
The energy roots of the point-interaction sector are determined by $\Phi(E_\alpha^*)=0$.
Here the abstract state label \(\alpha\) denotes the pair of angular-momentum quantum numbers on the sphere, $\alpha=(l,m)$, $l=0,1,2,\ldots$, and  $m=-l,\ldots,l$. 
The corresponding renormalized point-interaction eigenfunction centered at the north pole $N$ is
\begin{eqnarray}
\psi_\alpha(x)=\frac{G_0(x,N|E_\alpha^*)}{\sqrt{N_\alpha}}
=\frac{1}{\sqrt{N_\alpha}}
\sum_{l=0}^{\infty}
\frac{\phi_{l0}(x)\overline{\phi_{l0}(N)}}{E_l-E_\alpha^*} \;.
\label{SPH16}    
\end{eqnarray}
The normalization factor is given by
\begin{equation}
N_\alpha=-\Phi_E(E_\alpha^*)=
\sum_{l=0}^{\infty}\frac{2l+1}{4\pi R^2}\frac{1}{(E_l-E_\alpha^*)^2} \;.
\label{SPH15}    
\end{equation}
This state $\psi_\alpha$ is purely zonal: it contains only $m=0$ modes.

As the point interaction moves with $q(t)=\gamma(s(t))$, $E_{\alpha}^{*}$ flows according to
\begin{equation}
\Phi(E_\alpha^*(t),q(t))=0 \;.
\label{SPH17}    
\end{equation}
Using the same moving-center expansion, $E_\alpha^*(t)=E_\alpha^*+\Delta_\alpha^{(1)}(t)+\Delta_\alpha^{(2)}(t)+\cdots$, the first-order Taylor term in equation \eqref{SPH17} gives
\begin{eqnarray}
0=\Phi_E(E_\alpha^*)\Delta_\alpha^{(1)}(t)
+s(t)\dot\gamma^i(0)\nabla_i\Phi(E_\alpha^*,N) \;.
\label{SPH19}    
\end{eqnarray}
But \eqref{SPH13} gives $\nabla_i\Phi=0$, so $\Delta_\alpha^{(1)}(t)=0$.

At second order one obtains
\begin{eqnarray}
\begin{aligned}
0={}&\Phi_E \Delta_\alpha^{(2)}
+\frac12\Phi_{EE}\left(\Delta_\alpha^{(1)}\right)^2
+s(t)\dot\gamma^i \partial_i \Phi_{E} \Delta_\alpha^{(1)}
\cr
&+\frac{s(t)^2}{2}\dot\gamma^i\dot\gamma^j\nabla_i\nabla_j\Phi
+\frac{s(t)^2}{2}\frac{D\dot\gamma^i}{Ds}\nabla_i\Phi \;.
\end{aligned}
\label{SPH21}    
\end{eqnarray}
On the sphere with a great-circle trajectory,
\begin{equation}
\Delta_\alpha^{(1)}=0,
\qquad
\nabla_i\Phi=0,
\qquad
\nabla_i\nabla_j\Phi=0,
\qquad
\frac{D\dot\gamma^i}{Ds}=0 \;.
\label{SPH22}    
\end{equation}
Therefore every term except $\Phi_E \Delta_\alpha^{(2)}$ vanishes, and since $\Phi_E\neq0$ at a simple pole, $\Delta_\alpha^{(2)}(t)=0$.
In fact, by rotational symmetry the point-interaction energy levels are independent of the center position to all orders.  The important point is that this does not mean that the time-dependent problem is trivial.  The eigenvalues do not move, but the eigenfunctions and projectors do.

Let us choose the great circle through the north pole with $\varphi=0$.  Since $s=R\theta$ near the north pole,
\begin{equation}
\frac{d}{ds}=\frac{1}{R}\frac{d}{d\theta}\quad(\varphi=0).
\label{SPH24}    
\end{equation}
The Green function centered at the moving point is
\begin{equation}
G_0(x,\gamma(s)|E)=
G_0(x,N|E)
+s\,\partial_sG_0(x,N|E)
+\frac{s^2}{2}\partial_s^2G_0(x,N|E)+\cdots.
\label{SPH25}    
\end{equation}
From \eqref{SPH8},
\begin{equation}
\partial_sG_0(x,N|E)= \partial_sG_0(x,q(s)|E)|_{s=0} = 
\sum_{l,m}\frac{\phi_{lm}(x)\partial_s\overline{\phi_{lm}(N)}}{E_l-E} \;.    
\end{equation}
The required north-pole derivatives are obtained from the small-$\theta$ behavior of the spherical harmonics.

We use the standard Condon--Shortley convention for the spherical harmonics \cite{HassaniMathPhys},
namely
\begin{equation}
    Y_{lm}(\theta,\varphi)
    =
    \left[
    \frac{2l+1}{4\pi}
    \frac{(l-m)!}{(l+m)!}
    \right]^{1/2}
    P_l^m(\cos\theta)e^{im\varphi},
\end{equation}
where
\begin{equation}
    P_l^m(x)
    =
    (-1)^m(1-x^2)^{m/2}
    \frac{d^m}{dx^m}P_l(x) \;.
\end{equation}
For \(m=1\), the Condon--Shortley convention gives
$P_l^1(\cos\theta)
    =
    -\sin\theta\,P_l'(\cos\theta)$. Using $\sin\theta=\theta+O(\theta^3)$, and 
$P_l'(1)=\frac{l(l+1)}2$, we obtain, as \(\theta\to0\),
\begin{equation}
    P_l^1(\cos\theta)
    =
    -\frac{l(l+1)}2\,\theta+O(\theta^3) \;.
\end{equation}
Therefore
\begin{equation}
    Y_{l1}(\theta,0)
    =
    -\frac12
    \left[
    \frac{(2l+1)l(l+1)}{4\pi}
    \right]^{1/2}
    \theta
    +O(\theta^3) \;.
\end{equation}
Since \(\phi_{lm}=Y_{lm}/R\) and
\(\partial_s=R^{-1}\partial_\theta\), it follows that
\begin{equation}
    \partial_s\overline{\phi_{l1}(N)}
    =
    -\frac1{2R^2}
    \left[
    \frac{(2l+1)l(l+1)}{4\pi}
    \right]^{1/2} \;.
\end{equation}
For \(m=-1\), the Condon--Shortley relation $Y_{l,-m}=(-1)^m\overline{Y_{lm}}$ implies
$Y_{l,-1}=-\overline{Y_{l1}}$, and $\overline{Y_{l,-1}}=-Y_{l1}$. Hence
\begin{equation}
    \partial_s\overline{\phi_{l,-1}(N)}
    =
    \frac1{2R^2}
    \left[
    \frac{(2l+1)l(l+1)}{4\pi}
    \right]^{1/2} \;.
\end{equation}
For \(m=0\), the first derivative vanishes because
\(Y_{l0}\) is proportional to \(P_l(\cos\theta)\), whose
\(\theta\)-derivative contains a factor \(\sin\theta\). For
\(|m|\geq2\), one has
\(P_l^{|m|}(\cos\theta)=O(\theta^{|m|})\), so the first derivative at
\(\theta=0\) also vanishes. Thus
\begin{equation}
    \partial_s\overline{\phi_{lm}(N)}=0,
    \qquad m\neq \pm1 \;.
\end{equation}
Defining
\begin{equation}
    B_l
    :=
    \frac1{2R^2}
    \sqrt{
    \frac{(2l+1)l(l+1)}{4\pi}} \;, \label{SPH28}    
\end{equation}
we finally obtain
\begin{equation}
    \partial_s\overline{\phi_{l1}(N)}=-B_l,
    \qquad
    \partial_s\overline{\phi_{l,-1}(N)}=B_l,
    \qquad
    \partial_s\overline{\phi_{lm}(N)}=0
    \quad (m\neq\pm1) \;. \label{SPH27}    
\end{equation}
Hence
\begin{equation}
\partial_sG_0(x,N|E)=
\sum_{l=1}^{\infty}\frac{B_l}{E_l-E}
\left(\phi_{l,-1}(x)-\phi_{l1}(x)\right) \;.
\label{SPH29}    
\end{equation}
This equation gives the first selection rule: the first derivative of a zonal Green function at the north pole lives entirely in the $m=\pm1$ sector.

For the second derivative we similarly have
\begin{equation}
\partial_s^2G_0(x,N|E)=
\sum_{l,m}\frac{\phi_{lm}(x)\partial_s^2\overline{\phi_{lm}(N)}}{E_l-E} \;.
\label{SPH30}    
\end{equation}
To determine which terms contribute, we use the local behaviour of the
spherical harmonics near the north pole. Since
\begin{equation}
    Y_{lm}(\theta,\varphi)\sim \theta^{|m|}e^{im\varphi}
    \qquad
    (\theta\to0) \;,
\end{equation}
the second derivative at \(\theta=0\) can be nonzero only for
\(|m|=0\) or \(|m|=2\). The \(m=\pm1\) terms are linear in \(\theta\) to
leading order and therefore have vanishing second derivative at
\(\theta=0\), while all terms with \(|m|\geq3\) vanish to order at least
\(\theta^3\).

The $m=0$ piece is already enough to see the second-order return to the zonal sector.  Since
\begin{equation}
Y_{l0}(\theta)=\sqrt{\frac{2l+1}{4\pi}}P_l(\cos\theta),
\qquad
\frac{d^2}{d\theta^2}P_l(\cos\theta)\bigg|_{\theta=0}=-\frac{l(l+1)}{2} \;,
\label{SPH31}    
\end{equation}
one obtains
\begin{equation}
\partial_s^2\overline{\phi_{l0}(N)}
=-\frac{l(l+1)}{2R^3}\sqrt{\frac{2l+1}{4\pi}} \;.
\label{SPH32}    
\end{equation}
The explicit $m=0$ part is
\begin{equation}
\left[\partial_s^2G_0(x,N|E)\right]_{m=0}
=-\sum_{l=0}^{\infty}
\frac{l(l+1)}{2R^3}\sqrt{\frac{2l+1}{4\pi}}
\frac{\phi_{l0}(x)}{E_l-E} \;.
\label{SPH34}    
\end{equation}
There are also $m=\pm2$ second-derivative components, because $Y_{l,\pm2}\sim \theta^2e^{\pm2i\varphi}$ near the north pole.  
For \(m=2\),
\begin{equation}
    P_l^2(\cos\theta)
    =
    \sin^2\theta\,P_l''(\cos\theta) \;.
\end{equation}
Since
\begin{equation}
    \sin^2\theta=\theta^2+O(\theta^4),
    \qquad
    P_l''(1)=\frac{(l+2)(l+1)l(l-1)}{8} \;,
\end{equation}
we obtain
\begin{equation}
    \partial_s^2\overline{\phi_{l2}(N)}
    =
    C_l,
    \qquad
    C_l:=
    \frac{1}{4R^3}
    \sqrt{
    \frac{(2l+1)(l-1)l(l+1)(l+2)}{4\pi}
    } \;.
\end{equation}
Moreover, by the Condon--Shortley relation
\(Y_{l,-2}=\overline{Y_{l2}}\), so $\partial_s^2\overline{\phi_{l,-2}(N)}=C_l$. Hence
\begin{equation}
    \left[\partial_s^2G_0(x,N|E)\right]_{m=\pm2}
    =
    \sum_{l=2}^{\infty}
    \frac{C_l}{E_l-E}
    \left(
        \phi_{l2}(x)+\phi_{l,-2}(x)
    \right) \;.
\end{equation}
Thus the second derivative contains an $m=0$ part and an $m=\pm2$ part:
\begin{equation}
\partial_s^2G_0(x,N|E)=
\left[\partial_s^2G_0(x,N|E)\right]_{m=0}
+
\left[\partial_s^2G_0(x,N|E)\right]_{m=\pm2} \;.
\label{SPH33}    
\end{equation}
This is the second selection rule: second order contains an $m=0$ component, so it can couple back into the zonal point-interaction sector.

The free level $E_l$ has multiplicity $2l+1$.  The point interaction at $N$ does not couple to all of this degenerate space.  Indeed, for any vector
\begin{equation}
f(x)=\sum_{m=-l}^{l}c_m\phi_{lm}(x) \;,
\label{SPH35}    
\end{equation}
we have $f(N)= c_0 \phi_{l0}(N)$. Therefore the subspace
$\mathcal H_l^{\rm dec}(N)=
\{f\in\operatorname{span}(\phi_{l,-l},\ldots,\phi_{ll})\,:\,f(N)=0\}$ has dimension $2l$ and remains unaffected by the point interaction at $N$.  Equivalently, in the north-pole basis it is spanned by the $m\neq0$ modes.

The one-dimensional complementary direction in the free eigenspace is the zonal direction $m=0$.  The renormalized point-interaction states are not simply the individual $\phi_{l0}$, but the Green-function combinations \eqref{SPH16}.  This is why ordinary non-degenerate perturbation theory cannot be applied blindly to all spherical harmonics.  The correct interpretation is:
\begin{itemize}
\item the point-interaction roots $E_\alpha^*$ are simple in the coupled zonal sector;
\item the unaffected free subspaces are degenerate and should be treated in an adapted basis;
\item first-order motion along the chosen great circle couples the zonal point-interaction sector only to the $m=\pm1$ part of the unaffected free sector;
\item the orthogonal combinations inside the degenerate subspace remain uncoupled at this order.
\end{itemize}

Since the energies do not shift, the first-order perturbation comes entirely from the moving projector, equivalently from the first derivative of the Green function.

Let the initial state be the point-interaction eigenstate \(\psi_\alpha\), and let the final state be a free spherical harmonic \(\phi_{LM}\) with \(M\neq0\). Because the point interaction is initially supported at the north pole, and because $\phi_{LM}(N)=0,\qquad M\neq0$ these spherical harmonics are not affected by the static point interaction. Consequently, they remain eigenfunctions of the static renormalized Hamiltonian \(H(0)\), with the same eigenvalues as for the free Hamiltonian: $H(0)\phi_{LM}=E_L\phi_{LM}$, with $M\neq0$. Thus, although \(\phi_{LM}\) is a free spherical harmonic, it may also be used as an eigenfunction of \(H(0)\) in the off-diagonal perturbation formula.
If
\begin{equation}
    H(s)\psi_\alpha(s)=E_\alpha^*(s)\psi_\alpha(s),
    \qquad
    H(0)\phi_{LM}=E_L\phi_{LM} \;,
\end{equation}
then differentiating the first equation with respect to \(s\), evaluating at
\(s=0\), and taking the inner product with \(\phi_{LM}\) gives
\begin{equation}
    \langle\phi_{LM}|\partial_sH|\psi_\alpha\rangle
    =
    (E_\alpha^*-E_L)
    \langle\phi_{LM}|\partial_s\psi_\alpha\rangle \;,
\end{equation}
for \(\phi_{LM}\perp\psi_\alpha\). Since
\(\delta H^{(1)}(t)=s(t)\partial_sH|_{s=0}\), we obtain
\begin{equation}
\langle \phi_{LM}|\delta H^{(1)}|\psi_\alpha\rangle
=s(t)(E_\alpha^*-E_L)
\langle\phi_{LM}|\partial_s\psi_\alpha\rangle \;.
\label{SPH38}    
\end{equation}
But
\begin{equation}
\partial_s\psi_\alpha(x)=
\frac{\partial_sG_0(x,N|E_\alpha^*)}{\sqrt{N_\alpha}} \;,
\label{SPH39}    
\end{equation}
because $N_\alpha$ is position independent on the sphere. Therefore
\begin{equation}
\langle\phi_{LM}|\partial_s\psi_\alpha\rangle
=\frac{1}{\sqrt{N_\alpha}}
\frac{\partial_s\overline{\phi_{LM}(N)}}{E_L-E_\alpha^*} \;.
\label{SPH40}    
\end{equation}
Substituting this into equation \eqref{SPH38} yields
\begin{equation}
\langle \phi_{LM}|\delta H^{(1)}|\psi_\alpha\rangle
=-\frac{s(t)}{\sqrt{N_\alpha}}\partial_s\overline{\phi_{LM}(N)} \;.
\label{SPH42}    
\end{equation}
Using equation \eqref{SPH27},
\begin{equation}
\langle \phi_{L1}|\delta H^{(1)}|\psi_\alpha\rangle
=-\frac{s(t)B_L}{\sqrt{N_\alpha}} \;,
\qquad
\langle \phi_{L,-1}|\delta H^{(1)}|\psi_\alpha\rangle
=+\frac{s(t)B_L}{\sqrt{N_\alpha}} \;,
\label{SPH43}    
\end{equation}
and all $M\neq\pm1$ matrix elements vanish at first order:
\begin{equation}
\langle \phi_{LM}|\delta H^{(1)}|\psi_\alpha\rangle=0 \;,
\quad(M\neq\pm1) \;.
\label{SPH44}    
\end{equation}
Thus the first-order transition amplitude for $s(t)=\varepsilon\cos\Omega t$ is
\begin{equation}
 c_{LM}^{(1)}(t)=
-\frac{\ii}{\hbar} \mathcal{V}_{LM,\alpha}
\int_0^t \varepsilon\cos(\Omega t')
\exp\left[\frac{\ii}{\hbar}(E_L-E_\alpha^*)t'\right]dt',
\label{SPH45}    
\end{equation}
where
\begin{eqnarray}
\mathcal{V}_{L1,\alpha}=-\frac{B_L}{\sqrt{N_\alpha}},
\qquad
\mathcal{V}_{L,-1,\alpha}=+\frac{B_L}{\sqrt{N_\alpha}},
\qquad
\mathcal{V}_{LM,\alpha}=0\quad(M\neq\pm1).
\label{SPH46}    
\end{eqnarray}
The time integral is
\begin{eqnarray}
\int_0^t\cos(\Omega t')e^{\ii\omega t'}dt'
=\frac12\left[
\frac{e^{\ii(\omega+\Omega)t}-1}{\ii(\omega+\Omega)}
+\frac{e^{\ii(\omega-\Omega)t}-1}{\ii(\omega-\Omega)}
\right],
\label{SPH47}    
\end{eqnarray}
with $\omega=(E_L-E_\alpha^*)/\hbar$.

The first-order transition between two point-interaction states vanishes.  Indeed, $\langle\psi_\beta|\partial_s\psi_\alpha\rangle=0$,
because $\psi_\beta$ is zonal while $\partial_s\psi_\alpha$ lies in the $m=\pm1$ sector.  Therefore
\begin{equation}
\langle\psi_\beta|\delta H^{(1)}|\psi_\alpha\rangle=0 \;.
\label{SPH49}    
\end{equation}

At second order there are two kinds of contributions. First, there is a direct second-order moving-center term, coming from the second derivative of the Green function along the great circle:
\begin{equation}
\delta H^{(2)}_{\beta \alpha,\rm direct}(t)
=\frac{s(t)^2}{2}\langle \psi_\beta|\partial_s^2H|\psi_\alpha\rangle \;.
\label{SPH50}    
\end{equation}
For a final point-interaction state $\psi_\beta$, only the $m=0$ part of $\partial_s^2G_0$ contributes.  In the same fixed north-pole basis,
\begin{eqnarray}
\langle\psi_\beta|\partial_s^2\psi_\alpha\rangle_{m=0}
=\frac{1}{\sqrt{N_\alpha N_\beta}}
\sum_{l=0}^{\infty}
\frac{|\phi_{l0}(N)|^2}{E_l-E_\beta^*}
\left(\frac{-l(l+1)}{2R^2}\right)
\frac{1}{E_l-E_\alpha^*} \;,
\label{SPH51}    
\end{eqnarray}
where we have used 
\begin{equation}
\partial_s^2\overline{\phi_{l0}(N)}
=-\frac{l(l+1)}{2R^2}\overline{\phi_{l0}(N)} \;,
\label{SPH52}    
\end{equation}
which is the same as equation \eqref{SPH32} written in terms of the normalized value $\phi_{l0}(N)$.
The $m=\pm2$ part of $\partial_s^2G_0$ instead couples the point-interaction state to the $m=\pm2$ free sector.

Second, there is the iterated first-order contribution through the $m=\pm1$ free states:
\begin{eqnarray}
\begin{aligned}
 c_{\beta,\rm iter}^{(2)}(t)=
-&\frac{1}{\hbar^2}
\sum_{L=1}^{\infty}\sum_{M=\pm1}
\mathcal{V}_{\beta,LM} \mathcal{V}_{LM,\alpha}
\cr
&\times\int_0^t dt_1\,s(t_1)
\exp\left[\frac{\ii}{\hbar}(E_\beta^*-E_L)t_1\right]
\int_0^{t_1}dt_2\,s(t_2)
\exp\left[\frac{\ii}{\hbar}(E_L-E_\alpha^*)t_2\right] \;.
\end{aligned}
\label{SPH53}    
\end{eqnarray}
Here
\begin{equation}
\mathcal{V}_{LM,\alpha}=-\frac{\partial_s\overline{\phi_{LM}(N)}}{\sqrt{N_\alpha}} \;,
\qquad
\mathcal{V}_{\beta,LM}=-\frac{\partial_s\phi_{LM}(N)}{\sqrt{N_\beta}} \;,
\label{SPH54}    
\end{equation}
with the signs following from the same energy-denominator cancellation as in equation \eqref{SPH42}.

The full second-order amplitude from $\psi_\alpha$ to a point-interaction state $\psi_\beta$ is therefore
\begin{equation}
 c_\beta^{(2)}(t)=c_{\beta,\rm direct}^{(2)}(t)+c_{\beta,\rm iter}^{(2)}(t),
\label{SPH55}    
\end{equation}
where
\begin{eqnarray}
c_{\beta,\rm direct}^{(2)}(t)=
-\frac{\ii}{\hbar}
\int_0^t
\delta H_{\beta\alpha,\rm direct}^{(2)}(t')
\exp\left[\frac{\ii}{\hbar}(E_\beta^*-E_\alpha^*)t'\right]dt' \;.
\label{SPH56}    
\end{eqnarray}
The first-order motion leaves the zonal sector, while second order can return to it either directly through the $m=0$ part of the second derivative or indirectly through two successive $m=\pm1$ transitions.

The sphere example gives the following clean conclusions:
\begin{eqnarray}
\Delta_\alpha^{(1)}(t)=0,
\qquad
\Delta_\alpha^{(2)}(t)=0 \;,
\label{SPH57}    
\end{eqnarray}
not because the perturbation is absent, but because $\Phi(E,q)$ is independent of $q$ by the addition theorem.  The nontrivial dynamics is in the moving eigenfunctions and projectors:
\begin{equation}
\partial_sG_0(x,N|E)\text{ lies in the }m=\pm1\text{ sector},
\label{SPH58}    
\end{equation}
while
\begin{equation}
\partial_s^2G_0(x,N|E)\text{ contains }m=0\text{ and }m=\pm2\text{ sectors}.
\label{SPH59}    
\end{equation}
Consequently, the first-order transition from a point-interaction state goes to the adapted $m=\pm1$ subspace, whereas second order contains both a return to the zonal point-interaction sector and transitions to the $m=\pm2$ sector.  The degeneracy of the free spherical spectrum is therefore not ignored; it is resolved by choosing the basis adapted to the support point and to the tangent direction of the great-circle motion.

\subsection{Two-dimensional isotropic harmonic oscillator with a moving renormalized delta center}
\label{subsec:two-dimensional-oscillator}

Harmonic oscillators perturbed by delta-type interactions have been studied
from several complementary viewpoints.  The one-dimensional harmonic oscillator
with a delta-function perturbation was solved explicitly in Ref.~\cite{Patil2006DeltaHO}, while the statistical properties of the harmonic oscillator with an added delta potential were analyzed in Ref.~\cite{JankeCheng}.  Related models have also appeared in the study of Bose--Einstein condensation, both for attractive impurity centers in one dimension \cite{IoriattiRosaHipolito} and for harmonic traps decorated with Dirac delta interactions \cite{Huncu}.  From a more structural point of view, delta-function perturbations have been treated by path-integral methods \cite{Grosche} and by spectral analysis of delta-type perturbations of the harmonic oscillator \cite{Fassari1}.  The example considered below differs from these static settings in that the delta center is allowed to move in the plane, so that the perturbation is generated by the motion of the renormalized point interaction.

We now work out the two-dimensional isotropic harmonic oscillator with a moving singular perturbation supported at $q(t)=(q_1(t),q_2(t))$.
Only after the vector formulas have been obtained do we specialize to a straight oscillation.  Because the oscillator is isotropic, choosing the direction of this oscillation is a choice of coordinates, not a reduction to a one-dimensional delta interaction.

We use oscillator units $\hbar=m=\omega=1$ and 
\begin{equation}
    H_0=-\frac12\Delta+\frac12|r|^2 \;,
    \label{eq:2dho-H0}
\end{equation}
where $r=(x,y)$. The heat kernel is the product of two one-dimensional Mehler kernels \cite{ThangaveluHermiteLaguerre} and given by 
\begin{equation}
K_s(r,r')=\frac{1}{2\pi\sinh s}
\exp\left[-\frac{(|r|^2+|r'|^2)\cosh s-2r\cdot r'}{2\sinh s}\right] \;,
\label{eq:2dho-Mehler}
\end{equation}
where $r \cdot r'$ denotes the dot product of vectors $r$ and $r'$. On the diagonal it reduces to
\begin{equation}
    K_s(q,q)=\frac{1}{2\pi\sinh s}
    \exp\left[-|q|^2\tanh\frac{s}{2}\right] \;.
    \label{eq:2dho-diagonal-kernel}
\end{equation}
The dependence is only through the invariant radius $|q|^2=q_1^2+q_2^2$.  This observation already anticipates two facts: the first derivative at the origin vanishes, and the Hessian at the origin is proportional to $\delta_{ij}$.

The Cartesian oscillator eigenfunctions are given by \cite{KonishiPaffutiQM}
\begin{equation}
    \Phi_{n_1,n_2}(x,y)=\chi_{n_1}(x)\chi_{n_2}(y),
    \qquad
    E_{n_1,n_2}=n_1+n_2+1 \;,
    \label{eq:2dho-cartesian}
\end{equation}
where
\begin{equation}
    \chi_n(x)=\frac{1}{\pi^{1/4}\sqrt{2^n n!}}H_n(x)\e^{-x^2/2} \;.
\end{equation}
At the origin the one-dimensional factors satisfy
$\chi_{2p+1}(0)=0$, and $\chi'_{2p}(0)=0$. We write
\begin{equation}
    A_p:=\chi_{2p}(0)=(-1)^p\frac{\sqrt{(2p)!}}{\pi^{1/4}2^p p!},
    \qquad
    B_p:=\chi'_{2p+1}(0)=\sqrt{2(2p+1)}A_p \;.
    \label{eq:Ap-Bp}
\end{equation}
Thus
\begin{equation}
    A_p^2=\frac{(2p)!}{\sqrt\pi\,2^{2p}(p!)^2},
    \qquad
    B_p^2=2(2p+1)A_p^2 \;.
\end{equation}
The notation $p,\ell$ will be used for oscillator indices, so that it is not confused with the moving support coordinate $q$.

For a support point $q$, the renormalized principal function is 
\begin{equation}
    \Phi(E,q)=\int_0^\infty K_s(q,q)\bigl(\e^{-s\mu_0}-\e^{sE}\bigr)\dd s \;.
\end{equation}
At $q=0$ the shifted point-interaction energies $E_\alpha^*$ are the zeros $\Phi(E_\alpha^*,0)=0$, where $\alpha=(n_1,n_2)$. The normalized point-interaction eigenfunction is
\begin{equation}
    \psi_\alpha(r)=\frac{G_0(r,0|E_\alpha^*)}{\sqrt{N_\alpha}},
    \qquad
    N_\alpha=-\Phi_E(E_\alpha^*,0)>0 \;.
    \label{eq:2dho-point-state}
\end{equation}
Using the Cartesian basis, the Green function at the origin is
\begin{equation}
G_0(r,0|E)=
\sum_{p,\ell=0}^{\infty}
\frac{A_pA_\ell\,\chi_{2p}(x)\chi_{2\ell}(y)}{2p+2\ell+1-E} \;.
\label{eq:2dho-G-origin}
\end{equation}
Only even-even modes appear because $\Phi_{n_1,n_2}(0,0)=\chi_{n_1}(0)\chi_{n_2}(0)$ vanishes unless both indices are even.  The finite normalization derivative is therefore
\begin{equation}
    N_\alpha=
    \sum_{p,\ell=0}^{\infty}
    \frac{A_p^2A_\ell^2}{(2p+2\ell+1-E_\alpha^*)^2} \;.
    \label{eq:2dho-Nalpha}
\end{equation}
The undifferentiated diagonal Green function is logarithmically divergent in two dimensions; the derivative in equation \eqref{eq:2dho-Nalpha} is finite and is precisely the normalization of the renormalized state.

Let us now assume that the center of the delta potential executes a small oscillation
\begin{equation}
    q(t)=\varepsilon u\cos(\Omega t),
    \qquad
    u=(u_1,u_2),
    \qquad |u|=1,
    \qquad |\varepsilon|\ll1 \;.
    \label{eq:2dho-q-motion}
\end{equation}
The pole equation is $\Phi(E_\alpha(t),q(t))=0$, with $E_\alpha(t)=E_\alpha^*+\Delta_\alpha^{(1)}(t)+\Delta_\alpha^{(2)}(t)+\cdots$. The first derivative of the diagonal kernel is
\begin{equation}
    \partial_{q_i}K_s(q,q)
    =-2q_i\tanh\frac{s}{2}\,K_s(q,q) \;.
\end{equation}
Hence
\begin{equation}
    \partial_{q_i}K_s(q,q)\big|_{q=0}=0,
    \qquad
    \Phi_i(E,0)=0 \;.
    \label{eq:2dho-Phi-i-zero}
\end{equation}
This cancellation is also visible in the oscillator basis: for each product state, the derivative of $|\Phi_{n_1,n_2}(q)|^2$ at $q=0$ vanishes, because either an odd factor has zero value or an even factor has zero first derivative.  Therefore \(\Delta_\alpha^{(1)}(t)=0\). The vanishing is a two-dimensional symmetry statement, not a one-dimensional parity accident.

The second coordinate derivative follows from another differentiation of equation \eqref{eq:2dho-diagonal-kernel}:
\begin{equation}
\partial_{q_i}\partial_{q_j}K_s(q,q)=
\left[-2\delta_{ij}\tanh\frac{s}{2}
+4q_iq_j\tanh^2\frac{s}{2}\right]K_s(q,q) \;.
\end{equation}
At $q=0$,
\begin{equation}
\partial_{q_i}\partial_{q_j}K_s(q,q)\big|_{q=0}
=-2\delta_{ij}\tanh\frac{s}{2}K_s(0,0) \;.
\end{equation}
Consequently
\begin{equation}
    \Phi_{ij}(E,0)=C(E)\delta_{ij} \;,
    \label{eq:2dho-Hessian-isotropic}
\end{equation}
where
\begin{equation}
    C(E)=-2\int_0^\infty K_s(0,0)\tanh \left(\frac{s}{2} \right) 
    \bigl(\e^{-s\mu_0}-\e^{sE}\bigr)\dd s \;.
    \label{eq:2dho-C}
\end{equation}
The integral is finite.  Near $s=0$, one has $K_s(0,0)\sim(2\pi s)^{-1}$, $\tanh(s/2)\sim s/2$, and $\e^{-s\mu_0}-\e^{sE}=-(\mu_0+E)s+O(s^2)$.  Thus the integrand is regular at the lower limit.

The same Hessian can be checked term by term.  For the $x$-direction,
\begin{align}
\Phi_{11}(E,0)=&
\sum_{p,\ell=0}^{\infty}
2A_\ell^2A_p\chi_{2p}''(0)
\left(\frac1{2p+2\ell+1+\mu_0}-\frac1{2p+2\ell+1-E}\right)
\nonumber\\
&+\sum_{p,\ell=0}^{\infty}
2B_p^2A_\ell^2
\left(\frac1{2p+2\ell+2+\mu_0}-\frac1{2p+2\ell+2-E}\right) \;.
\label{eq:2dho-Hessian-spectral}
\end{align}
The first line comes from differentiating even-even modes twice; the second line comes from the modes odd in $x$ and even in $y$.  Rotational symmetry gives
\begin{equation}
    \Phi_{11}(E,0)=\Phi_{22}(E,0)=C(E),
    \qquad
    \Phi_{12}(E,0)=0 \;.
\end{equation}
Since $\Delta_\alpha^{(1)}=0$, the second-order pole shift reduces to
\begin{equation}
    \Phi_E(E_\alpha^*,0)\Delta_\alpha^{(2)}(t)
    +\frac12\Phi_{ij}(E_\alpha^*,0)q_i(t)q_j(t)=0 \;.
\end{equation}
Using $N_\alpha=-\Phi_E(E_\alpha^*,0)$ gives
\begin{equation}
    \Delta_\alpha^{(2)}(t)=
    \frac{1}{2N_\alpha}\Phi_{ij}(E_\alpha^*,0)q_i(t)q_j(t) \;.
    \label{eq:2dho-second-shift-general}
\end{equation}
For the isotropic oscillator and the motion \eqref{eq:2dho-q-motion}, this becomes
\begin{equation}
    \Delta_\alpha^{(2)}(t)=
    \frac{C(E_\alpha^*)}{2N_\alpha}\varepsilon^2\cos^2(\Omega t) \;.
    \label{eq:2dho-second-shift}
\end{equation}
The answer is independent of the direction $u$, as required by rotational symmetry.

Although the first-order energy shift vanishes, first-order transitions do not.  The derivative of the point-interaction eigenfunction in the direction of the motion is
\begin{equation}
    (u\cdot\nabla_q)\psi_\alpha(r;q)\big|_{q=0}
    =\frac{1}{\sqrt{N_\alpha}}
    (u\cdot\nabla_q)G_0(r,q|E_\alpha^*)\big|_{q=0} \;.
\end{equation}
The normalization derivative does not appear at first order because it is proportional to $\Phi_i(E_\alpha^*,0)$, which vanishes.  From the spectral representation,
\begin{align}
(u\cdot\nabla_q)G_0(r,q|E_\alpha^*)\big|_{q=0}
={}&\sum_{n_1,n_2}\frac{\Phi_{n_1,n_2}(r)u_i\partial_i\Phi_{n_1,n_2}(0)}{E_{n_1,n_2}-E_\alpha^*} \;.
\end{align}
The only nonzero derivatives at the origin are
\begin{equation}
    \partial_x\Phi_{2p+1,2\ell}(0,0)=B_pA_\ell \;,
    \qquad
    \partial_y\Phi_{2p,2\ell+1}(0,0)=A_pB_\ell \;.
\end{equation}
Therefore
\begin{align}
(u\cdot\nabla_q)\psi_\alpha(r;q)\big|_{q=0}
={}&\frac{1}{\sqrt{N_\alpha}}
\sum_{p,\ell=0}^{\infty}
\frac{u_1B_pA_\ell\,\Phi_{2p+1,2\ell}(r)}{2p+2\ell+2-E_\alpha^*}
\nonumber\\
&+\frac{1}{\sqrt{N_\alpha}}
\sum_{p,\ell=0}^{\infty}
\frac{u_2A_pB_\ell\,\Phi_{2p,2\ell+1}(r)}{2p+2\ell+2-E_\alpha^*} \;.
\label{eq:2dho-selection-vector}
\end{align}
This is the vector selection rule.  Motion along $u$ couples the point-interaction sector to the $u_1$-weighted $x$-odd family and the $u_2$-weighted $y$-odd family.  The explicit coefficients are
\begin{equation}
    B_pA_\ell=(-1)^{p+\ell}
    \frac{\sqrt{2(2p+1)(2p)!(2\ell)!}}{\sqrt\pi\,2^{p+\ell}p!\ell!} \;,
    \label{eq:2dho-BpAl}
\end{equation}
\begin{equation}
    A_pB_\ell=(-1)^{p+\ell}
    \frac{\sqrt{2(2\ell+1)(2p)!(2\ell)!}}{\sqrt\pi\,2^{p+\ell}p!\ell!} \;.
    \label{eq:2dho-ApBl}
\end{equation}
Let $\rho_x=(2p+1,2\ell)$ and $\rho_y=(2p,2\ell+1)$ denote allowed odd oscillator states.  Differentiating the instantaneous eigenvalue equation gives, for any unaffected oscillator state $\Phi_\rho$,
\begin{equation}
    \langle\Phi_\rho|H^{(1)}(t)|\psi_\alpha\rangle
    =(E_\alpha^*-E_\rho)\,
    \langle\Phi_\rho|q_i(t)\partial_{q_i}\psi_\alpha\rangle \;.
\end{equation}
It is convenient to factor out the time dependence and define
\begin{equation}
    \mathcal{W}_{(2p+1,2\ell);\alpha}^{(1)}(u)=-\frac{u_1B_pA_\ell}{\sqrt{N_\alpha}} \;,
    \qquad
    \mathcal{W}_{(2p,2\ell+1);\alpha}^{(1)}(u)=-\frac{u_2A_pB_\ell}{\sqrt{N_\alpha}} \;.
    \label{eq:2dho-W1}
\end{equation}
All other first-order matrix elements from $\psi_\alpha$ to oscillator modes vanish.  If $u=(1,0)$, only $x$-odd states remain; if $u=(0,1)$, only $y$-odd states remain.

For an allowed odd state $\rho$, set
$\omega_{\rho\alpha}=\frac{E_\rho-E_\alpha^*}{\hbar}$. With the oscillation \eqref{eq:2dho-q-motion}, the first-order transition amplitude is
\begin{equation}
    c_\rho^{(1)}(t)
    =-\frac{\ii}{\hbar} \mathcal{W}_{\rho\alpha}^{(1)}(u)
    \int_0^t\varepsilon\cos(\Omega t')\e^{\ii\omega_{\rho\alpha}t'}\dd t' \;.
\end{equation}
Thus
\begin{align}
    c_\rho^{(1)}(t)
    =-\frac{\ii\varepsilon}{2\hbar} \mathcal{W}_{\rho\alpha}^{(1)}(u)
    \left[
    \frac{\e^{\ii(\omega_{\rho\alpha}+\Omega)t}-1}{\ii(\omega_{\rho\alpha}+\Omega)}
    +
    \frac{\e^{\ii(\omega_{\rho\alpha}-\Omega)t}-1}{\ii(\omega_{\rho\alpha}-\Omega)}
    \right] \;.
    \label{eq:2dho-first-amplitude}
\end{align}

Near resonance, $\Omega\simeq\omega_{\rho\alpha}$, the second denominator must be interpreted by the limit
\begin{equation}
    \lim_{\omega\to\Omega}\frac{\e^{\ii(\omega-\Omega)t}-1}{\ii(\omega-\Omega)}=t \;.
\end{equation}
Thus the probability grows quadratically at short resonant times,
\begin{equation}
    |c_\rho^{(1)}(t)|^2\simeq
    \frac{\varepsilon^2}{4\hbar^2}|\mathcal{W}_{\rho\alpha}^{(1)}(u)|^2t^2 \;,
\end{equation}
up to oscillatory non-resonant terms.  This is the usual first-order resonance behaviour, now with a coupling fixed by the derivative of the renormalized point-interaction eigenfunction.

For $u=(1,0)$ this reads
\begin{align}
    c_{(2p+1,2\ell)}^{(1)}(t)
    =\frac{\varepsilon B_pA_\ell}{2\hbar\sqrt{N_\alpha}}
    \left[
    \frac{\e^{\ii(\omega_{\rho\alpha}+\Omega)t}-1}{\omega_{\rho\alpha}+\Omega}
    +
    \frac{\e^{\ii(\omega_{\rho\alpha}-\Omega)t}-1}{\omega_{\rho\alpha}-\Omega}
    \right] \;,
\end{align}
while the $y$-odd amplitudes vanish only because the imposed motion has no $y$ component.

The second-order structure has two distinct parts.  Let $\psi_\beta$ be another point-interaction state in the even sector.  Then
\begin{equation}
    c_\beta^{(2)}(t)=c_{\beta,\mathrm{direct}}^{(2)}(t)+c_{\beta,\mathrm{iter}}^{(2)}(t) \;.
\end{equation}
The direct term comes from the second coordinate derivative of the Hamiltonian along $u$:
\begin{equation}
    \mathcal W_{\beta\alpha}^{(2)}(u)
    :=\frac12u_iu_j
    \langle\psi_\beta|\partial_{q_i}\partial_{q_j}H(q)|_{q=0}|\psi_\alpha\rangle \;.
\end{equation}
Thus
\begin{equation}
    c_{\beta,\mathrm{direct}}^{(2)}(t)
    =-\frac{\ii}{\hbar}\mathcal W_{\beta\alpha}^{(2)}(u)
    \int_0^t\varepsilon^2\cos^2(\Omega t')
    \e^{\ii\omega_{\beta\alpha}t'}\dd t'.
    \label{eq:2dho-direct-second}
\end{equation}
For the diagonal case, consistency with equation \eqref{eq:2dho-second-shift} gives
\begin{equation}
    \mathcal W_{\alpha\alpha}^{(2)}(u)=
    \frac{1}{2N_\alpha}\Phi_{ij}(E_\alpha^*,0)u_iu_j
    =\frac{C(E_\alpha^*)}{2N_\alpha} \;.
\end{equation}
The elementary integral in equation \eqref{eq:2dho-direct-second} is
\begin{eqnarray}
\begin{aligned}
\int_0^t\cos^2(\Omega t')\e^{\ii\omega_{\beta\alpha}t'}\dd t'
=  \frac12\frac{\e^{\ii\omega_{\beta\alpha}t}-1}{\ii\omega_{\beta\alpha}}
+  \frac14\frac{\e^{\ii(\omega_{\beta\alpha}+2\Omega)t}-1}{\ii(\omega_{\beta\alpha}+2\Omega)} 
+\frac14\frac{\e^{\ii(\omega_{\beta\alpha}-2\Omega)t}-1}{\ii(\omega_{\beta\alpha}-2\Omega)} \;.
\end{aligned}    \label{intcos2}
\end{eqnarray}
The iterated term is a sum over the allowed odd oscillator modes.  Let
\begin{equation}
    \mathcal I_u=\{(2p+1,2\ell):p,\ell\ge0\}\cup\{(2p,2\ell+1):p,\ell\ge0\} \;,
\end{equation}
with the direction-dependent weights already contained in $W^{(1)}(u)$.  Then
\begin{align}
 c_{\beta,\mathrm{iter}}^{(2)}(t)=
-&\frac{1}{\hbar^2}
\sum_{\rho\in\mathcal I_u} \mathcal{W}_{\beta\rho}^{(1)}(u) \mathcal{W}_{\rho\alpha}^{(1)}(u)
\int_0^t\dd t_1\,\varepsilon\cos(\Omega t_1)
\e^{\frac{\ii}{\hbar}(E_\beta^*-E_\rho)t_1}
\nonumber\\
&\times
\int_0^{t_1}\dd t_2\,\varepsilon\cos(\Omega t_2)
\e^{\frac{\ii}{\hbar}(E_\rho-E_\alpha^*)t_2} \;.
\label{eq:2dho-iterated-second}
\end{align}
This formula displays the full two-dimensional character of the motion: for a generic direction $u$, both odd families contribute.

The calculations in this subsection verify the general moving-center formulas explicitly.  The pole does not move at first order, but the state does.  The leading energy shift is quadratic in the displacement and isotropic.  The first-order transition selection rule is vectorial: the direction of the motion selects the corresponding odd oscillator sector.

\subsection{Circular motion of a renormalized point interaction in the two-dimensional oscillator}

This example is a refinement of the previous two-dimensional oscillator example.  There we moved the support in a fixed direction,
$q(t)=\varepsilon u\cos\Omega t$. Now the point interaction moves on a small circle around the origin,
\begin{equation}
q(t)=\varepsilon(\cos\Omega t,\sin\Omega t),\qquad |\varepsilon|\ll1 \;.
\label{CIRC1}    
\end{equation}
The interaction is still a genuine two-dimensional renormalized point interaction.  The support is the point
\begin{equation}
q(t)=(q_1(t),q_2(t)),
\qquad
\delta^{(2)}(r-q(t))=\delta(x-q_1(t))\delta(y-q_2(t)) \;.
\label{CIRC2}    
\end{equation}
This example is useful because the direction of motion rotates, so the first-order transition is naturally written in circularly polarized combinations of the two Cartesian derivatives.
We keep the same dimensionless two-dimensional oscillator used above equation \eqref{eq:2dho-H0}.

The first-order Taylor expansion of the moving-center zero condition is 
\begin{equation}
0=\Phi_E \Delta_\alpha^{(1)}(t)+\Phi_iq_i(t)\;,
\label{CIRC13}    
\end{equation}
where all derivatives are evaluated at \((E_\alpha^*,0)\).  From the diagonal heat kernel,
\begin{equation}
    \partial_iK_s(q,q)=-2q_i\tanh\frac{s}{2}K_s(q,q) \;,
\label{CIRC14}
\end{equation}
so at the origin
$\partial_iK_s(q,q)|_{q=0}=0$. Consequently
\begin{equation}
\Phi_i(E,0)=\int_0^\infty \partial_iK_s(q,q)|_{q=0}
\left(\e^{-s\mu}-\e^{sE}\right)\dd s=0 \;.
\label{CIRC16}    
\end{equation}
Substituting this into equation \eqref{CIRC13} gives
$\Delta_\alpha^{(1)}(t)=0$. Thus the circular motion does not create a first-order pole shift.  This is the same cancellation as in the straight oscillation case, but now it happens for a rotating displacement vector.

At second order, because \(\Delta_\alpha^{(1)}=0\), the zero condition reduces to
\begin{equation}
0=\Phi_E \Delta_\alpha^{(2)}(t)+\frac12\Phi_{ij}q_i(t)q_j(t) \;.
\label{CIRC18}    
\end{equation}
The second coordinate derivative of the diagonal heat kernel is
\begin{equation}
\partial_i\partial_jK_s(q,q)=
\left[-2\delta_{ij}\tanh\frac{s}{2}+4q_iq_j\tanh^2\frac{s}{2}\right]K_s(q,q) \;.
\label{CIRC19}    
\end{equation}
At the origin, $\partial_i\partial_jK_s(q,q)|_{q=0}=-2\delta_{ij}\tanh\frac{s}{2}K_s(0,0)$. Therefore the Hessian has the isotropic form already given in equation \eqref{eq:2dho-Hessian-isotropic}, with \(C(E)\) defined in equation \eqref{eq:2dho-C}.
For the circular trajectory \eqref{CIRC1},
\begin{equation}
 q_i(t)q_j(t)\delta_{ij}=q_1(t)^2+q_2(t)^2
=\varepsilon^2\cos^2\Omega t+\varepsilon^2\sin^2\Omega t=\varepsilon^2.
\label{CIRC23}   
\end{equation}
Thus $\Phi_{ij}q_iq_j=C(E_\alpha^*)\varepsilon^2$. Using \(N_\alpha=-\Phi_E\), equation \eqref{CIRC18} gives
\begin{equation}
\Delta_\alpha^{(2)}=
\frac{C(E_\alpha^*)}{2N_\alpha}\varepsilon^2 \;.
\label{CIRC25}    
\end{equation}
This is the first important difference from the straight-line oscillation.  For
\(q(t)=\varepsilon u\cos\Omega t\), the quadratic shift was proportional to \(\varepsilon^2\cos^2\Omega t\).  For circular motion, \(|q(t)|^2\) is constant, so the second-order pole shift is constant.

The first-order change of the instantaneous eigenfunction is governed by
\begin{equation}
q_i(t)\partial_{q_i}G_0(r,q|E_\alpha^*)|_{q=0} \;.
\label{CIRC26}    
\end{equation}
For the circular trajectory,
\begin{equation}
q_i(t)\partial_{q_i}
=\varepsilon\left(\cos\Omega t\,\partial_x+\sin\Omega t\,\partial_y\right) \;.
\label{CIRC27}    
\end{equation}
Writing the trigonometric functions as exponentials gives
\begin{equation}
\cos\Omega t\,\partial_x+\sin\Omega t\,\partial_y
=\frac12\e^{\ii\Omega t}(\partial_x-\ii\partial_y)
+\frac12\e^{-\ii\Omega t}(\partial_x+\ii\partial_y) \;.
\label{CIRC28}    
\end{equation}
It is therefore natural to define
\begin{equation}
D_-:=\partial_x-\ii\partial_y,
\qquad
D_+:=\partial_x+\ii\partial_y \;.
\label{CIRC29}    
\end{equation}
Then
\begin{equation}
q_i(t)\partial_{q_i}
=\frac{\varepsilon}{2}\left(\e^{\ii\Omega t}D_-+\e^{-\ii\Omega t}D_+\right) \;.
\label{CIRC30}    
\end{equation}
This is the clean circular-motion version of the straight-direction derivative \(u\cdot\nabla\).  It shows explicitly that the driving has two circular frequency components, one rotating with \(+\Omega\) and one with \(-\Omega\).

From the spectral representation,
\begin{equation}
\partial_x G_0(r,q|E_\alpha^*)|_{q=0}
=\sum_{p,\ell=0}^\infty
\frac{B_pA_\ell\,\Phi_{2p+1,2\ell}(r)}{2p+2\ell+2-E_\alpha^*} \;,
\label{CIRC31}    
\end{equation}
while
\begin{equation}
\partial_yG_0(r,q|E_\alpha^*)|_{q=0}
=\sum_{p,\ell=0}^\infty
\frac{A_pB_\ell\,\Phi_{2p,2\ell+1}(r)}{2p+2\ell+2-E_\alpha^*} \;.
\label{CIRC32}    
\end{equation}
The coefficients \(B_pA_\ell\) and \(A_pB_\ell\) are the same as those displayed in equations \eqref{eq:2dho-BpAl} and \eqref{eq:2dho-ApBl}.
Thus the circular derivatives are
\begin{equation}
D_-G_0|_0=\sum_{p,\ell\ge0}
\frac{B_pA_\ell\Phi_{2p+1,2\ell}-\ii A_pB_\ell\Phi_{2p,2\ell+1}}
{2p+2\ell+2-E_\alpha^*} \;,
\label{CIRC34a}    
\end{equation}
\begin{equation}
D_+G_0|_0=\sum_{p,\ell\ge0}
\frac{B_pA_\ell\Phi_{2p+1,2\ell}+\ii A_pB_\ell\Phi_{2p,2\ell+1}}
{2p+2\ell+2-E_\alpha^*} \;.
\label{CIRC34b}    
\end{equation}
The two combinations are the Cartesian form of the angular-momentum selection rule.  Circular motion does not select purely \(x\)-odd or purely \(y\)-odd states; it selects the circularly polarized combinations of the two degenerate odd sectors.

We now introduce the static first-order matrix elements corresponding to the two circular components, defined as
\begin{equation}
\mathcal{W}_{\rho\alpha}^{[-]}
:=\mathcal{W}_{\rho\alpha}^{x}-\ii \mathcal{W}_{\rho\alpha}^{y},
\qquad
\mathcal{W}_{\rho\alpha}^{[+]}
:=\mathcal{W}_{\rho\alpha}^{x}+\ii \mathcal{W}_{\rho\alpha}^{y},
\label{CIRC35}    
\end{equation}
where
\begin{equation}
\mathcal{W}_{(2p+1,2\ell),\alpha}^{x}=-\frac{B_pA_\ell}{\sqrt{N_\alpha}} \;,
\qquad
\mathcal{W}_{(2p,2\ell+1),\alpha}^{y}=-\frac{A_pB_\ell}{\sqrt{N_\alpha}} \;,
\label{CIRC36}    
\end{equation}
and all incompatible Cartesian components vanish.  With this convention the first-order perturbation matrix element is
\begin{equation}
    \langle \psi_\rho|H^{(1)}(t)|\psi_\alpha\rangle
=\frac{\varepsilon}{2}\left(
\e^{\ii\Omega t} \mathcal{W}_{\rho\alpha}^{[-]}+
\e^{-\ii\Omega t} \mathcal{W}_{\rho\alpha}^{[+]}
\right) \;.
\label{CIRC37}
\end{equation}
It follows from \eqref{CIRC37} that
\begin{eqnarray} 
\begin{aligned}
c_\rho^{(1)}(t) & = 
-\frac{\ii}{\hbar}\int_0^t
\langle \psi_\rho|H^{(1)}(t')|\psi_\alpha\rangle
\e^{\ii\omega_{\rho\alpha}t'}\dd t' \\  & = 
- \frac{\ii\varepsilon}{2\hbar}
\left[
\mathcal{W}_{\rho\alpha}^{[-]}
\int_0^t\e^{\ii(\omega_{\rho\alpha}+\Omega)t'}\dd t'
+
\mathcal{W}_{\rho\alpha}^{[+]}
\int_0^t\e^{\ii(\omega_{\rho\alpha}-\Omega)t'}\dd t'
\right] \;.
\end{aligned}
\label{CIRC40}
\end{eqnarray}
After evaluation of the elementary integrals,  we get
\begin{equation}
 c_\rho^{(1)}(t)=
-\frac{\varepsilon}{2\hbar}
\left[
\mathcal{W}_{\rho\alpha}^{[-]}
\frac{\e^{\ii(\omega_{\rho\alpha}+\Omega)t}-1}{\omega_{\rho\alpha}+\Omega}
+
\mathcal{W}_{\rho\alpha}^{[+]}
\frac{\e^{\ii(\omega_{\rho\alpha}-\Omega)t}-1}{\omega_{\rho\alpha}-\Omega}
\right].
\label{CIRC42}    
\end{equation}
Thus circular motion gives two resonance channels.  The component rotating as \(\e^{\ii\Omega t}\) produces \(\omega_{\rho\alpha}+\Omega\), and the component rotating as \(\e^{-\ii\Omega t}\) produces \(\omega_{\rho\alpha}-\Omega\).

The second-order Taylor term in the moving support is
\begin{equation}
H_{\rm direct}^{(2)}(t)=\frac12 q_i(t)q_j(t)H_{ij},
\qquad
H_{ij}:=\partial_{q_i}\partial_{q_j}H(q)|_{q=0}.
\label{CIRC43}    
\end{equation}
Using equation \eqref{CIRC30}, we find
\begin{equation}
q_iq_jH_{ij}=\frac{\varepsilon^2}{4}
\left[
\e^{2\ii\Omega t}H_{--}+\e^{-2\ii\Omega t}H_{++}+H_{-+}+H_{+-}\right],
\label{CIRC44}    
\end{equation}
where
\begin{equation}
H_{--}:=(\partial_x-\ii\partial_y)^2H,
\qquad
H_{++}:=(\partial_x+\ii\partial_y)^2H,
\label{CIRC45a}    
\end{equation}
\begin{equation}
H_{-+}=H_{+-}=\partial_x^2H+\partial_y^2H.
\label{CIRC45b}    
\end{equation}
Hence
\begin{equation}
H_{\rm direct}^{(2)}(t)=\varepsilon^2\left[
\mathcal U_{\beta\alpha}^{(0)}+
\mathcal U_{\beta\alpha}^{(+2)}\e^{2\ii\Omega t}+
\mathcal U_{\beta\alpha}^{(-2)}\e^{-2\ii\Omega t}
\right]_{\beta\alpha} \;,
\label{CIRC46}    
\end{equation}
with matrix elements $\mathcal U_{\beta\alpha}^{(0)}=\frac14\langle\psi_\beta|H_{xx}+H_{yy}|\psi_\alpha\rangle$,
\begin{equation}
\mathcal U_{\beta\alpha}^{(+2)}=\frac18\langle\psi_\beta|H_{--}|\psi_\alpha\rangle,
\qquad
\mathcal U_{\beta\alpha}^{(-2)}=\frac18\langle\psi_\beta|H_{++}|\psi_\alpha\rangle \;.
\label{CIRC47b}    
\end{equation}
For the diagonal case, this formula must reproduce the pole shift \eqref{CIRC25}. Indeed,
\begin{equation}
\langle\psi_\alpha|H_{xx}|\psi_\alpha\rangle
=\langle\psi_\alpha|H_{yy}|\psi_\alpha\rangle
=\frac{C(E_\alpha^*)}{N_\alpha} \;,
\label{CIRC48}    
\end{equation}
and the mixed rotational pieces have zero diagonal expectation by isotropy. Therefore
\begin{equation}
\varepsilon^2\mathcal U_{\alpha\alpha}^{(0)}
=\varepsilon^2\frac14\left(\frac{C}{N_\alpha}+\frac{C}{N_\alpha}\right)
=\frac{C(E_\alpha^*)}{2N_\alpha}\varepsilon^2 \;,
\label{CIRC49}    
\end{equation}
which is exactly equation \eqref{CIRC25}.

The corresponding direct second-order transition amplitude is
\begin{eqnarray}
\begin{aligned}
 c_{\beta,{\rm direct}}^{(2)}(t)
=-\frac{\ii\varepsilon^2}{\hbar}
\Bigg[\mathcal U_{\beta\alpha}^{(0)}
\frac{\e^{\ii\omega_{\beta\alpha}t}-1}{\ii\omega_{\beta\alpha}}
& +\mathcal U_{\beta\alpha}^{(+2)}
\frac{\e^{\ii(\omega_{\beta\alpha}+2\Omega)t}-1}{\ii(\omega_{\beta\alpha}+2\Omega)}
\\
&+\mathcal U_{\beta\alpha}^{(-2)}
\frac{\e^{\ii(\omega_{\beta\alpha}-2\Omega)t}-1}{\ii(\omega_{\beta\alpha}-2\Omega)}
\Bigg] \;.
\end{aligned}
\label{CIRC50}    
\end{eqnarray}
This is the circular-motion analogue of the \(\cos^2\Omega t\) direct term.  The difference is that the constant part is now geometrically forced by \(|q(t)|^2=\varepsilon^2\), while the anisotropic part rotates at \(\pm2\Omega\).

The iterated term is obtained by inserting the first-order circular perturbation twice.  Let \(s,r\in\{+1,-1\}\), and define
\begin{equation}
\mathcal{W}_{ab}^{[+1]}:=\mathcal{W}_{ab}^{[-]} \;,
\qquad
\mathcal{W}_{ab}^{[-1]}:=\mathcal{W}_{ab}^{[+]} \;,
\label{CIRC51}    
\end{equation}
so that the first-order matrix element can be written compactly as
\begin{equation}
H_{ab}^{(1)}(t)=\frac{\varepsilon}{2}\sum_{s=\pm1}\e^{\ii s\Omega t} \mathcal{W}_{ab}^{[s]} \;.
\label{CIRC52}    
\end{equation}
Then
\begin{eqnarray}
\begin{aligned}
 c_{\beta,{\rm iter}}^{(2)}(t)
=-\frac{\varepsilon^2}{4\hbar^2}
\sum_\rho\sum_{s,r=\pm1}
& \mathcal{W}_{\beta\rho}^{[s]} \mathcal{W}_{\rho\alpha}^{[r]}
\int_0^t\dd t_1\,
\e^{\ii(\omega_{\beta\rho}+s\Omega)t_1}
\\
&\times
\int_0^{t_1}\dd t_2\,
\e^{\ii(\omega_{\rho\alpha}+r\Omega)t_2} \;.
\end{aligned}
\label{CIRC53}    
\end{eqnarray}
The time-ordered integral is elementary.  For
\begin{equation}
A=\omega_{\beta\rho}+s\Omega,
\qquad
B=\omega_{\rho\alpha}+r\Omega \;,
\label{CIRC54}    
\end{equation}
we have
\begin{equation}
    \mathcal I(A,B;t):=\int_0^t\dd t_1\,\e^{\ii At_1}
\int_0^{t_1}\dd t_2\,\e^{\ii Bt_2}
=\frac1{\ii B}\left[
\frac{\e^{\ii(A+B)t}-1}{\ii(A+B)}-
\frac{\e^{\ii At}-1}{\ii A}
\right] \;,
\label{CIRC55}
\end{equation}
with the usual limiting interpretation if a denominator vanishes. Thus
\begin{equation}
 c_{\beta,{\rm iter}}^{(2)}(t)
=-\frac{\varepsilon^2}{4\hbar^2}
\sum_\rho\sum_{s,r=\pm1}
\mathcal{W}_{\beta\rho}^{[s]} \mathcal{W}_{\rho\alpha}^{[r]}
\mathcal I(\omega_{\beta\rho}+s\Omega,\omega_{\rho\alpha}+r\Omega;t) \;.
\label{CIRC56}    
\end{equation}
The intermediate states \(\rho\) are precisely the odd oscillator states reached by the first circular derivative.  In a basis adapted to angular momentum, this is the statement that the first insertion changes angular momentum by one unit, and two first-order insertions can return to the radial sector or reach the corresponding second-order sector.

The straight-line motion and circular motion share the same local renormalized Hessian \(\Phi_{ij}=C\delta_{ij}\), so both have $\Delta^{(1)}=0$.
They differ at second order because the radius of circular motion is constant:
$|q(t)|^2=\varepsilon^2$ instead of  $|q(t)|^2=\varepsilon^2\cos^2\Omega t$. Therefore
\begin{equation}
\Delta_\alpha^{(2)}=\frac{C(E_\alpha^*)}{2N_\alpha}\varepsilon^2
\quad\text{is constant for circular motion}.
\label{CIRC57}    
\end{equation}
The first-order transition is not a single Cartesian selection rule but the circular combination
$\partial_x\mp\ii\partial_y$, so the first-order amplitude has the two resonance denominators
$\omega_{\rho\alpha}+\Omega$, $\omega_{\rho\alpha}-\Omega$.
At second order the direct term contains a constant part and rotating \(\pm2\Omega\) parts, while the iterated term is the ordered sum over two circular first-order insertions.  This is the cleanest oscillator example for displaying the angular-momentum content of the moving renormalized point interaction.

\subsection{One-dimensional oscillator without renormalization}
\label{sec:1dho-detailed}

The main purpose of the manuscript is to develop a renormalized perturbation theory for time-dependent point interactions with discrete spectrum. A one-dimensional delta interaction does not need renormalization. Nevertheless, it provides an exact comparison case for the moving-center formalism, because the Green function, the eigenfunctions, and the parity selection rules are all explicit.
We therefore use this example as a one-dimensional analogue of the moving-delta calculation. The role of the renormalized principal function is played by the ordinary one-dimensional point-interaction denominator. Moreover, it allows us to compare the projection-based formulation with the standard time-dependent perturbation expansion of a moving delta potential. 

Let
\begin{equation}
H_{0}= -\frac{1}{2}\frac{d^2}{dx^2}+\frac12 x^2 \;,    
\end{equation}
where we have used units such that $\hbar=m=\omega=1$. Its normalized eigenfunctions and eigenvalues \cite{KonishiPaffutiQM} are
\begin{equation}
\varphi_n(x)=\frac{1}{\pi^{1/4}\sqrt{2^n n!}}
H_n\!\left(x \right)\exp\!\left(-\frac{x^2}{2}\right),
\qquad
E_n= \left(n+\frac12\right) \;.   
\end{equation}
We take an attractive delta center with fixed strength \(\lambda>0\), moving slightly around the origin:
\begin{equation}
H(a)=H_{0}-\lambda\delta(x-a),
\qquad
a(t)=\varepsilon f(t),
\qquad |\varepsilon|\ll1  \;.   
\end{equation}
For a simple oscillating center, one may put
\begin{equation}
f(t)=\cos(\Omega t),
\qquad a(t)=\varepsilon\cos(\Omega t) \;.
\end{equation}
If the opposite sign convention for the delta coupling is used, one should replace \(\lambda\) by \(-\lambda\) in the formulas below.  The perturbative structure is unchanged.

The harmonic-oscillator Green function is
\begin{equation}
    G_{0}(x,y|E)=\sum_{n=0}^{\infty}
\frac{\varphi_n(x)\varphi_n(y)}{E_n-E} \;.
\label{HO1}
\end{equation}
In one dimension the point-interaction eigenvalues in the sector seen by the delta center are the zeros of
\begin{equation}
\Phi(E,a)=\frac1\lambda-G_{0}(a,a|E) \;.
\label{HO2}    
\end{equation}
At the origin,
\begin{equation}
\Phi(E_\alpha^*,0)=0
\quad\Longleftrightarrow\quad
G_{0}(0,0|E_\alpha^*)=\frac1\lambda \;.
\label{HO3}    
\end{equation}
The corresponding normalized eigenfunction is
\begin{equation}
\psi_\alpha(x)=\frac{G_{0}(x,0|E_\alpha^*)}{\sqrt{N_\alpha}} \;,
\qquad
N_\alpha=\partial_EG_{0}(0,0|E_\alpha^*) \;.
\label{HO4}    
\end{equation}
This is exactly the one-dimensional analogue of the general formula
derived above.

Because the delta is at the origin, the odd oscillator states remain unaffected at \(a=0\), to emphasize this aspect, we introduce the following notation:
\begin{equation}
\chi_r(x):=\varphi_{2r+1}(x) \;,
\qquad
\mathcal E_r:=E_{2r+1} \;.
\label{HO5}    
\end{equation}
They satisfy \(\chi_r(0)=0\). Thus, the full basis at the central position consists of the even point-interaction states \(\psi_\alpha\) and the odd oscillator states \(\chi_r\).  This split is important: the first-order motion of the center couples even and odd parity sectors.

The diagonal Green function along the center is
\begin{equation}
G_{0}(a,a|E)
=\sum_{n=0}^{\infty}\frac{\varphi_n(a)^2}{E_n-E} \;.
\label{HO6}    
\end{equation}
Its first derivative is
\begin{equation}
\partial_a G_0(a,a|E)
=2\sum_{n=0}^{\infty}
\frac{\varphi_n(a)\varphi_n'(a)}{E_n-E} \;.
\label{HO7}    
\end{equation}
At \(a=0\), parity gives
$\varphi_{2j}'(0)=0$, $\varphi_{2j+1}(0)=0$. Therefore each term in \eqref{HO7} vanishes at the origin and
\begin{equation}
\partial_a G_0(0,0|E)=0,
\qquad
\Phi_a(E,0)=-\partial_a G_0(0,0|E)=0.
\label{HO8}    
\end{equation}
The second derivative is
\begin{equation}
\partial_{a}^{2} G_0(0,0|E)
=2\sum_{n=0}^{\infty}
\frac{\varphi_n'(0)^2+\varphi_n(0)\varphi_n''(0)}{E_n-E} \;.
\label{HO9}    
\end{equation}
If desired, the oscillator equation gives at the origin
\begin{equation}
\varphi_n''(0)=-\frac{2mE_n}{\hbar^2}\varphi_n(0) \;,
\label{HO10}    
\end{equation}
so \eqref{HO9} may be rewritten entirely in terms of the values and slopes of the known oscillator eigenfunctions at zero.

The moving zero condition in this case is $0=\Phi(E_\alpha^*(t),a(t))$, with $E_\alpha^*(t)=E_\alpha^*+\Delta_\alpha^{(1)}(t)+\Delta_\alpha^{(2)}(t)$.
Expanding the principal function around \((E_\alpha^*,0)\) gives
\begin{equation}
0=\Phi_E\Delta_\alpha^{(1)}+
\Phi_a a(t) \;,
\label{HO12}    
\end{equation}
at first order.  Since \(\Phi_a=0\),
$\Delta_\alpha^{(1)}(t)=0$. At second order,
\begin{equation}
0=
\Phi_E\Delta_\alpha^{(2)}(t)
+\frac12\Phi_{aa}a(t)^2
+\Phi_{Ea}\Delta_\alpha^{(1)}(t)a(t)
+\frac12\Phi_{EE}\bigl(\Delta_\alpha^{(1)}(t)\bigr)^2.
\label{HO14}    
\end{equation}
Using \(\Delta_\alpha^{(1)}=0\), this reduces to
\begin{equation}
\Phi_E\Delta_\alpha^{(2)}(t)+\frac12\Phi_{aa}a(t)^2=0 \;.    
\end{equation}
Now
\begin{equation}
\Phi_E(E,0)=-\partial_E G_{0}(0,0|E),
\qquad
\Phi_{aa}(E,0)=- \partial_{a}^2 G_0(0,0|E) \;.    
\end{equation}
Therefore
\begin{equation}
\Delta_\alpha^{(2)}(t)
=-\frac{\partial_{a}^2 G_0(0,0|E_\alpha^*)}{2\,\partial_E G_{0}(0,0|E_\alpha^*)}\,a(t)^2 \;.
\label{HO15}    
\end{equation}
For the harmonic motion \(a(t)=\varepsilon\cos\Omega t\),
\begin{equation}
\Delta_\alpha^{(2)}(t)
=-\frac{\partial_{a}^2 G_{0}(0,0|E_\alpha^*)}{2N_\alpha}\,
\varepsilon^2\cos^2(\Omega t) \;.
\label{HO16}    
\end{equation}
This explicitly shows the parity result: the first non-zero energy shift is quadratic in the displacement.

The normalized point-sector eigenfunction at a shifted center is
\begin{equation}
\psi_\alpha(x;a)=\frac{G_{0}(x,a|E_\alpha^*(a))}{\sqrt{N_\alpha(a)}} \;.
\label{HO17}    
\end{equation}
At first order, \(\Delta_\alpha^{(1)}=0\) and \( \partial_a N_\alpha(0)=0\).  Hence
\begin{equation}
\left.\frac{\partial\psi_\alpha(x;a)}{\partial a}\right|_{a=0}
=\frac{1}{\sqrt{N_\alpha}}
\left.\frac{\partial G_{0}(x,a|E_\alpha^*)}{\partial a}\right|_{a=0} \;.
\label{HO18}    
\end{equation}
Using the spectral representation, we have
\begin{equation}
\left.\partial_aG_{0}(x,a|E_\alpha^*)\right|_{a=0}
=\sum_{n=0}^{\infty}
\frac{\varphi_n(x)\varphi_n'(0)}{E_n-E_\alpha^*} \;.
\label{HO19}    
\end{equation}
Only odd oscillator states have \(\varphi_n'(0)\neq0\), so \(\partial_a\psi_\alpha\) is odd.  This proves the selection rule:
first-order motion of the central delta couples even point states only to odd oscillator states. Indeed,
\begin{equation}
\langle\chi_r|\partial_a\psi_\alpha\rangle
=\frac{1}{\sqrt{N_\alpha}}
\frac{\chi_r'(0)}{\mathcal E_r-E_\alpha^*} \;.
\label{HO21}    
\end{equation}
Differentiating the stationary eigenvalue equation
\begin{equation}
H(a)|\psi_\alpha(a)\rangle=E_\alpha(a)|\psi_\alpha(a)\rangle \;,    
\end{equation}
with respect to \(a\), then setting \(a=0\) and projecting onto the odd state \(\chi_r\), we get 
\begin{equation}
\langle\chi_r|H_a|\psi_\alpha\rangle
=(E_\alpha^*-\mathcal E_r)\langle\chi_r|\partial_a\psi_\alpha\rangle \;,
\label{HO22}    
\end{equation}
since \(E_\alpha'(0)=0\). Using the result \eqref{HO21}, we find
\begin{equation}
\langle\chi_r|H_a|\psi_\alpha\rangle
=-\frac{\chi_r'(0)}{\sqrt{N_\alpha}}.
\label{HO23}    
\end{equation}
We now make the following observation: 
this is the same result obtained from a  direct (perhaps naive) formal distributional expansion, 
\[
H_a=H_{0}-\lambda\delta(x-a)
=H_0+a\lambda\delta'(x)-\frac{a^2\lambda}{2}\delta''(x)+\cdots,
\label{HO24}
\]
provided the sign convention in \eqref{HO2} is used. Therefore, the first-order time-dependent perturbation generated by the moving center has the non-zero matrix element
\begin{equation}
\delta H_{r\alpha}^{(1)}(t)
=-a(t)\frac{\chi_r'(0)}{\sqrt{N_\alpha}}.
\label{HO25}    
\end{equation}
All even-even and odd-odd first-order matrix elements vanish by parity:
\begin{equation}
\delta H_{\beta\alpha}^{(1)}(t)=0,
\qquad
\delta H_{rs}^{(1)}(t)=0.
\label{HO26}    
\end{equation}
Thus, starting initially in the point-sector state \(\psi_\alpha\), the first-order amplitude to the odd oscillator state \(\chi_r\) is
\begin{equation}
 c_r^{(1)}(t)
=-\frac{\ii}{\hbar}\int_0^t
\delta H_{r\alpha}^{(1)}(t')
\exp\!\left[\frac{\ii}{\hbar}(\mathcal E_r-E_\alpha^*)t'\right]\dd t' \;.
\label{HO27}    
\end{equation}
Substituting \eqref{HO25} into the above, we obtain
\begin{equation}
c_r^{(1)}(t)
=\frac{\ii\chi_r'(0)}{\hbar\sqrt{N_\alpha}}
\int_0^t a(t')
\exp\!\left[\frac{\ii}{\hbar}(\mathcal E_r-E_\alpha^*)t'\right]\dd t' \;.
\label{HO28}    
\end{equation}
For \(a(t)=\varepsilon\cos\Omega t\), define $\omega_{r\alpha}=\frac{\mathcal E_r-E_\alpha^*}{\hbar}$. Then
\begin{equation}
 c_r^{(1)}(t)
=\frac{\varepsilon\chi_r'(0)}{2\hbar\sqrt{N_\alpha}}
\left[
\frac{\e^{\ii(\omega_{r\alpha}+\Omega)t}-1}{\omega_{r\alpha}+\Omega}
+\frac{\e^{\ii(\omega_{r\alpha}-\Omega)t}-1}{\omega_{r\alpha}-\Omega}
\right].
\label{HO29}    
\end{equation}
The resonant denominators show the expected driving frequencies
\(\Omega\simeq |\omega_{r\alpha}|\).

We can write the expansion in powers of the small displacement as
\begin{equation}
\delta H(t)=a(t) \mathcal{W}^{(1)}+a^2(t) \mathcal{W}^{(2)}+\cdots \;,
\label{HO30}    
\end{equation}
where
\begin{equation}
\mathcal{W}_{r\alpha}^{(1)}=-\frac{\chi_r'(0)}{\sqrt{N_\alpha}},
\qquad
\mathcal{W}_{\beta\alpha}^{(1)}=0,
\qquad
\mathcal{W}_{rs}^{(1)}=0 \;.
\label{HO31}    
\end{equation}
The operator \(W^{(2)}\) is the second derivative coefficient of the moving-center Hamiltonian.  In the direct distributional notation it corresponds to
$\mathcal{W}^{(2)}=-\frac{\lambda}{2}\delta''(x)$, while in the Green-function notation it is obtained by the same second-order kernel expansion used in the moving-center section above.  For diagonal point-sector elements it must reproduce \eqref{HO15}:
\begin{equation}
\mathcal{W}_{\alpha\alpha}^{(2)}
=-\frac{\partial_{a}^2 G_0 (0,0|E_\alpha^*)}{2N_\alpha} \;.
\label{HO33}    
\end{equation}
By parity, \(\mathcal{W}^{(2)}\) is even.  Therefore it couples even states to even states and odd states to odd states, but not even states to odd states:
$\mathcal{W}_{r\alpha}^{(2)}=0$. Consequently, if the initial state is \(\psi_\alpha\), the second-order amplitude to an even point-sector state \(\psi_\beta\) has two contributions:
$c_\beta^{(2)}(t)=c_{\beta,{\rm direct}}^{(2)}(t)+c_{\beta,{\rm iter}}^{(2)}(t)$. The direct second-order contribution is
\begin{equation}
 c_{\beta,{\rm direct}}^{(2)}(t)
=-\frac{\ii}{\hbar} \mathcal{W}_{\beta\alpha}^{(2)}
\int_0^t a(t')^2
\exp\!\left[\frac{\ii}{\hbar}(E_\beta^*-E_\alpha^*)t'\right]\dd t' \;.
\label{HO36}    
\end{equation}
The iterated contribution through intermediate odd oscillator states is
\begin{eqnarray}
\begin{aligned}
 c_{\beta,{\rm iter}}^{(2)}(t)
=&-\frac{1}{\hbar^2}
\sum_{r=0}^{\infty}
\mathcal{W}_{\beta r}^{(1)} \mathcal{W}_{r\alpha}^{(1)}
\int_0^t\dd t_1\,a(t_1)
\e^{\frac{\ii}{\hbar}(E_\beta^*-\mathcal E_r)t_1}
\cr
&\hspace{3.0cm}\times
\int_0^{t_1}\dd t_2\,a(t_2)
\e^{\frac{\ii}{\hbar}(\mathcal E_r-E_\alpha^*)t_2} \;.
\end{aligned}
\label{HO37}
\end{eqnarray}
Using \eqref{HO31}, we get
\begin{equation}
\mathcal{W}_{\beta r}^{(1)} \mathcal{W}_{r\alpha}^{(1)}
=\frac{\chi_r'(0)^2}{\sqrt{N_\beta N_\alpha}} \;.
\label{HO38}    
\end{equation}
Hence
\begin{eqnarray}
    \begin{aligned}
 c_{\beta,{\rm iter}}^{(2)}(t)
=&-\frac{1}{\hbar^2\sqrt{N_\beta N_\alpha}}
\sum_{r=0}^{\infty}\chi_r'(0)^2
\int_0^t\dd t_1\,a(t_1)
\e^{\frac{\ii}{\hbar}(E_\beta^*-\mathcal E_r)t_1}
\cr
&\hspace{2.8cm}\times
\int_0^{t_1}\dd t_2\,a(t_2)
\e^{\frac{\ii}{\hbar}(\mathcal E_r-E_\alpha^*)t_2} \;.
\end{aligned}
\label{HO39}
\end{eqnarray}
This is the explicit second-order analogue of the manuscript's time-dependent perturbation formula.  It shows the selection rule very clearly:
\[
\text{even }\psi_\alpha \xrightarrow{\; \mathcal{W}^{(1)}\;} \text{odd }\chi_r
\xrightarrow{\;\mathcal{W}^{(1)}\;} \text{even }\psi_\beta.
\]
For the harmonic oscillation \(a(t)=\varepsilon\cos\Omega t\), both time integrals in \eqref{HO36}--\eqref{HO39} are elementary sums of exponentials.  For example,
$\int_0^t a(t')^2\e^{\ii\omega t'}\dd t'
=\varepsilon^2\int_0^t\cos^2(\Omega t')\e^{\ii\omega t'}\dd t'$ with \eqref{intcos2}. Thus the direct second-order term contains frequencies \(0\) and \(2\Omega\), whereas the iterated term contains two ordered insertions of the first-order frequency \(\Omega\).

\begin{remark}
This explicit oscillator example confirms the general moving-center formulas in a transparent setting:
\begin{enumerate}
\item the energy derivative and the coordinate derivative are distinct;
\item at the symmetric center, the first coordinate derivative of the principal function vanishes by parity;
\item therefore \(\Delta_\alpha^{(1)}=0\), while \(\Delta_\alpha^{(2)}\neq0\) in general;
\item the first-order transition is not zero, but it goes from even point-sector states to odd oscillator states;
\item second order returns to the even sector through either the direct even operator \(\mathcal{W}^{(2)}\) or two successive odd first-order insertions.
\end{enumerate}
This is the method developed in the manuscript: expand the moving support, identify the first- and second-order Hamiltonian pieces, and then insert them into the time-dependent perturbation equations.    
\end{remark}

\begin{remark}
The agreement with the formal one-dimensional expansion should not be interpreted as eliminating the domain issue. Expanding a moving delta distribution produces derivatives such as $\delta'$ and $\delta''$, and these terms must be understood through self-adjoint boundary conditions, or equivalently through the resolvent formulation. The one-dimensional example therefore provides a useful consistency check: where the usual formulation is available, the projection-based method reproduces the same first- and second-order transition structure, while avoiding the need to manipulate singular distributions as ordinary small potentials.
\end{remark}

\section{Uniformly moving delta center on a circle}
\label{sec:uniform-moving-delta-circle-example}

We finally compare the local moving-center expansion with an exactly solvable model which still has discrete spectrum.  Let the particle move on a circle of radius \(R\), and write the arclength coordinate as \(x\in[0,L)\), where \(L=2\pi R\).  We consider the attractive point interaction
\begin{equation}
    H_v(t)=-\frac{\hbar^2}{2m}\frac{d^2}{dx^2}-\lambda\delta_L(x-vt),
    \qquad \lambda>0,
    \label{eq:ring-moving-delta-H}
\end{equation}
with periodic boundary condition in \(x\).  Here \(\delta_L\) denotes the periodic delta distribution on the circle, \(\Omega\) is the constant angular
velocity of the moving center, and \(v=R\Omega\) is the constant tangential velocity of the center.  This model is useful because the time dependence can be removed exactly by passing to the rotating frame, while the spectrum remains discrete.

If we put \(\xi=x-vt\) modulo \(L\), and define $\chi(\xi,t)=\Psi(\xi+vt,t)$, then
\begin{equation}
    i\hbar\partial_t\chi
    =\left(\frac{p_\xi^2}{2m}-v p_\xi-\lambda\delta_L(\xi)\right)\chi,
    \qquad
    p_\xi=-i\hbar\frac{d}{d\xi}.
    \label{eq:ring-rotating-frame-H}
\end{equation}
Equivalently,
\begin{equation}
    \frac{p_\xi^2}{2m}-v p_\xi
    =\frac{(p_\xi-mv)^2}{2m}-\frac12mv^2 .
\end{equation}
The remaining first-order derivative term can be removed by the gauge transformation
\begin{equation}
    \chi(\xi,t)=e^{i\kappa_v\xi}\widetilde\chi(\xi,t),
    \qquad
    \kappa_v:=\frac{mv}{\hbar}.
\end{equation}
After this transformation, \(\widetilde\chi\) satisfies the time-independent
equation
\begin{equation}
    i\hbar\partial_t\widetilde\chi
    =
    \left(
    -\frac{\hbar^2}{2m}\frac{d^2}{d\xi^2}
    -\lambda\delta_L(\xi)
    -\frac12mv^2
    \right)\widetilde\chi .
    \label{eq:ring-gauged-H}
\end{equation}
The gauge transformation changes the periodicity condition into a twisted
one.  Since the point interaction is located at the identified point
\(\xi=0\equiv L\), the wave function is twisted-continuous, but its derivative
has the usual delta jump.  Thus
\begin{equation}
    \widetilde\chi(L^-,t)
    =
    e^{-i\vartheta_v}\widetilde\chi(0^+,t),
    \qquad
    \vartheta_v:=\kappa_vL=\frac{mvL}{\hbar},
    \label{eq:ring-twisted-continuity}
\end{equation}
and
\begin{equation}
    \partial_\xi\widetilde\chi(0^+,t)
    -
    e^{i\vartheta_v}\partial_\xi\widetilde\chi(L^-,t)
    =
    -\frac{2m\lambda}{\hbar^2}\widetilde\chi(0^+,t).
    \label{eq:ring-twisted-jump}
\end{equation}
Equivalently,
\begin{equation}
    \partial_\xi\widetilde\chi(L^-,t)
    =
    e^{-i\vartheta_v}
    \left[
    \partial_\xi\widetilde\chi(0^+,t)
    +
    \frac{2m\lambda}{\hbar^2}\widetilde\chi(0^+,t)
    \right].
\end{equation}
The derivative condition
\(\widetilde\chi'(L^-)=e^{-i\vartheta_v}\widetilde\chi'(0^+)\)
would be correct only in the absence of the delta interaction.

Thus a uniformly moving delta center on the circle is equivalent to a static
delta center on the circle in the presence of an Aharonov--Bohm-type twist,
with the usual delta matching condition imposed at the support point.  The
exact solutions are therefore obtained by solving the stationary problem on
\(0<\xi<L\),
\begin{equation}
    -\frac{\hbar^2}{2m}u''(\xi)=\mathcal E u(\xi),
    \qquad 0<\xi<L,
    \label{eq:ring-stationary-free}
\end{equation}
with the matching conditions
\begin{equation}
    u(L^-)=e^{-i\vartheta_v}u(0^+),
    \label{eq:ring-u-continuity}
\end{equation}
and
\begin{equation}
    u'(0^+)-e^{i\vartheta_v}u'(L^-)
    =
    -\frac{2m\lambda}{\hbar^2}u(0^+).
    \label{eq:ring-u-jump}
\end{equation}

If \(u_n\) solves this stationary twisted point-interaction problem with
energy \(\mathcal E_n\), then the corresponding time-dependent solution in
the gauged co-moving frame is
\begin{equation}
    \widetilde{\chi}_n(\xi,t)
    =
    \exp\left[
    -\frac{i}{\hbar}
    \left(\mathcal E_n-\frac12mv^2\right)t
    \right]u_n(\xi).
    \label{eq:ring-tildechi-solution}
\end{equation}
Equivalently, in the original co-moving frame,
\begin{equation}
    \chi_n(\xi,t)
    =
    e^{i\kappa_v\xi}
    \exp\left[
    -\frac{i}{\hbar}
    \left(\mathcal E_n-\frac12mv^2\right)t
    \right]u_n(\xi).
    \label{eq:ring-chi-solution}
\end{equation}
The solution in the original laboratory coordinate is finally recovered by
\begin{equation}
    \Psi_n(x,t)=\chi_n(x-vt,t).
    \label{eq:ring-lab-from-comoving}
\end{equation}
A general solution is obtained by superposing these modes.

For positive-energy branches, $\mathcal E=\frac{\hbar^2k^2}{2m}$, we write
$u(\xi)=Ae^{ik\xi}+Be^{-ik\xi}$ with $0<\xi<L$. Substituting this expression into the twisted continuity condition
\eqref{eq:ring-u-continuity} and the jump condition
\eqref{eq:ring-u-jump} gives a homogeneous linear system for \(A\) and \(B\).
The condition for a nontrivial solution is
\begin{equation}
    \cos\vartheta_v
    =
    \cos(kL)
    -
    \frac{m\lambda}{\hbar^2k}\sin(kL).
    \label{eq:ring-positive-spectrum}
\end{equation}
Thus the positive-energy spectrum is determined by
\eqref{eq:ring-positive-spectrum}.

For the attractive negative branch we put $\mathcal E=-\frac{\hbar^2\gamma^2}{2m}$ with $\gamma>0$.
Equivalently, this is obtained from the positive-energy equation by the
substitution \(k=i\gamma\).  Hence the bound-state spectral equation is
\begin{equation}
    \cos\vartheta_v
    =
    \cosh(\gamma L)
    -
    \frac{m\lambda}{\hbar^2\gamma}\sinh(\gamma L).
    \label{eq:ring-bound-spectrum}
\end{equation}
Equations \eqref{eq:ring-positive-spectrum} and
\eqref{eq:ring-bound-spectrum} are therefore unchanged, but they must be
understood as consequences of the twisted continuity condition together with
the delta jump condition, not of twisted derivative continuity.

If \(u_n\) is a normalized solution of the twisted static problem with energy
\(\mathcal E_n\), then the corresponding exact solution in the laboratory
frame is
\begin{equation}
    \Psi_n(x,t)
    =
    \exp\left[-\frac{i}{\hbar}
    \left(\mathcal E_n-\frac12mv^2\right)t\right]
    \exp\left[i\kappa_v(x-vt)\right]
    u_n(x-vt).
    \label{eq:ring-exact-lab-solution}
\end{equation}
The twist in \(u_n\) precisely compensates the phase gained by
\(\exp(i\kappa_v x)\), so that \(\Psi_n\) is periodic in the laboratory
coordinate \(x\).

This exact solution also clarifies the local expansion used in Section~\ref{sec:moving}.  If the velocity is ignored and one looks only at the frozen instantaneous Hamiltonian with center at \(q\), then
\begin{equation}
    H(q)=T(q)H(0)T(q)^*,
    \qquad
    T(q)=\exp\left(-\frac{i}{\hbar}q p\right).
\end{equation}
The eigenfunctions
are translated, but the eigenvalues are unchanged.  Equivalently, the free
Green function on the circle depends only on the difference of its arguments,
so its diagonal value \(G_0(q,q|E)\), and hence the corresponding principal
function, is independent of the support point.  Thus the first and second
derivatives of the frozen eigenvalues with respect to the center position
vanish, giving
\begin{equation}
    \Delta_n^{(1)}(t)=\Delta_n^{(2)}(t)=0
\end{equation}
The nontrivial effect of the motion is therefore not a shift of the frozen
eigenvalues, but the time dependence of the translated eigenfunctions, or
equivalently the velocity-dependent term appearing in the co-moving frame.

The entire local dependence is carried by the spectral projections,
\begin{equation}
    P_n(q)=T(q)P_n(0)T(q)^*.
\end{equation}
Expanding the exact translation formula gives
\begin{equation}
    P_n(q)=P_n(0)-\frac{iq}{\hbar}[p,P_n(0)]
    -\frac{q^2}{2\hbar^2}[p,[p,P_n(0)]]+O(q^3).
    \label{eq:ring-projector-expansion}
\end{equation}
This is precisely the structure of the moving-center expansion: symmetry removes the pole shift, while the first- and second-order effects come from the motion of the projector.

The small parameter controlling the local expansion is the displacement
measured in units of the length scale on which the relevant static
eigenfunction changes.  To make this statement quantitative, let \(u_n\) be
a normalized static eigenfunction and let \(T(q)\) denote the translation
operator $(T(q)u_n)(x)=u_n(x-q)$. For small \(q\),
\begin{equation}
    T(q)u_n
    =
    u_n-q\,\partial_xu_n+\frac{q^2}{2}\partial_x^2u_n+O(q^3).
\end{equation}
The overlap between the translated state and the original one is therefore
\begin{equation}
    \langle u_n,T(q)u_n\rangle
    =
    1-q\langle u_n,\partial_xu_n\rangle
    +\frac{q^2}{2}\langle u_n,\partial_x^2u_n\rangle
    +O(q^3).
\end{equation}
Since \(u_n\) is normalized, \(\langle u_n,\partial_xu_n\rangle\) is purely
imaginary.  Moreover, after integrating by parts, assuming periodic boundary
conditions on the circle,
$\operatorname{Re}\langle u_n,\partial_x^2u_n\rangle
    =
    -\|\partial_xu_n\|^2$. Hence
\begin{equation}
    |\langle u_n,T(q)u_n\rangle|^2
    =
    1
    -
    q^2
    \left(
    \|\partial_xu_n\|^2
    -
    |\langle u_n,\partial_xu_n\rangle|^2
    \right)
    +O(q^3).
\end{equation}
This motivates the definition $\ell_n^{-2}
    :=
    \|\partial_xu_n\|^2
    -
    |\langle u_n,\partial_xu_n\rangle|^2$. Thus
$|\langle u_n,T(q)u_n\rangle|^2
    =
    1-q^2\ell_n^{-2}+O(q^3)$. The subtraction term
\(|\langle u_n,\partial_xu_n\rangle|^2\) removes the component of
\(\partial_xu_n\) parallel to \(u_n\), which corresponds only to a change of
phase.  Therefore \(\ell_n\) measures the length scale over which the
translated state changes in a physically relevant way. 
For instance, a pure plane wave changes under translation only by a phase, so
\(\|\partial_xu_n\|^2=|\langle u_n,\partial_xu_n\rangle|^2\) and
\(\ell_n^{-2}=0\); the overlap modulus remains exactly one.  For a localized
or nontrivial static eigenfunction, \(\ell_n\) is finite and gives the
natural displacement scale for the local expansion.

Thus the second-order local approximation is reliable on time intervals for
which the displacement remains small compared with both the intrinsic
variation length of the relevant eigenfunction and the circumference of the
circle,
$|q(t)|=|vt|\ll \ell_n$, and $|q(t)|\ll L$.
For the full time-dependent dynamics there is, in addition, an independent
smallness condition associated with the velocity.  In the co-moving frame the
Hamiltonian contains the term \(-vp\), or equivalently the twisted boundary
condition involves $\vartheta_v=\frac{mvL}{\hbar}$. Thus the velocity-induced coupling should be small compared with the relevant
spectral gaps.  A convenient sufficient condition is
\begin{equation}
    \frac{|v|\,
    \sqrt{\langle u_n,(p-\langle p\rangle_n)^2u_n\rangle}}
    {\Delta_n}
    \ll 1,
    \qquad
    \Delta_n:=\min_{m\neq n}|E_m-E_n|,
\end{equation}
where \(p=-i\hbar\partial_x\) and
\(\langle p\rangle_n=\langle u_n,pu_n\rangle\).  Since
$\sqrt{\langle u_n,(p-\langle p\rangle_n)^2u_n\rangle}
    =
    \hbar\,\ell_n^{-1}$, this condition may also be written as
$\frac{\hbar |v|}{\ell_n\,\Delta_n}\ll 1$.

Therefore the local expansion is controlled by the displacement \(q(t)\),
while the accuracy of the dynamical approximation is also controlled by the
velocity-induced term \(-vp\).

For the attractive bound branch at zero velocity, \(\ell_n\) is of the order of the bound-state localization length, determined by the solution \(\gamma_0\) of
\begin{equation}
    \gamma_0\tanh\left(\frac{\gamma_0L}{2}\right)=\frac{m\lambda}{\hbar^2}.
\end{equation}

The exact rotating-frame solution displays an additional dynamical parameter which is not part of the frozen instantaneous pole shift.  It is the dimensionless twist
\begin{equation}
    \vartheta_v=\frac{mvL}{\hbar}.
\end{equation}
For small velocity, \(\cos\vartheta_v=1-\vartheta_v^2/2+O(\vartheta_v^4)\), so the exact rotating-frame eigenvalues differ from the static ones only at order \(v^2\).  For example, if \(\gamma(v)=\gamma_0+O(v^2)\) denotes the solution of \eqref{eq:ring-bound-spectrum}, then
\begin{equation}
    \gamma(v)=\gamma_0-
    \frac{\vartheta_v^2}{2F_\gamma(\gamma_0)}+O(\vartheta_v^4),
    \label{eq:ring-gamma-small-v}
\end{equation}
where
\begin{equation}
    F_\gamma(\gamma_0)=
    L\sinh(\gamma_0L)
    -\frac{m\lambda}{\hbar^2}
    \frac{\gamma_0L\cosh(\gamma_0L)-\sinh(\gamma_0L)}{\gamma_0^2} .
\end{equation}
The laboratory-frame phase in \eqref{eq:ring-exact-lab-solution} also contains the inertial contribution \(-mv^2/2\).  These terms are not contradictions of the instantaneous result \(\Delta_n^{(1)}=\Delta_n^{(2)}=0\); rather, they are dynamical effects associated with working in the rotating frame.

This example provides a discrete spectrum version of the moving delta problem, modified to compare with the approximation used in our formalism.  The local Taylor expansion in the support position is controlled by \(|q(t)|/\ell_n\) and by the requirement that the center remains in a local coordinate patch. The complete non-autonomous dynamics also depends on the independent velocity parameter \(\vartheta_v=mvL/\hbar\), or equivalently on the size of the Coriolis term \(-vp\) compared with the relevant spectral gaps.

\section{Conclusions}
\label{sec:conclusion}

We have formulated time-dependent perturbation theory for renormalized point interactions in settings where the unperturbed Hamiltonian has purely discrete spectrum.  The key point is that the singular Hamiltonian should not be perturbed at the formal level.  Instead, the perturbation is introduced through the renormalized data: the principal function, the shifted poles of the resolvent, and the associated rank-one projections.

We have also clarified the role of finite degeneracies.  A single point interaction does not couple to an entire degenerate eigenspace, but only to the evaluation direction represented by the diagonal spectral projection at the support point.  The orthogonal subspace, consisting of eigenfunctions that vanish at the support, remains unaffected for a fixed center.  When the center moves, these adapted subspaces vary and produce the selection rules seen explicitly in the sphere and isotropic-oscillator examples.

For a time-dependent renormalized strength, the pole shifts are determined by the scalar equation $\Phi(E_k(t),\mu(t))=0$.  The first two corrections are given by equations \eqref{eq:Delta1-mu} and \eqref{eq:Delta2-mu}.  The induced Hamiltonian perturbation is not merely the diagonal energy shift; it also contains derivatives of the spectral projections.  These projection derivatives vanish in the diagonal expectation value at first order but generate the off-diagonal transition amplitudes.

For a moving point interaction, the same structure persists with additional derivatives with respect to the support point.  The expansion naturally separates the variation of the pole from the variation of the eigenprojection.  In homogeneous situations, such as the two-sphere, the principal function is position independent and the pole motion disappears; the physical effect is then entirely geometric and is encoded in the moving projections.  At symmetry points of the harmonic oscillator, the first-order pole shift similarly vanishes, while transition amplitudes may still be nonzero because the eigenfunctions change under displacement of the impurity.

The one-dimensional harmonic oscillator illustrates the non-renormalized version of the same mechanism.  There the principal function is simply $\Phi=\lambda^{-1}-G_0(a,a|E)$, all diagonal Green-function derivatives are finite, and the coupling $\lambda(t)$ may be varied directly.  For a moving center at the parity-symmetric point $a=0$, the first energy shift vanishes but the shifted even point-interaction state couples at first order to odd oscillator modes.  This provides a useful consistency check on the renormalized two-dimensional calculation.

We also included an exactly solvable discrete-spectrum example: a point interaction moving uniformly on a circle.  In the rotating frame this problem becomes a time-independent delta interaction with a twisted boundary condition.  The exact solution shows that the frozen instantaneous eigenvalues are independent of the support point, so the local expansion is controlled by the translated spectral projections.  The reliability of the second-order approximation is governed by the displacement relative to the natural variation length of the static eigenfunctions, while the full dynamics also contains the independent velocity, or twist, parameter $mvL/\hbar$.

The formalism can be extended in several directions.  One may treat several moving centers by replacing the scalar principal function with the finite-dimensional principal matrix, include avoided crossings by degenerate perturbation theory, or compare the renormalized point interaction with regularized short-range potentials such as narrow Gaussians.  These extensions are natural because the present formulation is expressed entirely in terms of resolvents, heat kernels, and spectral projections.

\section*{Declaration of generative AI and AI-assisted technologies in the manuscript preparation process}

During the preparation of this work, the authors used OpenAI's ChatGPT
as an auxiliary tool for language editing, improving readability, organizing
the presentation, and checking the exposition of some examples. The authors
independently reviewed, verified, and edited all AI-assisted content, and
take full responsibility for the content of the publication.

\section*{CRediT authorship contribution statement}

Kaya G\"uven Akba\c{s}: Conceptualization, Formal analysis, Methodology,
Validation, Writing -- original draft, Writing -- review and editing.

Fatih Erman: Conceptualization, Formal analysis, Methodology, Validation,
Writing -- original draft, Writing -- review and editing.

O.~Teoman Turgut: Conceptualization, Formal analysis, Methodology,
Validation, Writing -- original draft, Writing -- review and editing.

\section*{Acknowledgements}
O.~T.~Turgut and F.~Erman gratefully acknowledge A.~Mostafazadeh for valuable discussions. The authors dedicate this work to Prof.~M.~Znojil on the occasion of his 80th birthday, with their best wishes for many more productive years. Last, but not least, the authors thank P.~Exner for the inspiration that led to this work. 
The authors are grateful to the anonymous reviewer for the careful reading of
the manuscript and for the constructive suggestions.  The
reviewer's comments helped us improve the clarity, accuracy, and overall
presentation of the paper.

\end{document}